\documentclass[12pt,oneside,english]{amsart}

\usepackage{lmodern}

\usepackage[T1]{fontenc}
\usepackage[utf8]{inputenc}

\usepackage[verbose,tmargin=1.225in,bmargin=1.225in,lmargin=1.225in,rmargin=1.225in]{geometry}
\usepackage[english]{babel}

\usepackage{amsmath}
\usepackage{amstext}
\usepackage{amsthm}
\usepackage{amssymb}
\usepackage{amsfonts}
\usepackage{mathrsfs}
\usepackage{stmaryrd}
\usepackage{esint}
\usepackage{calc}

\usepackage{graphicx}
\usepackage{xcolor}
\usepackage{color}
\usepackage{subfig}

\usepackage{array}
\usepackage{multirow}
\usepackage{rotating}
\usepackage{lscape}
\usepackage{multicol}
\providecommand{\tabularnewline}{\\}

\usepackage[authoryear]{natbib}
\usepackage{setspace}
\usepackage{csquotes}
\usepackage{float}
\usepackage{caption}
\usepackage{multifootnote}
\usepackage{fixltx2e}
\usepackage[normalem]{ulem}

\usepackage[unicode=true,pdfusetitle,
 bookmarks=true,bookmarksnumbered=false,bookmarksopen=false,
 breaklinks=false,pdfborder={0 0 0},pdfborderstyle={},backref=false,colorlinks=true]
 {hyperref}
\hypersetup{
 citecolor=blue, linkcolor=blue}

\renewcommand{\qed}{{\hfill\itshape Q.E.D.}}

\providecommand{\remarkname}{Remark}
\providecommand{\theoremname}{Theorem}
\providecommand{\corollaryname}{Corollary}
\providecommand{\lemmaname}{Lemma}

\theoremstyle{remark}
\newtheorem{rem}{\remarkname}
\theoremstyle{plain}
\newtheorem{thm}{\theoremname}
\newtheorem{cor}{\corollaryname}
\theoremstyle{definition}
\newtheorem{assumption}{Assumption}
\theoremstyle{plain}
\newtheorem{lem}{\lemmaname}

\numberwithin{equation}{section}
\numberwithin{figure}{section}

\floatstyle{ruled}

\thanks{\noindent $^\dag$ Department of Economics, Johns Hopkins University, and Nuffield College, 
rspady@jhu.edu.}
\thanks{ $\S$ School of Economics, University of Bristol, and Department of Economics, University of Melbourne, 
s.stouli@bristol.ac.uk.}
\thanks{We are grateful to Whitney Newey for his encouragements and useful comments, and to Clément Imbert and Rosa Matzkin for helpful discussions. We also thank seminar participants at Bristol, UC San Diego, 
Oxford, Lehigh, LSE, Johns Hopkins, Laval, Carlos III, Manchester, Bonn, TSE, CREST, Queensland, UPenn, Monash and at multiple conferences, including the 2020 Econometric Society World Congress. We thank Diego Lara de Andrés, Xiaoran Liang, and Thomas Stringham for excellent research assistance. 
Stouli acknowledges support from ESRC through a New Investigator Grant (ES/S012362/1), 
and from UKRI under the UK government’s Horizon Europe funding guarantee (EP/Y004159/1).
}

\date{\today}

\begin{document}
\title{Gaussian Transforms Modeling and the Estimation of Distributional
Regression Functions}
\author{Richard H. Spady$^\dag$ and Sami Stouli$^\S$}
\begin{abstract}
We propose flexible Gaussian representations for conditional cumulative
distribution functions and give a concave likelihood criterion for
their estimation. Optimal representations satisfy the monotonicity
property of conditional cumulative distribution functions, including
in finite samples and under general misspecification. We use these
representations to provide a unified framework for the flexible Maximum
Likelihood estimation of conditional density, cumulative distribution,
and quantile functions at parametric rate. Our formulation yields
substantial simplifications and finite sample improvements over related
methods. An empirical application to the gender wage gap in the United
States illustrates our framework.
\end{abstract}

\maketitle
\textsc{\small{}Keywords:}{\small{} Conditional density estimation,
conditional distributions, conditional quantiles, maximum likelihood,
misspecification, monotonicity, convexity, gender wage gap.}{\small\par}

\section{Introduction}

The flexible modeling and estimation of conditional distributions
are important for the analysis of various econometric and statistical
problems. Conditional probability density functions (PDF) are used
in the program evaluation literature where the generalized propensity
score takes the form of a conditional density (e.g., \citet{Imbens:2000}).
Conditional cumulative distribution functions (CDF) are core building
blocks in the estimation of nonseparable models with endogeneity where
the control variable takes the form of a conditional CDF (e.g., \citet{Imbens Newey 2009},
\citet{CFNSV}). Both conditional PDFs and CDFs are key components
in counterfactual analysis (e.g., \citet{DiNardoetal}, \citet{Chernozhukov Fern Melly 2013}).

For a continuous outcome variable $Y$ and a vector of explanatory
variables $X$, difficulties arise in the formulation of a flexible
model and in the choice of a loss function for the estimation of the
conditional PDF or CDF of $Y$ given $X$. First, a flexible class
of models may include elements that do not satisfy the defining properties
of PDFs (being positive and integration to one) and CDFs (monotonicity
and range in the unit interval) at each value of $X$. This allows
the chosen loss function to select a model that does not satisfy these
properties in finite samples and/or under misspecification.\footnote{These problems have motivated the development of a variety of post-estimation
methods to correct initial estimators (e.g., \citet{HM:2003}, \citet{GHI:2003},
\citet{Chernozhukov Fern Gali 2010}, \citet{HLL:2020}).} Second, maximum likelihood (ML) formulations are often difficult
to implement because nonconcave and/or unbounded likelihoods naturally
arise in the context of flexible modeling, a fact that has motivated
the development of alternatives to ML.\footnote{\citet[Chapter 1.7.4]{Vapnik:1999} motivates the development of alternatives
to ML by the ``narrowness of the ML method'', which is illustrated
by the unbounded likelihood arising for Gaussian mixture models. The
robust approach of \citet{Huber:1981} provides a convex alternative
to the use of Gaussian likelihoods for the concomitant estimation
of location and scale parameters (cf. \citet{Owen:2007} and \citet{SS:2018b}
for a discussion). Section \ref{subsec:Discussion} further elaborates
on this issue.} Third, alternative nonparametric approaches such as kernel regression
suffer from the curse of dimensionality, a limitation that restricts
their use in practice. 

One approach to address these difficulties is to specify a flexible
class of models that satisfy the defining properties of PDFs and CDFs
by construction. Sieve ML approaches specify conditional PDFs as scaled
positive transformations of linear combinations of known functions
of $Y$ and $X$, and use the implied log-likelihood function for
estimation (e.g., \citet{Stone:1994}). Support vector methods specify
conditional PDFs as weighted averages of nonnegative kernel functions
of $Y$ and $X$, with weights restricted to be nonnegative and selected
so that the implied representation is as close as possible to the
empirical data distribution (\citet{Vapnik:1999}). Another approach
is to focus on flexible conditional CDF modeling, while discarding
the monotonicity requirement. Distribution regression (\citet{ForesiPerrachi}, \citet{Chernozhukov Fern Melly 2013})
specifies each level of the CDF of $Y$ given $X$ as a known CDF
transformation of a linear combination of the components of $X$.
The conditional CDF is then estimated at each $Y$ value by a sequence
of binary outcome ML estimators.

In this paper we take a different approach by formulating a flexible
class of Gaussian representations for conditional CDFs, instead of
modeling conditional PDFs or CDFs directly. We first expand the range
of the conditional CDF to the real line by application of a Gaussian
quantile transform, and then specify the resulting object as a linear
combination of known functions of $Y$ and $X$. The coefficients
in this linear combination are characterized by a globally concave
likelihood-based loss function that discards nonmonotone models within
the specified class, and selects the optimal element according to
the Kullback-Leibler Information Criterion (KLIC; \citet{Akaike:1973},
\citet{White:1982}). Implied models for conditional PDFs and CDFs
are flexible and can be estimated efficiently by ML under correct
specification. We provide estimation and inference results allowing
for misspecification, and we derive a dual formulation for our estimator
that we use for implementation.

We make four main contributions to the existing literature. First,
we reformulate the problems of modeling and estimating conditional
PDFs and CDFs in terms of conditional CDF representations with range
the real line. This enables us to specify a flexible class of linear
models for these representations. 
Compared to sieve ML formulations,
conditional PDF models implied by our representations preserve flexibility
while avoiding the need for a scaling factor that in general cannot
be calculated in closed-form. 
Compared to support vector methods that restrict
model coefficients to be nonnegative, our ML approach avoids discarding
potentially more accurate valid approximations (\citet{WGSVVW:1999}).

Second, we give an information-theoretic criterion for the selection
of a globally monotone model within the class of Gaussian representations
in linear form. Under general misspecification, this formulation characterizes
quasi-Gaussian representations that correspond to well-defined KLIC
optimal conditional PDF and CDF approximations to the true data probability
distribution. These approximations satisfy the defining properties
of PDFs and CDFs at each value of $X$, both in finite samples with
probability approaching one and in the population. Compared to the
distribution regression criterion, our approach allows for the global
approximation of conditional CDFs and affords KLIC optimality for
the selected model. 
Compared to sieve ML formulations that specify positive conditional
PDFs, our approach relies on the loss function to rule out conditional
PDFs that are negative or zero with positive probability.

Third, we provide a unified approach to the flexible ML estimation
of conditional PDFs and CDFs at parametric rate. Flexibility is obtained
without the inclusion of infinite-dimensional parameters, in this
way going beyond distribution and quantile regression (\citet{Koenker:Bassett1978}),
as well as parametric location-scale formulations (\citet{He:1997},
\citet[Section 2]{SS:2018a}, \citet{MachadoSantosSilva:2019}). Function-valued
parameters lead to slower than root-$n$ conditional PDF estimation
(\citet{RotheWied:2020}), and location-scale models restrict the
shape of conditional distributions. Our approach alleviates the curse
of dimensionality and hence also provides an alternative to kernel-based
methods for the nonparametric estimation of conditional distributions
(e.g., \citet{LR:2007}). Our approach further applies to conditional
quantile function (CQF) estimation and is thus also related to the
kernel-based estimator of \citet{Matzkin:2003} for nonseparable models.
In contrast, we propose a parametric formulation for their flexible
ML estimation, which also facilitates the imposition of shape constraints
from economic theory.

Fourth, we derive a dual formulation for our ML estimator that demonstrates
considerable computational benefits. Compared to distribution and
quantile regression, the dual formulation provides a convex programming
problem (\citet{BV:2004}) for the one-step estimation of conditional
PDFs and CDFs at each sample points. Compared to dual regression and
its generalization (\citet{SS:2018a}), we find that  
the dual formulation has the important advantage of being a mathematical
programming problem with linear constraints.

Taken together, these contributions define a novel, unified framework
for flexible distributional regression analysis. Our framework yields
substantial simplifications over related methods, while providing
all the theoretical guarantees afforded by ML. In numerical simulations,
we find that these features translate into largely improved finite
sample performance compared to kernel methods, the Matzkin estimator,
distribution and quantile regression. A distributional analysis of
the gender wage gap in the United States illustrates the benefits
of our methods for empirical practice. 
Our methods extend to a variety of settings, including Logistic transform regression models, 
mixed discrete-continuous outcome distributions, and multiple outcomes.

Section \ref{sec:Section2} introduces our modeling framework.
Section \ref{sec:Section3} gives results under misspecification.
Section \ref{sec:Section4} contains estimation and inference results,
and duality theory is derived in Section \ref{sec:Section5}. Section
\ref{sec:Section6} illustrates our methods, Section \ref{sec:Section7}
gives extensions and Section \ref{sec:Section8} concludes. Proofs
of Theorems \ref{thm:Thm1}-\ref{thm:Thm2} and \ref{thm:Thm6} are
given in the Appendix. The Supplemental Material (\citet{SS:2025}) contains proofs
of Theorems \ref{thm:Thm3}-\ref{thm:Thm5}, technical results,
implementation details, and results of numerical simulations.

\section{Gaussian Transforms Modeling\label{sec:Section2}}

Let $Y$ be a continuous outcome variable and $X$ a vector of explanatory
variables. A transformation to Gaussianity of the conditional CDF
$F_{Y|X}(Y|X)$ of $Y$ given $X$ occurs by application of the Gaussian
quantile function $\Phi^{-1}$,
\begin{equation}
e=\Phi^{-1}(F_{Y\mid X}(Y\mid X))\equiv g(Y,X),\label{eq:eyrep}
\end{equation}
where the resulting Gaussian Transform (GT) $e$ is a zero mean and
unit variance Gaussian random variable and is independent from $X$,
by construction. With $y\mapsto F_{Y\mid X}(y|X)$ strictly increasing,
the corresponding map $y\mapsto g(y,X)$ is also strictly increasing,
with well-defined inverse denoted $e\mapsto g^{-1}(e,X)$. 

Important statistical objects such as conditional PDFs, CDFs and CQFs
can be expressed as known functionals of $g(Y,X)$. The conditional
PDF of $Y$ given $X$ is
\[
f_{Y\mid X}(Y\mid X)=\phi(g(Y,X))\{\partial_{y}g(Y,X)\},\quad\partial_{y}g(Y,X)\equiv\frac{\partial g(Y,X)}{\partial y},
\]
where $e\mapsto\phi(e)$ is the Gaussian PDF and $\partial_{y}g(y,x)$
is a partial derivative, and the conditional CDF and CQF of $Y$ given
$X$ are
\[
F_{Y\mid X}(Y\mid X)=\Phi(g(Y,X)),\quad Q_{Y\mid X}(u\mid X)=g^{-1}(\Phi^{-1}(u),X),\quad u\in(0,1),
\]
respectively. The GT $g(Y,X)$ thus constitutes a natural modeling
object in the context of distributional regression models for $f_{Y|X}(Y|X)$,
$F_{Y|X}(Y|X)$, and $Q_{Y|X}(u|X)$. We refer to these objects as
the `Distributional Regression Functions' (DRF).

In this paper we consider the class of conditional CDFs with Gaussian
representation $e=g(Y,X)$ in linear form, where $g(Y,X)$ is specified
as a linear combination of known transformations of $Y$ and $X$.
Expanding the range of conditional CDFs from the unit interval to
the real line allows for the formulation of linear yet flexible representations.
The implied models for DRFs are parsimonious
and able to capture complex features of the entire statistical relationship
between $Y$ and $X$. In particular, these models allow for nonlinearity
and nonseparability of this relationship. 

\subsection{Gaussian representations in linear form}

Let $W(X)$ be a $K\times1$ vector of known functions of $X$ and
$S(Y)$ a $J\times1$ vector of known functions of $Y$, and write
$\otimes$ for their Kronecker product. Assume that $W(X)$ includes
an intercept, i.e., has first component $1$, and that $S(Y)$ has
first two components $(1,Y)'$ and derivative $dS(Y)/dy=s(Y)$, a
vector of functions continuous on $\mathbb{R}$. Given a random vector
$(Y,X')'$ with support $\mathcal{YX}=\mathcal{Y}\times\mathcal{X}$,
where $\mathcal{Y}=\mathbb{R}$ and $\mathcal{Y}$ and $\mathcal{X}$
are the marginal supports of $Y$ and $X$, respectively, our first
assumption defines a GT regression model.

\begin{assumption}For some $b_{0}\in\mathbb{R}^{JK}$, a GT regression
model takes the form
\begin{equation}
e=b_{0}'T(X,Y),\quad\partial_{y}\{b_{0}'T(X,Y)\}=b_{0}'t(X,Y)>0,\quad e\mid X\sim N(0,1),\label{eq:e}
\end{equation}
where $T(X,Y)\equiv W(X)\otimes S(Y)$ and $t(X,Y)\equiv W(X)\otimes s(Y)$.\label{ass:Ass1}\end{assumption}

The GT $g(Y,X)$ in (\ref{eq:eyrep}) is specified as a linear combination
of the known functions $T(X,Y)$. The linear form of $e$ is preserved
by the derivative function $b_{0}'t(X,Y)$ which is simultaneously
specified as a linear combination of $t(X,Y)$. For bounded $W(X)$
and $s(Y)$ with good approximation properties such as splines or wavelets, 
this linear specification is a flexible model for (\ref{eq:eyrep}).\footnote{If some component of $W(X)$ or $s(Y)$ has range the real line over 
$\mathcal{X}$ or $\mathcal{Y}$, respectively, then its coefficient 
in (2.2) must be zero since there is no $b_{0}\neq0$ such that $b_{0}'t(X,Y)>0$ 
in that case. For vector $X$, an example of a flexible structure 
takes $W(X)=W_{1}(X_{1})\varotimes\cdots\varotimes W_{\dim(X)}(X_{\dim(X)})$, 
where $W_{l}(X_{l})$ includes an intercept and a vector of bounded 
approximating functions, $l=1,\ldots,\dim(X)$.}$^{,}$\footnote{
When $b_0'T(X,Y)$ is viewed as an approximation, condition $b_0't(X,Y) > 0$ 
in (\ref{eq:e}) restricts the set of admissible approximating functions, 
without compromising approximation quality under regularity conditions and for $W(X)$ rich enough. 
For $y \mapsto g(y,x)$ uniformly smooth in $x$ and with $S(Y)$ 
chosen from a suitable class of approximating functions, 
there is $\beta(X)$ with $\beta(X)'s(Y)>0$ and uniform approximation error 
of $\beta(X)'S(Y)$ for $g(Y,X)$ of the same order as that of an unconstrained approximation. 
This is because the approximation errors for $y\mapsto g(y,X)$ of 
each approximation type are of the same order for each $x$ 
(e.g., \citet{DeVore:1977} for splines and \citet{AY:1992} for wavelets), 
and hence uniformly over $x$ if $y\mapsto g(y,x)$ is smooth uniformly in $x$.
Moreover, if $W(X)$ is mean-square spanning  
and $S(Y)$ has finite conditional second moment, 
then there is $b$ such that $E[\{ b'T(X,Y) - \beta(X)'S(Y) \}^2] \rightarrow 0$ as $K \rightarrow \infty$. 
Therefore,  for $K$ large enough, $b't(X,Y)>0$, by $\beta(X)'s(Y)>0$, and the orders of constrained and 
unconstrained approximations being the same 
now implies that taking $b'T(X,Y)$ with $b't(X,Y)>0$  in (\ref{eq:e}) preserves approximation quality.
}
When the nonconstant components of $W(X)$ and $s(Y)$ are specified 
as nonnegative spline functions (\citet{CS:1966}, \citet{Ramsay:1988}), 
we refer to the implied representations as `Spline-Spline models'. 
We note that we do not impose $b_{0}>0$.

Model (\ref{eq:e}) corresponds to a well-defined probability distribution
for $Y$ given $X$ with GT in linear form. The implied forms of the
conditional PDF and CDF are
\begin{equation}
f_{Y\mid X}(Y\mid X)=\phi(b_{0}'T(X,Y))\{b_{0}'t(X,Y)\},\quad F_{Y\mid X}(Y\mid X)=\Phi(b_{0}'T(X,Y)).\label{eq:DRF1}
\end{equation}
The resulting log conditional density forms the basis of our approach:
\begin{equation}
\log f_{Y\mid X}(Y\mid X)=-\frac{1}{2}[\log(2\pi)+\{b_{0}'T(X,Y)\}^{2}]+\log(b_{0}'t(X,Y)).\label{eq:log PDF}
\end{equation}
Viewed as an equation in the GT $b_{0}'T(X,Y)$,  (\ref{eq:log PDF})  
defines
the problem of characterizing GTs of the form (\ref{eq:e}), ruling
out implied conditional PDFs that are negative or zero with positive
probability, and hence also nonmonotone conditional CDFs. 

Formulation (\ref{eq:log PDF})  
differs from sieve ML approaches
that specify conditional PDFs as $f_{Y|X}(Y|X)=\varphi(b_{0}'T(X,Y))/\int\varphi(b_{0}'T(X,y))dy$
with $\varphi(\cdot)$ a nonnegative function such as the exponential
or the square functions. The log conditional PDF is
\begin{equation}
\log f_{Y\mid X}(Y\mid X)=\log\varphi(b_{0}'T(X,Y))-\log\int\varphi(b_{0}'T(X,y))dy,\label{eq:sieve ML}
\end{equation}
with scaling factor $\int\varphi(b_{0}'T(X,y))dy$ not available in
closed-form in general, and with implied ML first-order conditions
that are nonlinear in $b_{0}'T(X,Y)$. In contrast, (\ref{eq:log PDF})
avoids the scaling factor and yields first-order conditions linear
in $b_{0}'T(X,Y)$ (cf. (\ref{eq:FOCs}) and (\ref{eq:MMrep}) below). 
Formulation (\ref{eq:log PDF}) also 
differs from support vector methods that define the problem
of characterizing conditional PDFs as solving the equation
\begin{equation}
F_{YX}(y,x)=\int_{-\infty}^{x}\int_{-\infty}^{y}f(s\mid t)dF_{X}(t)ds\label{eq:svm int eq}
\end{equation}
for $f$, where $f$ is specified as a linear combination of nonnegative
kernel functions of $Y$ and $X$ with nonnegative coefficients, which
is overly restrictive  (cf. Section 5 in \citet{WGSVVW:1999}). 
Compared to both (\ref{eq:sieve ML}) and (\ref{eq:svm int eq}),
shape constraints from economic theory are also easier to impose using
(\ref{eq:log PDF}), with GT in closed-form.\footnote{Shape constraints often apply to CDFs or CQFs (e.g., \citet{BKM:2014}, 
\citet{CW:2017}), and they easily translate into restrictions on 
the shape of GTs. For example, under Assumption \ref{ass:Ass1}, a 
nonseparable demand model $Y=h(X,e)$ is nonincreasing in prices 
$X$ if the linear constraints $\partial_{x}t(X,Y)'b_{0}\geq0$ hold 
in (\ref{eq:log PDF}). By $\partial_{x}Q_{Y|X}(u|X)=-\left.\partial_{x}F_{Y|X}(y_{0}|X)/\partial_{y}F_{Y|X}(y_{0}|X)\right|_{y_{0}=Q_{Y|X}(u|X)}$ 
and $\partial_{y}F_{Y|X}(Y|X)>0$, nonincreasing demand is implied 
by $\partial_{x}F_{Y|X}(Y|X)=\phi(b_{0}'T(X,Y))\{\partial_{x}t(X,Y)'b_{0}\}\geq0$, 
and hence by $\partial_{x}t(X,Y)'b_{0}\geq0$.}
\begin{rem}
For $W(X)=1$, models for marginal PDF and CDF of $Y$ arise as a
particular case of (\ref{eq:e}). 
For $J=2$, the Jacobian term $b_{0}'t(X,Y)$ does not depend on $Y$ which restricts $F_{Y\mid X}(Y|X)$
to Gaussianity at all values of $X$.
\begin{rem}
Our modeling framework also applies when $\mathcal{Y}$ is bounded
since $Y$ can always be monotonically transformed to a random variable
with support expanded over the real line (cf. Remark \ref{Rk:Boundedsupp} in Section
\ref{sec:Implementation} of the Supplemental Material).

\end{rem}
\end{rem}

\subsection{Characterization}

For $\Theta=\{b\in\mathbb{R}^{JK}:\Pr[b't(X,Y)>0]=1\}$, we define
the population objective function
\begin{equation}
Q(b)=E\left[-\frac{1}{2}\left(\log(2\pi)+\{b'T(X,Y)\}^{2}\right)+\log\left(b't(X,Y)\right)\right],\quad b\in\Theta.\label{eq:popML}
\end{equation}
This criterion introduces a natural logarithmic barrier function (e.g.,
\citet{BV:2004}) in the form of the log of the Jacobian term $b't(X,Y)$.
This is important because the monotonicity requirement for the conditional
CDF is imposed directly by the objective in the definition of the
effective domain of $Q(b)$, i.e., the region in $\mathbb{R}^{JK}$
where $Q(b)>-\infty$. An equivalent interpretation is that the effective
domain of $Q(b)$ contains the set of parameter values that are admissible
for GT regression models with positive conditional PDF, by virtue
of the presence of both the Gaussian density function and the logarithmic
barrier function in (\ref{eq:popML}).

We characterize the shape and properties of $Q(b)$ under the following
assumption.

\begin{assumption}$E[||T(X,Y)||^{2}]<\infty$, $E[||t(X,Y)||^{2}]<\infty$,
and the smallest eigenvalue of $E[T(X,Y)T(X,Y)']$ is bounded away
from zero.\label{ass: Ass2}\end{assumption}

These conditions restrict the set of dictionaries we allow for, as
well as the probability distribution of $Y$ conditional on $X$.
In particular, because $T(X,Y)$ includes $Y$, Assumption \ref{ass: Ass2}
requires $Y$ to have finite second moment. The moment conditions
in Assumption \ref{ass: Ass2} are also sufficient for the second-derivative
matrix of $Q(b)$,
\begin{equation}
\Gamma(b)\equiv E\left[\gamma(Y,X,b)\right],\quad\gamma(Y,X,b)\equiv-T(X,Y)T(X,Y)'-\frac{t(X,Y)t(X,Y)'}{\{b't(X,Y)\}^{2}},\label{eq:hessian}
\end{equation}
to exist for each $b\in\Theta$. Nonsingularity of $E[T(X,Y)T(X,Y)']$
guarantees that $\Gamma(b)$ is negative definite, and hence that
$Q(b)$ is strictly concave and has a unique maximum.
\begin{thm}
\label{thm:Thm1}If Assumptions \ref{ass:Ass1}-\ref{ass: Ass2} hold
then $Q(b)$ is strictly concave and has a unique maximum in $\Theta$
at $b_{0}$.
\end{thm}
By standard ML theory (e.g., \citet{Newey:McFadden:1994}, p. 2124),
this result implies identification of $b_{0}$, with $b_{0}$ being
the only solution to the first-order conditions
\begin{equation}
E\left[\psi(Y,X,b_{0})\right]=0,\;\psi(Y,X,b)\equiv-T(X,Y)(b'T(X,Y))+\frac{t(X,Y)}{b't(X,Y)},\;b\in\Theta.\label{eq:FOCs}
\end{equation}
When $Y|X\sim N(0,1)$, an interesting connection with the the Stein
equation for standard Gaussian random variables arises (e.g., Lemma
2.1 in \citet{Chen Gold Shao 2010}). In that case, $b_{0}'T(X,Y)=Y$
and $b_{0}'t(X,Y)=1$ satisfy the conditions of model (\ref{eq:e}).
Theorem \ref{thm:Thm1} then implies that $b_{0}=(0,1,0_{JK-2})'$
uniquely solves (\ref{eq:FOCs}):
\begin{align*}
E\left[\psi(Y,X,b_{0})\right]=E\left[-T(X,Y)Y+t(X,Y)\right] & =E[W(X)\otimes\{-S(Y)Y+s(Y)\}]\\
 & =E[W(X)]\otimes E[-S(Y)Y+s(Y)]=0,
\end{align*}
since $E[-S(Y)Y+s(Y)]=0$ has the form of the Stein equation, and
hence holds for any vector of continuously differentiable functions
$S(Y)$ with $E[|s_{j}(Y)|]<\infty$, $j\in\{1,\ldots,J\}$. In contrast,
(\ref{eq:FOCs}) holding with $b_{0}\neq(0,1,0_{JK-2})'$ will indicate
deviations of $Y$ from Gaussianity and independence from $X$. Since
$b_{0}$ satisfies (\ref{eq:e}), conditions (\ref{eq:FOCs}) thus
characterize a transformation of $Y$ to Gaussianity at each $X$ value. 
Hence, Theorem \ref{thm:Thm1} has the following testable implications for model (\ref{eq:e}).
\begin{cor}
\label{cor:Cor1}If there exists $b_{0}$ such that model (\ref{eq:e})
holds then, for any vectors of functions $\overline{W}(X)$ and of
continuously differentiable functions $\overline{S}(e)$ such that
$\overline{T}(X,e)\equiv\overline{W}(X)\otimes\overline{S}(e)$
and $\overline{t}(X,e)\equiv\partial_{e}\overline{T}(X,e)$ satisfy
Assumption \ref{ass: Ass2} with $T=\overline{T}$, $t=\overline{t}$
and $Y=e$, the following hold: (i) $(0,1,0_{JK-2})'$ uniquely solves
\[
\max_{b\in\overline{\Theta}}E[\log(\phi(b'\overline{T}(X,e))\{b'\overline{t}(X,e)\})],\quad\overline{\Theta}\equiv\{b\in\mathbb{R}^{JK}:\Pr[b'\overline{t}(X,e)>0]=1\},
\]
and (ii) the `Stein score' conditions $E[-\overline{T}(X,e)e+\overline{t}(X,e)]=0$
hold. 
\end{cor}

\subsection{Discussion\label{subsec:Discussion}}

The general modeling of $F_{Y|X}(Y|X)$ can be done indirectly by
specifying a representation for $Y$ given $X$,
\begin{equation}
Y=H(X,e),\quad e\mid X\sim F_{e},\label{eq:e form}
\end{equation}
with $H(X,e)$ strictly increasing in $e$, a random variable with
distribution $F_{e}$ and independent of $X$. The specification of
$H$ and $F_{e}$ then determines the form of $F_{Y|X}(Y|X)$:
\begin{equation}
F_{Y\mid X}(y\mid X)=F_{e}(H^{-1}(y,X)),\quad y\in\mathbb{R},\label{eq:eform2}
\end{equation}
where $y\mapsto H^{-1}(y,X)$ denotes the inverse function of $e\mapsto H(X,e)$.
In this approach, for a specified distribution $F_{e}$ the object
of modeling is the function $H(X,e)$. 

In Econometrics, (\ref{eq:e form}) is often characterized as `nonlinear
and nonseparable' in order to draw attention to the potentially complex
$Y$\textendash $X$ structure at constant $e$ and the lack of additive
structure in $e$ (e.g., \citet{Chesher:2003}, \citet{Matzkin:2003}).
These are essential features of $H$ that allow for the shape of the
conditional distribution of $Y$ to vary across values of $X$. An
alternative approach to (\ref{eq:e form})-(\ref{eq:eform2}) that
preserves nonlinearity and nonseparability is to model $F_{Y|X}(Y|X)$
directly as
\begin{equation}
F_{Y\mid X}(y\mid X)=F_{e}(g(y,X)),\quad y\in\mathbb{R},\label{eq:Y form}
\end{equation}
for some strictly increasing function $y\mapsto g(y,X)$. In our approach,
the object of modeling is the quantile transform $g(X,Y)=F_{e}^{-1}(F_{Y|X}(Y|X))$
which has distribution $F_{e}$ and is independent of $X$ by construction,
for some specified quantile function $F_{e}^{-1}$.

The difference between modeling the statistical relationship between $Y$ and $X$
according to (\ref{eq:eform2}) or (\ref{eq:Y form})
is not innocuous. With $f_{e}$ denoting the PDF of $e$, the definition
of the conditional PDF of $Y$ given $X$ implied by the indirect
approach (\ref{eq:eform2}),
\begin{equation}
f_{Y\mid X}(y\mid X)=f_{e}(H^{-1}(y,X))\{\partial_{y}H^{-1}(y,X)\},\quad y\in\mathbb{R},\label{eq:dens_eform}
\end{equation}
involves the inverse function of the modeling object $H$. In general
this inverse function does not have a closed-form expression, except
for some simple cases like the location-scale model $H(X,e)\equiv X'\beta_{1}+(X'\beta_{2})e$
with $X'\beta_{2}>0$. Furthermore, expression (\ref{eq:dens_eform})
gives rise to a nonconcave likelihood for even the simplest specifications
of $H$ and $F_{e}$, including the location and location-scale models
with Gaussian $e$ (\citet{Owen:2007}, \citet{SS:2018b}). In contrast,
a major advantage of representation (\ref{eq:Y form}) is that the
corresponding expression for $f_{Y|X}(Y|X)$ circumvents the inversion
step since
\[
f_{Y\mid X}(Y\mid X)=f_{e}(g(Y,X))\{\partial_{y}g(Y,X)\}.
\]
This formulation allows for the direct specification of flexible models
for $g(Y,X)$ that are characterized by a concave likelihood. Hence,
considerable computational advantages accrue in estimation when $e=g(Y,X)$
can be computed in closed-form, as further demonstrated by the duality
analysis in Section \ref{sec:Section5}. This formulation also leads
to well-defined representations for $F_{Y|X}(Y|X)$ under misspecification.

\section{Quasi-Gaussian Representations under Misspecification\label{sec:Section3}}

We study the properties of quasi-Gaussian representations for $F_{Y|X}(Y|X)$
that are generated by maximization of the objective $Q(b)$ under
general misspecification, i.e., when there is no representation of
the form (\ref{eq:e}) that satisfies either the Gaussianity or the
independence properties, or both. We find that the implied approximations
for the true DRFs are well-defined
and KLIC optimal.

\subsection{Existence and uniqueness}

Under Assumption \ref{ass: Ass2} the objective function $Q(b)$ is continuous and
strictly concave over $\Theta$, and hence admits at most one maximizer. 
Assumption \ref{ass: Ass2} is also sufficient for the level sets of $Q(b)$ to be compact, and hence for
existence of a maximizer. Compactness of the level sets 
is a consequence of the explosive behavior of $Q(b)$ at the boundary of $\Theta$. 
By the quadratic term $-\{b'T(X,Y)\}^{2}$
being negative, as $b$ approaches the boundary of $\Theta$ the log
Jacobian term diverges to $-\infty$, and hence so does $-\{b'T(X,Y)\}^{2}/2+\log\{b't(X,Y)\}$
on a set with positive probability. 
This is sufficient to conclude that the objective function $Q(b)$
diverges to $-\infty$, and hence that there exists at least one maximizer
to $Q(b)$ in $\Theta$, denoted $b^{*}$.

Under misspecification, to the maximizer $b^{*}$ corresponds the
quasi-Gaussian representation $e^{*}=T(X,Y)'b^{*}\equiv g^{*}(Y,X)$,
where $g^{*}(Y,X)$ is an element of the set
\[
\mathcal{E}\equiv\left\{ m:\Pr[m(Y,X)=b'T(X,Y)]=1\right\} 
\]
with $b\in\Theta$. By definition of $\Theta$, $y\mapsto b'T(X,y)$
is strictly increasing for each $b\in\Theta$ with probability one,
and hence each $m\in\mathcal{E}$ has a well-defined inverse function.
We note that nonsingularity of $E[T(X,Y)T(X,Y)']$ implies that $g^{*}(Y,X)$
is  unique in $\mathcal{E}$, i.e., there is no $m=g^{*}$ in $\mathcal{E}$
with $m(y,x)=b'T(x,y)$ everywhere and $b\neq b^{*}$.

\begin{sloppy}Define the range of $y\mapsto\Phi(m(y,x))$ as $\mathcal{U}_{x}(m)\equiv\{u\in(0,1):\Phi(m(y,x))=u\textrm{ for some }y\in\mathbb{R}\}$,
for $m\in\mathcal{E}$ and $x\in\mathcal{X}$. To the quasi-Gaussian
representation $g^{*}(Y,X)$ correspond a conditional PDF approximation,
defined as
\begin{equation}
f^{*}(Y,X)\equiv\phi(g^{*}(Y,X))\{\partial_{y}g^{*}(Y,X)\},\label{eq:f*(Y,X)}
\end{equation}
and conditional CDF and CQF approximations, defined as
\[
F^{*}(Y,X)\equiv\Phi(g^{*}(Y,X)),\quad Q^{*}(u,X)\equiv g^{*-1}(\Phi^{-1}(u),X),\quad u\in\mathcal{U}_{X}(g^{*}),
\]
where $e\mapsto g^{*-1}(e,X)$ denotes the inverse of $y\mapsto g^{*}(y,X)$.
These representations are unique in, respectively, the following spaces
\begin{align*}
\mathcal{D} & \equiv\left\{ f:\Pr[f(Y,X)=\phi(m(Y,X))\{\partial_{y}m(Y,X)\}]=1\right\} \\
\mathcal{F} & \equiv\left\{ F:\Pr[F(Y,X)=\Phi(m(Y,X))]=1\right\} \\
\mathcal{Q} & \equiv\left\{ Q:\Pr[Q(u,X)=m^{-1}(\Phi^{-1}(u),X)\;\textrm{for all}\;u\in\mathcal{U}_{X}(m)]=1\right\} 
\end{align*}
with $m\in\mathcal{E}$, and where $e\mapsto m^{-1}(e,X)$ denotes the
inverse of $y\mapsto m(y,X)$. Therefore, the DRF approximations are well-defined, with positive conditional
PDF and monotone conditional CDF and CQF approximations.\par\end{sloppy}
\begin{thm}
\label{thm:Thm2}If Assumption \ref{ass: Ass2} 
holds then there exists a unique maximum $b^{*}$ to $Q(b)$ in $\Theta$.
Consequently, the quasi-Gaussian representation $g^{*}(Y,X)$ and
the corresponding approximations for the DRFs are unique.
\end{thm}

\subsection{KLIC optimality}

When the elements of $\mathcal{D}$ are proper conditional PDFs, a
further motivation for the use of $Q(b)$ is the KLIC optimality of
the implied DRFs under misspecification (\citet{Akaike:1973}, \citet{White:1982}).

Since each $f\in\mathcal{D}$ satisfies $f>0$ by construction, an
element $f\in\mathcal{D}$ is a proper conditional PDF if it satisfies
$\int_{\mathbb{R}}f(y,X)dy=1$ with probability one. A necessary and
sufficient condition for this to hold is that the boundary conditions
\begin{equation}
\lim_{y\rightarrow-\infty}b'T(X,y)=-\infty,\quad\lim_{y\rightarrow\infty}b'T(X,y)=\infty,\label{eq:Boundary conditions}
\end{equation}
hold with probability one, for all $b\in\Theta$. Given a specified
dictionary such that (\ref{eq:Boundary conditions}) holds, Theorem
\ref{thm:Thm2} implies that the approximation $f^{*}(Y,X)$ in (\ref{eq:f*(Y,X)})
is the unique maximum selected by the population criterion in $\mathcal{D}$,
i.e.,
\[
f^{*}=\arg\max_{f\in\mathcal{D}}E\left[\log f(Y,X)\right],
\]
and hence that $f^{*}(Y,X)$ is the KLIC closest probability distribution
to $f_{Y|X}(Y|X)$. The corresponding $F^{*}$ and $Q^{*}$ are then
the KLIC optimal conditional CDF and CQF approximations for $F_{Y|X}(Y|X)$
and $Q_{Y|X}(u|X)$, respectively.
\begin{thm}
\label{thm:Thm3}If $E[|\log f_{Y|X}(Y|X)|]<\infty$ and (\ref{eq:Boundary conditions})
holds with probability one for all $b\in\Theta$, then $f^{*}$ is
the KLIC closest probability distribution to $f_{Y|X}(Y|X)$ in $\mathcal{D}$,
i.e., 
\[
f^{*}=\arg\min_{f\in\mathcal{D}}E\left[\log\left(\frac{f_{Y\mid X}(Y\mid X)}{f(Y,X)}\right)\right],
\]
where each $f\in\mathcal{D}$ is a proper conditional PDF. Moreover,
$f^{*}$ is related to the KLIC optimal conditional CDF $F^{*}$ in
$\mathcal{F}$ by
\[
F^{*}(y,X)=\int_{-\infty}^{y}f^{*}(t,X)dt,\quad y\in\mathbb{R},
\]
and to the well-defined inverse of $y\mapsto F^{*}(y,X)$, the KLIC
optimal CQF $u\mapsto Q^{*}(X,u)$ in $\mathcal{Q}$ with derivative
\[
\frac{\partial Q^{*}(X,u)}{\partial u}=\frac{1}{f^{*}(Q^{*}(X,u),X)}>0,\quad u\in(0,1),
\]
with probability one.
\end{thm}
Under the boundary conditions (\ref{eq:Boundary conditions}), the
set $\mathcal{F}$ is the space of conditional CDFs with Gaussian
representation in linear form, and the set $\mathcal{Q}$ is the space
of corresponding well-defined CQFs. A necessary and sufficient condition
for (\ref{eq:Boundary conditions}) is obtained, for instance, if
the limits $\lim_{y\rightarrow\pm\infty}|S_{j}(y)|$ are finite, $j\in\{3,\ldots,J\}$.
Under this maintained condition, the varying coefficients representation
of $e$ in (\ref{eq:e}),
\begin{equation}
e=\beta_{1}(X)+\beta_{2}(X)Y+\sum_{j=3}^{J}\beta_{j}(X)S_{j}(Y),\;\beta_{j}(X)=W(X)'b_{0j},\;j\in\{1,\ldots,J\},\label{eq:Varying}
\end{equation}
implies that $\beta_{2}(X)>0$ is necessary for the boundary conditions
(\ref{eq:Boundary conditions}) because otherwise $\lim_{y\rightarrow\infty}\beta(X)'S(y)$
would be finite or $-\infty$, and $\lim_{y\rightarrow-\infty}\beta(X)'S(y)$
would be finite or $\infty$, with $\beta(X)=(\beta_{1}(X),\ldots,\beta_{J}(X))'$. 
The linear term $\beta_{2}(X)Y$ in (\ref{eq:Varying}) implies that $\beta_{2}(X)>0$ 
is also sufficient for (\ref{eq:Boundary conditions}). We note
that $\beta_{2}(X)>0$ is implied by the derivative condition $\beta(X)'s(Y)=\beta_{2}(X)+\sum_{j=3}^{J}\beta_{j}(X)s_{j}(Y)>0$
if the transformations $s_{j}(Y)$, $j\in\{3,\ldots,J\}$, are specified
to be zero outside some compact region of $\mathbb{R}$,\footnote{This and the maintained assumption that $\lim_{y\rightarrow\pm\infty}|S_{j}(y)|<\infty$
are satisfied for instance if, for each $j\in\{3,\ldots,J\}$, the
transformations $S_{j}(Y)$ are defined as $S_{j}(y)\equiv\int_{-\infty}^{y}s_{j}(t)dt$,
for nonnegative spline functions $s_{j}(Y)\neq0$ on a compact subset
of $\mathbb{R}$, as $s_{j}(Y)=0$ outside this region and $S_{j}(Y)$
is then a CDF over the entire real line (\citet{CS:1966}, \citet{Ramsay:1988}).} since the derivative then reduces to $\beta_{2}(X)$ outside this
region. The boundary conditions (\ref{eq:Boundary conditions}) then
effectively hold under a location-scale restriction in the tails of
the distribution of $Y$ given $X$. We also note that (\ref{eq:Boundary conditions})
always holds for $J=2$ since the derivative condition is $\beta_{2}(X)>0$
in that particular case.

\section{Estimation and Inference\label{sec:Section4}}

\subsection{Maximum Likelihood estimation}

We assume that we observe a sample of $n$ independent and identically
distributed realizations $\{(y_{i},x_{i})\}_{i=1}^{n}$ of the random
vector $(Y,X')'$. Using the sample analog of $Q(b)$, we define the
GT regression estimator
\begin{equation}
\widehat{b}\equiv\arg\max_{b\in\Theta}n^{-1}\sum_{i=1}^{n}\left\{ -\frac{1}{2}[\log(2\pi)+\{b'T(x_{i},y_{i})\}^{2}]+\log(b't(x_{i},y_{i}))\right\} .\label{eq:ML}
\end{equation}
We derive the asymptotic properties of $\hat{b}$ under the following
assumptions.

\begin{assumption}(i) $\{(y_{i},x_{i})\}_{i=1}^{n}$ are identically
and independently distributed, and (ii) $E[||T(X,Y)||^{4}]<\infty$.\label{ass:Ass4}\end{assumption}

Assumption \ref{ass:Ass4}(i) can be replaced with the condition that
$\{(y_{i},x_{i})\}_{i=1}^{n}$ is stationary and ergodic (\citet{Newey:McFadden:1994}).
Assumption \ref{ass:Ass4}(ii) is needed for consistent estimation
of the asymptotic variance-covariance matrix of $\widehat{b}$.

\begin{sloppy}Recalling the definitions of $\gamma(Y,X,b)$ and $\Gamma(b)$
in (\ref{eq:hessian}) and $\psi(Y,X,b)$ in (\ref{eq:FOCs}), the
variance-covariance matrix of $\hat{b}$ is $\Gamma^{-1}\Psi\Gamma^{-1}/n$,
where $\Gamma\equiv\Gamma(b^{*})$ and $\Psi\equiv E[\psi(Y,X,b^{*})\psi(Y,X,b^{*})']$.
Estimators of $\Gamma$ and $\Psi$ are defined as $\widehat{\Gamma}=n^{-1}\sum_{i=1}^{n}\gamma(y_{i},x_{i},\hat{b})$
and $\widehat{\Psi}=n^{-1}\sum_{i=1}^{n}\psi(y_{i},x_{i},\hat{b})\psi(y_{i},x_{i},\hat{b})'$,
respectively. An estimator of $\Gamma^{-1}$ is any symmetric generalized
inverse $\widehat{\Gamma}^{-}$ of $\widehat{\Gamma}$. Under Assumptions
\ref{ass: Ass2} and \ref{ass:Ass4}, $\widehat{\Gamma}$ will be
nonsingular with probability approaching one (cf. Lemma \ref{lem:Nonsingular},
Section \ref{sec:Auxiliary-Results} in the Supplemental Material),
and hence $\widehat{\Gamma}^{-}$ will be the standard inverse.\par\end{sloppy}
\begin{thm}
\label{thm:Thm4}If Assumptions \ref{ass: Ass2}-\ref{ass:Ass4} hold,
then (i) there exists $\hat{b}$ in $\Theta$ with probability approaching
one; (ii) $\hat{b}\rightarrow_{p}b^{*}$; and (iii) $n^{\frac{1}{2}}(\hat{b}-b^{*})\rightarrow_{d}N(0,\Gamma^{-1}\Psi\Gamma^{-1})$.
Moreover, $\widehat{\Gamma}^{-}\widehat{\Psi}\widehat{\Gamma}^{-}\rightarrow^{p}\Gamma^{-1}\Psi\Gamma^{-1}$. 
\end{thm}
Theorem \ref{thm:Thm4}(i) shows that correct specification is not
required for existence of a globally monotone estimate $\hat{b}'T(X,Y)$
for large enough samples. Theorem \ref{thm:Thm4}(ii) holds without
compactness of $\Theta$, by concavity of the objective (e.g., \citet{Newey:McFadden:1994},
Section 2.6). Under correct specification, $\Gamma=-\Psi$ by the
information matrix equality and the estimator is efficient, with asymptotic
variance-covariance matrix $-\Gamma^{-1}$ (e.g., \citet{Newey:McFadden:1994}).
A valid method for model selection is adaptive Lasso ML (\citet{Zou:2006},
\citet{LGF:2012}, \citet{HN:2020}). The selected model is a sparse
KLIC optimal approximation for $g(Y,X)$.

\begin{rem}
Lasso penalized ML allows for the dimension of $W(X)$ to increase with sample size. 
From (\ref{eq:Varying}) it is apparent that the objective function has the multiple-index structure $E[m(W(X)'b_{1},\ldots,W(X)'b_{J},Y)]$, 
with strictly convex negative log-likelihood $(z_{1},\ldots,z_{J})\mapsto m(z_{1},\ldots,z_{J},Y)$.
High-dimensional estimation results and penalty selection methods
in \citet{CS:2021} apply to this case.
\end{rem}

\subsection{Estimation of DRFs}

An estimator for $g^{*}(y,x)$ is $\widehat{g}^{*}(y,x)\equiv T(x,y)'\widehat{b}$,
and estimators for DRFs are
\[
\widehat{f}^{*}(y,x)\equiv\phi(\widehat{g}^{*}(y,x))\{\partial_{y}\widehat{g}^{*}(y,x)\},\quad\widehat{F}^{*}(y,x)\equiv\Phi(\widehat{g}^{*}(y,x)),\quad(y,x)\in\mathcal{YX},
\]
and
\[
\widehat{Q}^{*}(x,u)\equiv\{y\in\mathbb{R}:\Phi(\widehat{g}^{*}(y,x))=u,\;\partial_{y}\widehat{g}^{*}(y,x)>0\},\quad x\in\mathcal{X},\quad u\in\mathcal{U}_{x}(g^{*}).
\]
These estimators are known functionals of $\widehat{b}$, and hence
their asymptotic distribution follows by application of the Delta
method.
\begin{thm}
\label{thm:Thm5}Suppose that $\Xi\equiv\Gamma^{-1}\Psi\Gamma^{-1}$
is positive definite. Under Assumptions \ref{ass: Ass2}-\ref{ass:Ass4}
we have: (i) for $(y,x)\in\mathcal{Y}\mathcal{X}$,
\[
n^{\frac{1}{2}}(\widehat{f}^{*}(y,x)-f^{*}(y,x))\rightarrow_{d}N(0,\phi(g^{*}(y,x))^{2}\Delta(x,y)'\Xi\Delta(x,y)),
\]
where $\Delta(x,y)\equiv-g^{*}(y,x)\{\partial_{y}g^{*}(y,x)\}T(x,y)+t(x,y)$,
and
\[
n^{\frac{1}{2}}(\widehat{F}^{*}(y,x)-F^{*}(y,x))\rightarrow_{d}N(0,\phi(g^{*}(y,x))^{2}T(x,y)'\Xi T(x,y));
\]
(ii) for $x\in\mathcal{X}$, $u\in\mathcal{U}_{x}(g^{*})$,
\[
n^{\frac{1}{2}}(\widehat{Q}^{*}(x,u)-Q^{*}(x,u))\rightarrow_{d}N(0,\{\partial_{y}g^{*}(y_{0},x)\}^{-2}T(x,y_{0})'\Xi T(x,y_{0})), 
\]
where $y_{0}=Q^{*}(u,x)$.
\end{thm}
Theorem \ref{thm:Thm4} provides an estimator for the asymptotic variance-covariance
matrix $\Xi$.

\begin{rem}
To implement (\ref{eq:ML}) we expand the 
parameter space $\Theta$ to the larger space $\Theta_{n}=\{b\in\mathbb{R}^{JK}:b't(x_{i},y_{i})>0,\,i\in\{1,\ldots,n\}\}$,
the effective domain of $Q_{n}(b)$. This implies that there is 
$b\in\Theta_{n}$ such that $b't(X,Y)\leq0$ with positive probability.
One can verify that $\widehat{b}\in\Theta$ holds after estimation
by checking the quasi-global monotonicity (QGM) property $\widehat{b}'t(x,y)>0$
on a fine grid of values that covers $\mathcal{Y}\times\mathcal{X}$.\label{Rk:QGM}
\end{rem}

\subsection{Comparison with alternative methods}

Our approach is related to distribution regression estimation of conditional
CDFs, for the (probit) model
\begin{equation}
F_{Y|X}(y\mid X)=\Phi(W(X)'\beta(y)),\quad y\in\mathbb{R},\label{eq:DRmodel}
\end{equation}
where $\beta(y)$ a vector of unknown functions. One sense in which
our approach is flexible is that model (\ref{eq:e}) can approximate
$F_{Y|X}(y|X)$ in (\ref{eq:DRmodel}) arbitrarily well for $S(y)$
rich enough.\footnote{Formally, write $S^{J}(Y)=S(Y)$ and suppose that Assumptions \ref{ass: Ass2} 
holds, $E[||\beta(Y)||^{2}]<\infty$ and that for any $K$ vector 
of functions $b(Y)$ with $E[||b(Y)||^{2}]<\infty$ there are $J\times1$ 
vectors $a_{k}^{J}$, $k\in\{1,\ldots,K\}$, such that $E[\sum_{k=1}^{K}\{b_{k}(Y)-S^{J}(Y)'a_{k}^{J}\}^{2}]\rightarrow0$ 
as $J\rightarrow\infty$. Then, $E[\{F_{Y|X}(Y|X)-\Phi(b'[W(X)\otimes S^{J}(Y)])\}^{2}]\rightarrow0$ 
as $J\rightarrow\infty$, by an argument similar to Theorem 11 in \citet{NeweyStouli:2025}.} If we define $\beta^{*}(y)=(\beta_{1}^{*}(y),\ldots,\beta_{K}^{*}(y))'$
with $\beta_{k}^{*}(y)\equiv S(y)'b_{k}^{*}$ and $b_{k}^{*}=(b_{k1}^{*},\ldots,b_{kJ}^{*})'$,
$k\in\{1,\ldots,K\}$, then $F^{*}(y,X)=\Phi(W(X)'\beta^{*}(y))$
provides increasingly accurate KLIC optimal approximations to (\ref{eq:DRmodel})
as $J$ increases. When $\beta^{*}(y)=\beta(y)$, our formulation
characterizes $\beta(y)$ for each $y$ in $\mathbb{R}$ simultaneously
whereas distribution regression characterizes $\beta(y)$ pointwise.
While both corresponding estimators are $\sqrt{n}$ consistent, the
ML estimator (\ref{eq:ML}) is efficient. When (\ref{eq:DRmodel})
is misspecified and multiple values of $Y$ are of interest, the distribution
regression criterion does not provide approximation guarantees for
$F_{Y|X}(y|X)$, whereas $F^{*}(y,X)$ is KLIC optimal.

Both conditional CDF models in (\ref{eq:DRF1}) and (\ref{eq:DRmodel})
are particular cases of conditional transformation models of the form
$F(g(Y,X))=F(\sum_{l=1}^{L}g_{l}(Y,X))$ (\citet{HKB:2014}), for
some specified CDF $F$. For estimation, taking this additive structure
as a starting point leads to considering the restricted form
\begin{equation}
F(g(Y,X))=F(\sum_{l=1}^{L}g_{l}(Y,X_{l}))=F(\sum_{l=1}^{L}b_{l}'[W_{l}(X_{l})\otimes S(Y)]),X=(X_{1},\ldots,X_{L})'\label{eq:CTM}
\end{equation}
(\citet[Section 5]{HKB:2014}), with monotonicity constraints on each
partial transformation $g_{l}$ (\citet[Section 4.5]{HMB:2018}, \citet[p. 1362]{CKK:2024}).
Here we find that the additive structure imposed by conditional transformation
models is not required for the formulation of fully flexible models
for conditional CDFs. Model (\ref{eq:DRF1}) gives a flexible generalization
of (\ref{eq:CTM}) and does not impose unnecessary monotonicity
restrictions.

Our approach is also related to quantile regression estimation of
CQFs. For the quantile regression model $Q_{Y|X}(u|X)=W(X)'\beta(u)$,
$u\in(0,1)$, the coefficients $\beta(u)$ are estimated at each $u$
by a sequence of linear programming problems. In general, this model
and (\ref{eq:e}) are not nested, but they coincide for $W(X)$ and
$S(Y)$ rich enough. Compared to our approach, the quantile regression
loss function allows for outcomes with non finite second moment, and
enjoys robustness properties when estimating a specific quantile.
When multiple quantiles are of interest, quantile regression does
not provide approximation guarantees for $Q_{Y|X}(u|X)$, whereas
$Q^{*}(X,u)$ is KLIC optimal.

Quantile and distribution regression can result in both finite sample
estimates and population approximations under misspecification that
do not satisfy the monotonicity properties of CQFs and CDFs, respectively.\footnote{This problem has motivated the development of a variety of methods
to avoid or repair intersecting quantile surfaces (e.g., \citet{He:1997},
\citet{DV:2008}, \citet{Chernozhukov Fern Gali 2010}, \citet{YT:2017},
\citet{SS:2018a}).} When this occurs, rearrangement methods provide improvements 
(\citet{Chernozhukov Fern Gali 2010}),
but without optimality guarantees. Here we propose a one-step resolution
to the quantile and probability curves crossing problem through the
information-theoretic selection of globally monotone models within
a specified class. Our criterion performs global model comparisons
within a set of proper conditional PDFs, and hence allows for the
selection of an optimal KLIC approximation. The implied CQF is then
free of crossing and the conditional CDF monotone, both are optimal
in a transparent sense, and they are estimated at parametric rate
together with the conditional PDF.

Compared to nonparametric kernel-based methods, flexible estimation
at parametric rate alleviates the curse of dimensionality when $X$
is a vector. In our simulations we find that it also yields substantial
finite sample improvements for conditional PDF, CDF and nonseparable
model estimation when $X$ is a scalar. The QGM property in Remark \ref{Rk:QGM}
avoids the need for correcting estimates (\citet{GHI:2003}) and multiplicity
of inverses (cf. \citet{Matzkin:2003}, p. 1358). The linear form
of GT estimates makes the QGM property easy to check and has the advantage
that imposing shape constraints from economic theory in estimation
is straightforward. Compared to \citet{Matzkin:2003}, the direct
specification of the distribution of $e$ through the log density
in (\ref{eq:log PDF}) also simplifies the choice and imposition of
a normalization.

Because $Q^{*}(X,u)$ is a proper CQF it can be used for simulation
of data that mimics the true data generating process, using the representation
$\widetilde{Y}=Q^{*}(X,U)$, $U|X\sim U(0,1)$. Thus, our approach
also complements Generative Adversarial Networks (GAN) (\citet{Goodfellow:2014})
used in Econometrics for simulation of complex datasets (\citet{AIMM:2024}).
GANs are flexible predictive methods able to mitigate the curse of
dimensionality. Compared to our approach, GANs do not directly produce
conditional distribution estimates and require solving computationally
challenging nonconvex nonconcave min-max games. Recent proposals substitute
the Wasserstein distance in the objective function and use a gradient-based
penalty for more stable implementation (\citet{AB:2017}). In contrast
our approach uses the KLIC to perform ML estimation of conditional
distributions in closed-form, while global concavity alleviates computational
difficulties.

\section{Duality Theory\label{sec:Section5}}

Considerable computational advantages accrue from our ML criterion
where the GT enters in closed-form. To (\ref{eq:ML}) corresponds
a dual formulation that can be cast into the modern convex programming
framework (\citet{BV:2004}). We derive the dual problem and establish
the properties of the dual solutions.
\begin{thm}
\label{thm:Thm6} If Assumptions \ref{ass: Ass2}-\ref{ass:Ass4}
are satisfied then the following hold.

(i) The dual of (\ref{eq:ML}) is
\begin{align}
\min_{(u,v)\in\mathbb{R}^{n}\times(-\infty,0)^{n}} & -n\left(\frac{1}{2}\log(2\pi)+1\right)+\sum_{i=1}^{n}\left\{ \frac{u_{i}^{2}}{2}-\log\left(-v_{i}\right)\right\} \label{eq:Dual objective}\\
\textrm{subject to}\quad & -\sum_{i=1}^{n}\left\{ T(x_{i},y_{i})u_{i}+t(x_{i},y_{i})v_{i}\right\} =0\label{eq:Dual scores}
\end{align}
the dual GT regression problem, with solution $\widehat{\alpha}=(\widehat{u}',\widehat{v}')'$.

(ii) The dual GT regression program (\ref{eq:Dual objective})-(\ref{eq:Dual scores})
admits the method-of-moments representation
\begin{equation}
\sum_{i=1}^{n}\left\{ -T(x_{i},y_{i})\{b'T(x_{i},y_{i})\}+\frac{t(x_{i},y_{i})}{b't(x_{i},y_{i})}\right\} =0,\label{eq:MMrep}
\end{equation}
the first-order conditions of (\ref{eq:ML}).

(iii) With probability approaching one we have: (a) existence and
uniqueness, i.e., there exists a unique pair $(\widehat{b}',\widehat{\alpha}')'$
that solves (\ref{eq:ML}) and (\ref{eq:Dual objective})-(\ref{eq:Dual scores}),
and
\begin{equation}
\widehat{u}_{i}=\widehat{b}'T(x_{i},y_{i}),\quad\hat{v}_{i}=-\frac{1}{\widehat{b}'t(x_{i},y_{i})},\quad i\in\{1,\ldots,n\};\label{eq:Dual_FOCs}
\end{equation}
(b) strong duality, i.e., the value of (\ref{eq:ML}) equals the value
of (\ref{eq:Dual objective})-(\ref{eq:Dual scores}).
\end{thm}
\begin{sloppy}
The dual formulation in Theorem \ref{thm:Thm6} demonstrates important
computational properties of GT regression. The Hessian matrix of the
dual problem (\ref{eq:Dual objective})-(\ref{eq:Dual scores}) is
\[
\left[\begin{array}{cc}
I_{n} & 0_{n\times n}\\
0_{n\times n} & \textrm{diag}(1/v_{i}^{2})
\end{array}\right],
\]
a positive definite diagonal matrix for all $v\in(-\infty,0)^{n}$,
with $I_{n}$ denoting the $n\times n$ identity matrix and $\textrm{diag}(1/v_{i}^{2})$ 
the $n\times n$ diagonal matrix with elements $(1/v_{1}^{2},\ldots,1/v_{n}^{2})$.
Thus the dual problem is a strictly convex mathematical program with
sparse Hessian matrix and $JK$ linear constraints. We implement this
computationally convenient formulation using the state-of-the-art
convex programming solvers ECOS (\citet{ECOS:2013}) and SCS (\citet{ODCPB2016}).
\par\end{sloppy}

It is useful to compare problem (\ref{eq:Dual objective})-(\ref{eq:Dual scores})
with the dual formulation of quantile regression. The dual problem
for the linear $\tau$ quantile regression of $y_{i}$ on $W(x_{i})$
is
\[
\max_{u}\left\{ y'u:\sum_{i=1}^{n}W(x_{i})u_{i}=(1-\tau)\sum_{i=1}^{n}W(x_{i}),\quad u\in[0,1]^{n}\right\} ,\quad\tau\in(0,1),
\]
with solutions that take value $0$ or $1$, except for $K$ sample
points that define the fitted quantile regression surface and are
assigned $u$ values that are neither $0$ nor $1$ (cf. Chapter 3.5.4
in \citet{Koenker:2005}). When multiple quantiles are of interest,
violation of monotonicity may arise because each quantile surface
must interpolate $K$ sample points, while each quantile regression
is implemented separately. In contrast, (\ref{eq:Dual objective})-(\ref{eq:Dual scores})
assigns $u$ values in $\mathbb{R}^{n}$, avoiding box constraints
on $u$, while imposing monotonicity since $-1/v_{i}>0$ at a solution.
Implied quantile surfaces are then unrestricted beyond the form of
$u_{i}$ in (\ref{eq:Dual_FOCs}), and strong duality ensures KLIC
optimality.

In addition to KLIC optimality and the logarithmic barrier in the
objective, linearity of the constraints is also an important advantage
of (\ref{eq:Dual objective})-(\ref{eq:Dual scores}) compared to
the alternative generalized dual regression characterization of conditional
CDFs and CQFs (\citet{SS:2018a}) for which the mathematical program
is of the form
\begin{equation}
\max\left\{ y'e:\sum_{i=1}^{n}\mathcal{T}(x_{i},e_{i})=0,\quad e\in\mathbb{R}^{n}\right\} ,\label{eq:DualReg}
\end{equation}
where $\mathcal{T}(x_{i},e_{i})$ is a vector of known functions of
$x_{i}$ and $e_{i}$ including $e_{i}$ and $(e_{i}^{2}-1)/2$, so
that $e$ enters nonlinearly into the constraints. (\ref{eq:DualReg})
has first-order conditions
\begin{equation}
y_{i}=\widehat{b}'\{\partial_{e_{i}}\mathcal{T}(x_{i},e_{i})\},\quad i\in\{1,\ldots,n\},\label{eq:GDR_FOCs}
\end{equation}
where $\widehat{b}$ is the Lagrange multiplier vector for the constraints
in (\ref{eq:DualReg}), but where the solution is now determined by
a system of $n$ nonlinear equations instead of having a closed-form
expression as in (\ref{eq:Dual_FOCs}). This further illustrates the
benefits accruing from closed-form modeling of $e=g(Y,X)$, compared
to direct modeling of $y_{i}$ in (\ref{eq:GDR_FOCs}).

\begin{rem}
Compared to the methods above, (\ref{eq:Dual objective})-(\ref{eq:Dual scores})
provides a one-step estimator for both conditional PDFs and CDFs at
each sample point, formed as $\varPhi(\widehat{u}_{i})$ and $\phi(\widehat{u}_{i})(-1/\hat{v}_{i})$,
respectively. Moreover, uniform convergence of the empirical distribution
of $\{\widehat{u}_{i}\}_{i=1}^{n}$ can be used to establish the validity
of bootstrap methods for uniform inference on DRFs (\citet{Chernozhukov Fern Melly 2013}). Uniform convergence
follows from the method-of-moments representation (\ref{eq:MMrep})
and steps similar to the proof of Theorem 6 in \citet{SS:2018a}.
\end{rem}

\section{Distributional Gender Wage Gap Analysis\label{sec:Section6}}

We illustrate our framework with an application to the distributional
gender wage gap in the United States. We use data on wages, hours
and weeks worked, educational attainment, gender, race, age, industry
and occupation from the 2019 American Community Survey (\citet{Rugglesetal:2025}).\footnote{A large representative 
survey that covers 1\% of the U.S. population,
with mandatory participation, and available on the IPUMS-USA website
(https://usa.ipums.org/usa/).} For a sample of white employees working full time (more than $34$
hours a week, at least $50$ weeks a year) in metropolitan areas, we stratified the data
according to industry-occupation pairs, and selected the 41 pairs
for which the support of education is the same across genders, to allow for wage comparisons over 
the whole support.
We follow \citet{BCS:2024} and include individuals 
aged $25$ to $65$, and we discard individuals below the $7.25$ federal minimum hourly
wage. 
We have 199,785 observations in total, with 245 to 32,828 observations per industry-occupation pair.

For each pair, we estimate DRFs implied
by the model
\begin{equation}
e^{*}=g^{*}(Y,X,D)=[\{W(X)\otimes(1,D)'\}\otimes S(Y)]'b^{*},\quad X=(X_{1},X_{2})',\label{eq:e* wages}
\end{equation}
where $Y$ denotes weekly wages, $X_{1}$ years of education, $X_{2}$
experience, and $D$ is a gender dummy variable. This provides 
a flexible framework for estimation of distributional and quantile
treatment effects. Differences in wages across genders occur if some
component of $D[W(X)\otimes S(Y)]$ in (\ref{eq:e* wages}) has nonzero
coefficient. The sign, magnitude and location of gender differences
in the wage  distribution can be analyzed by estimating the quantile
gender wage gap over the support of $X$ and $u\in(0,1)$,
\begin{equation}
\mathtt{GWG}(X,u)=100\times\left[\log Q^{*}(X,D=\mathtt{Male},u)-\log Q^{*}(X,D=\mathtt{Female},u)\right].\label{eq:qGWG}
\end{equation}
Following \citet{BCS:2024} we treat education and experience as continuous,
and for both $W(X)$ and $S(Y)$ we consider a range of spline transformations.\footnote{For model selection 
we implement an adaptive Lasso ML version of our 
estimator, and we select the specification that minimizes the Bayes 
Information Criterion (BIC). Details of our implementation for the empirical application are given 
in Section \ref{sec:Implementation} of the Supplemental material. 
All computational procedures can be implemented in the software R 
(\citet{R:2024}) using open source packages for convex optimization 
such as CVX, and its R implementation CVXR (\citet{CVXR}).} Spline functions satisfy the conditions of our modeling framework
and have been demonstrated to be remarkably effective when applied
to the related problems of log density estimation (\citet{KS:1991})
or monotone regression function estimation (\citet{Ramsay:1988}).
Thus, (\ref{eq:e* wages})-(\ref{eq:qGWG}) is a flexible specification
for the gender wage gap, allowing for both observed and unobserved
heterogeneity.

\begin{figure}[t]
\begin{centering}
\subfloat[Conditional PDF of earnings by $\textrm{Years of Education}\in\{12,16,19,20\}$.\label{fig:PDF}]{\begin{centering}
\includegraphics[width=7.25cm,height=5cm]{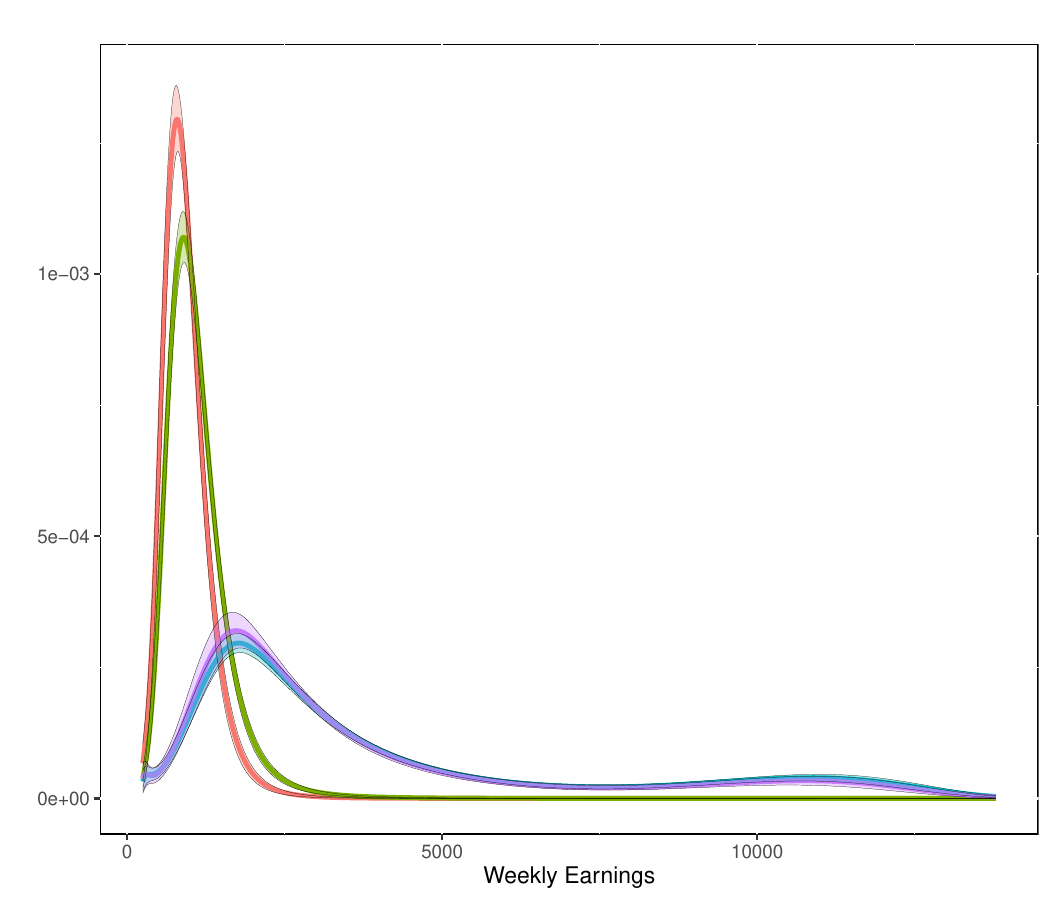}\hfill{}\includegraphics[width=7.25cm,height=5cm]{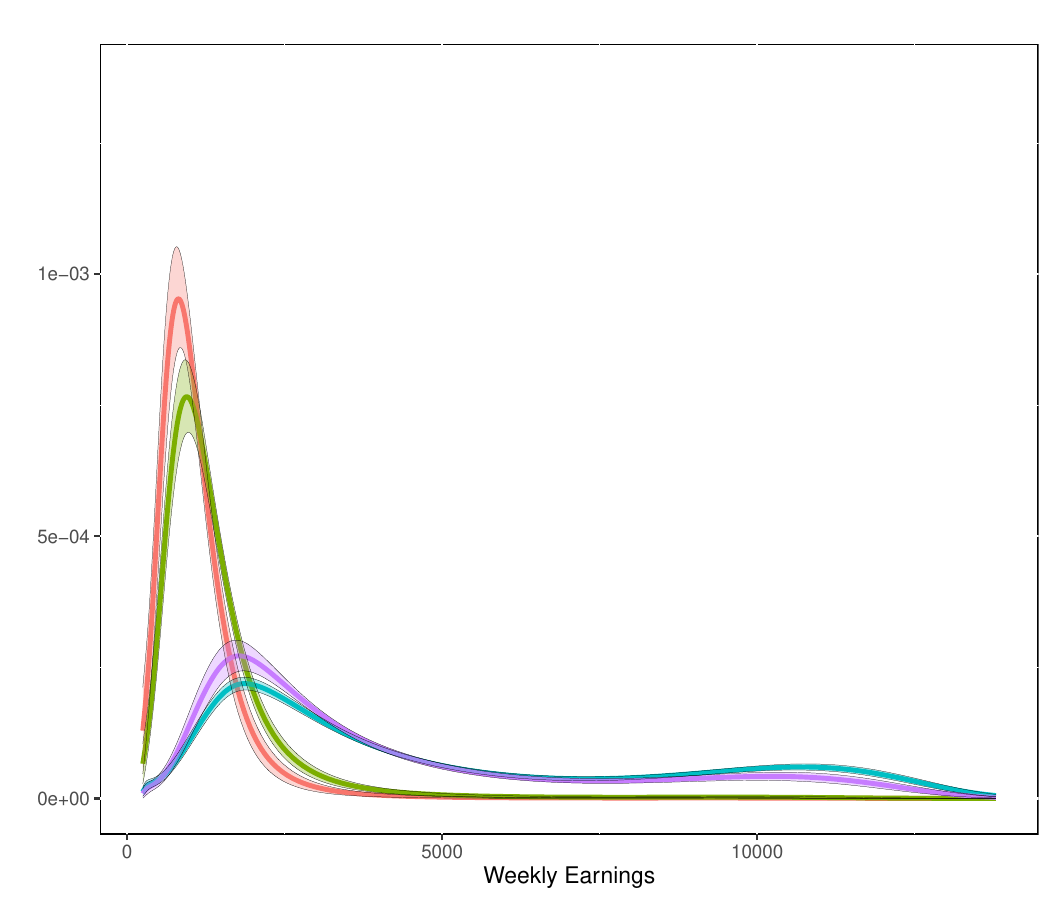}
\par\end{centering}
}
\par\end{centering}
\begin{centering}
\subfloat[Conditional CDF of earnings by $\textrm{Years of Education}\in\{12,16,19,20\}$.\label{fig:CDF}]{\begin{centering}
\includegraphics[width=7.25cm,height=5cm]{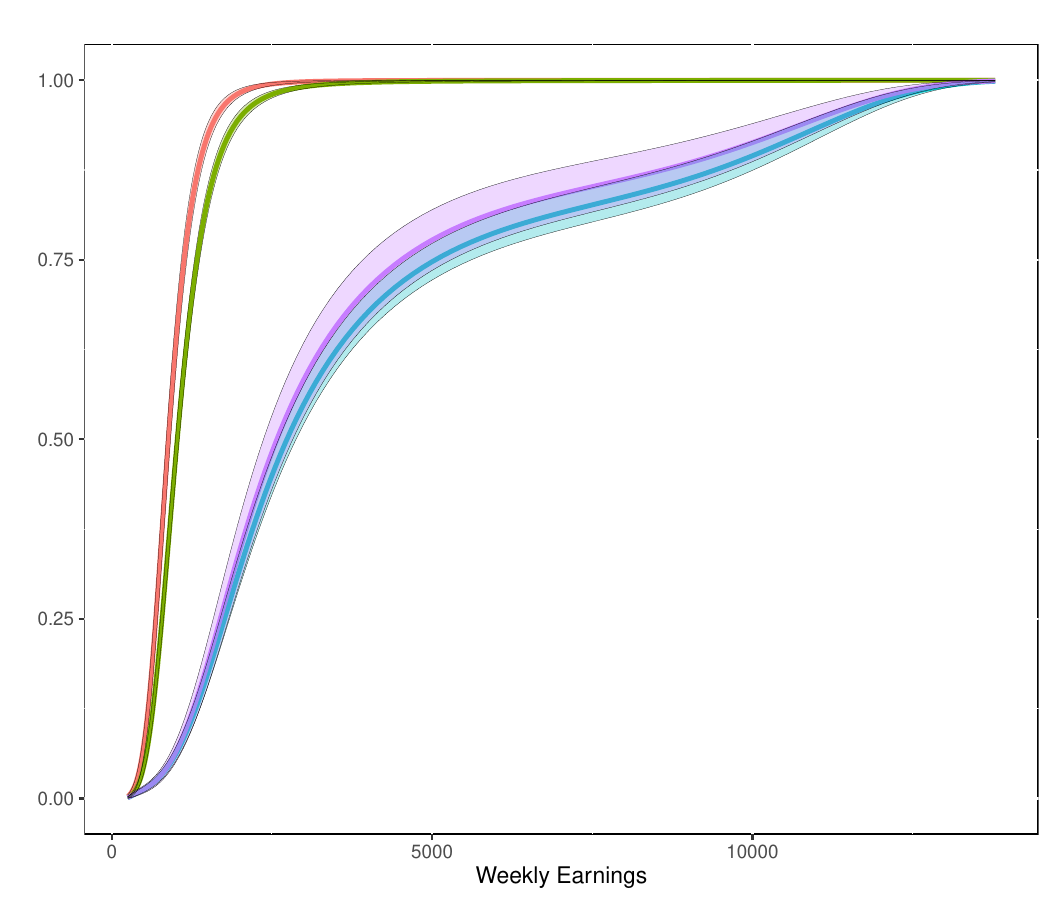}\hfill{}\includegraphics[width=7.25cm,height=5cm]{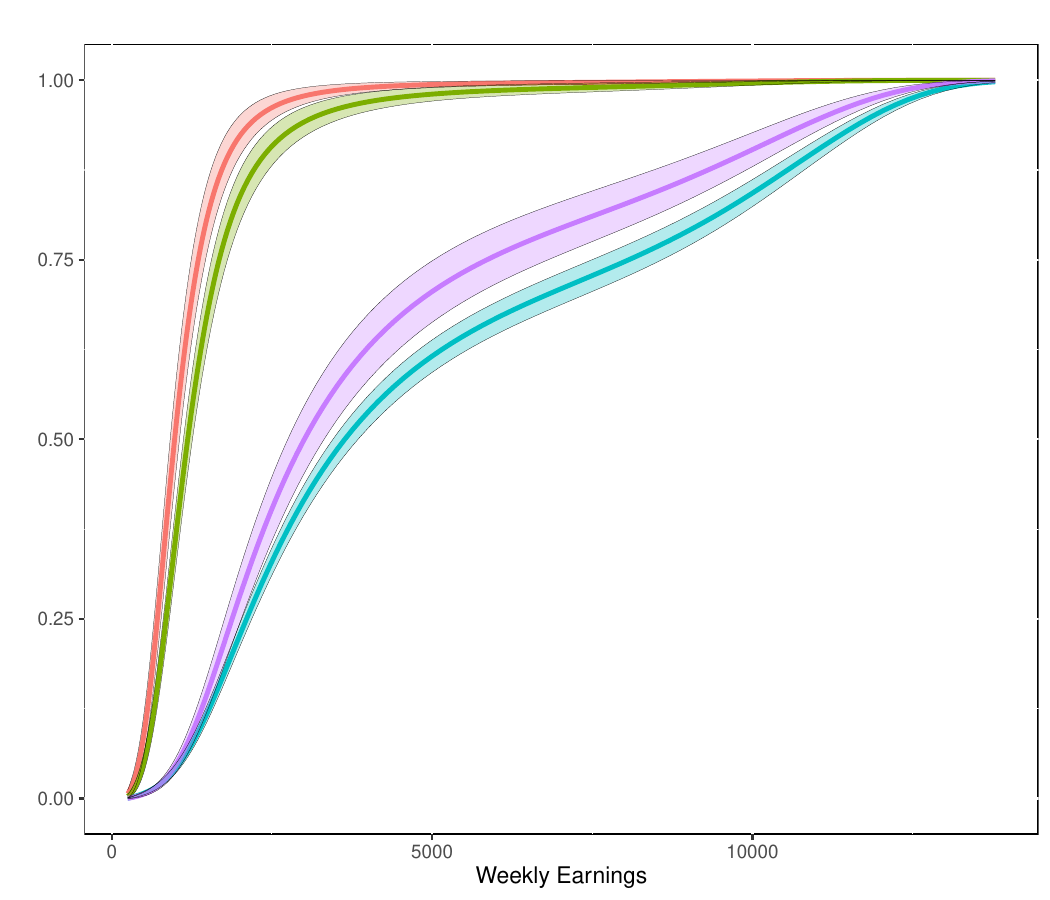}
\par\end{centering}
}
\par\end{centering}
\caption{Conditional PDF and CDF for $\mathtt{Female}$ (left) and $\mathtt{Male}$
(right), with confidence bands.\label{fig:PDFsCDFs}}
\end{figure}
Distributional gender wage gap analysis is challenging because the
shape of the conditional wage distribution varies across $X$ values,
and because of key features of the dataset such as high skewness of
wages, the presence of large outliers, multiple covariates, and the
effectively discrete measurements of education and experience. Applying
high-dimensional mean regression to similar data, \citet{BCS:2024}
find that accounting for heterogeneity in observables (e.g., education,
industry and occupation) is important in understanding the gender gap. 
We complement their analysis by studying distributional
wage differences across genders, and flexibly modeling observed heterogeneity
across education and experience levels. For each industry-occupation
pair, we obtain DRFs over
their entire support, test for their equality across genders, 
estimate $\mathtt{GWG}(X,u)$, and give confidence bands for
all objects.

\begin{figure}[t]
\subfloat[CQF with confidence bands, for $u\in\{0.1,0.25,0.5,0.75,0.9\}$.\label{fig:CQF_CI}]{\begin{centering}
\includegraphics[width=7.25cm,height=5cm]{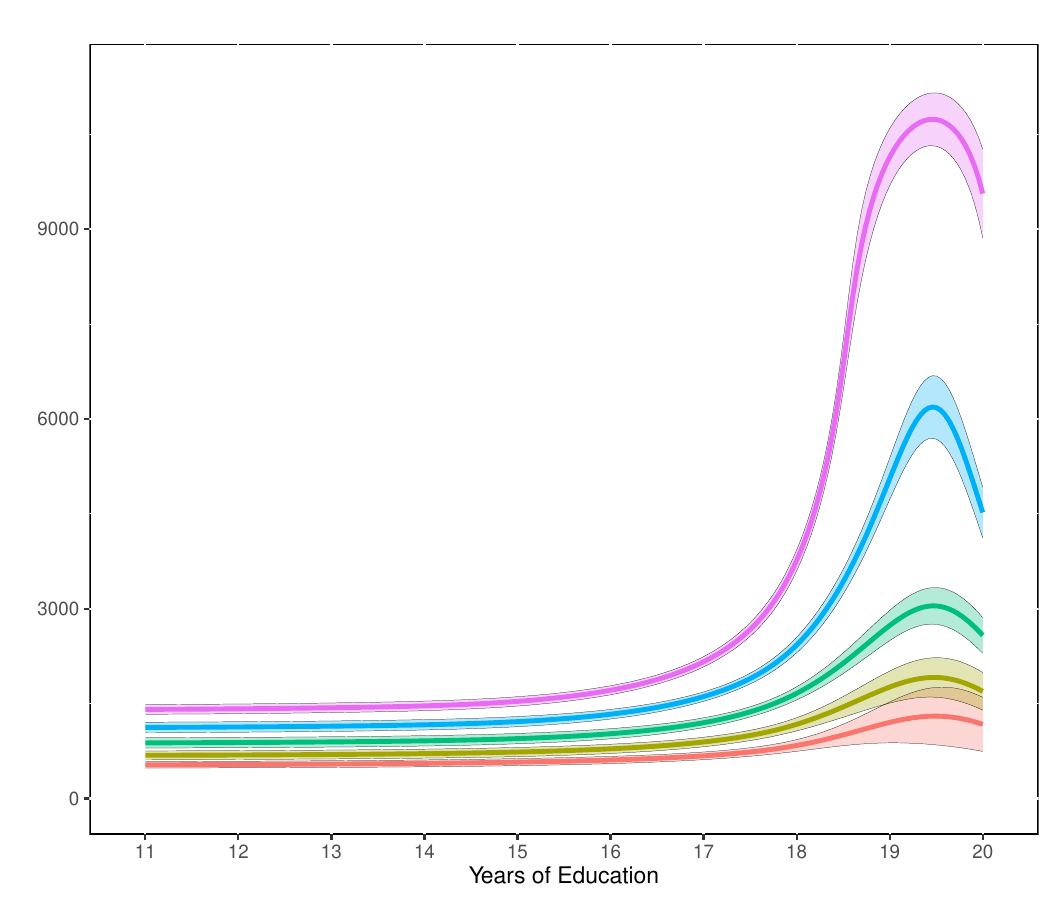}\hfill{}\includegraphics[width=7.25cm,height=5cm]{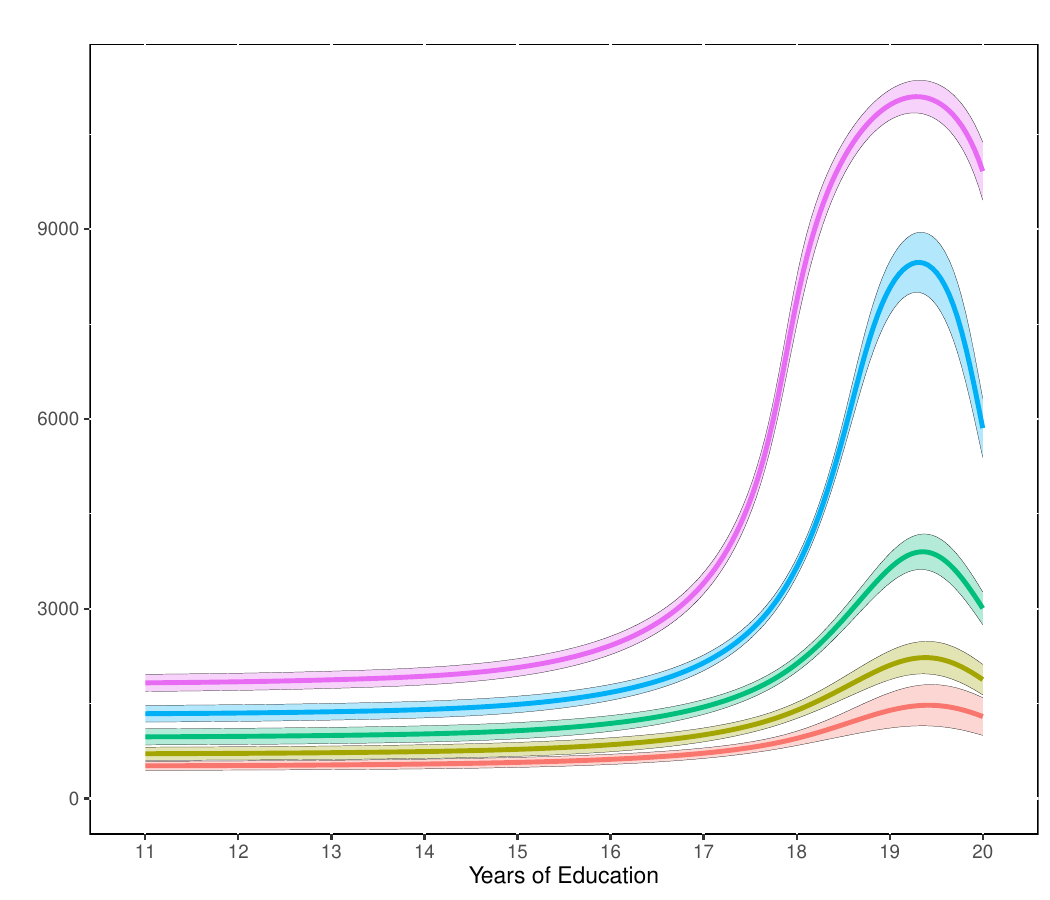}
\par\end{centering}
}

\subfloat[CQF with scatterplots by gender, for $u\in\{0.05,0.10,\ldots,0.95\}$.\label{fig:CQF_Full}]{\includegraphics[width=7.25cm,height=5cm]{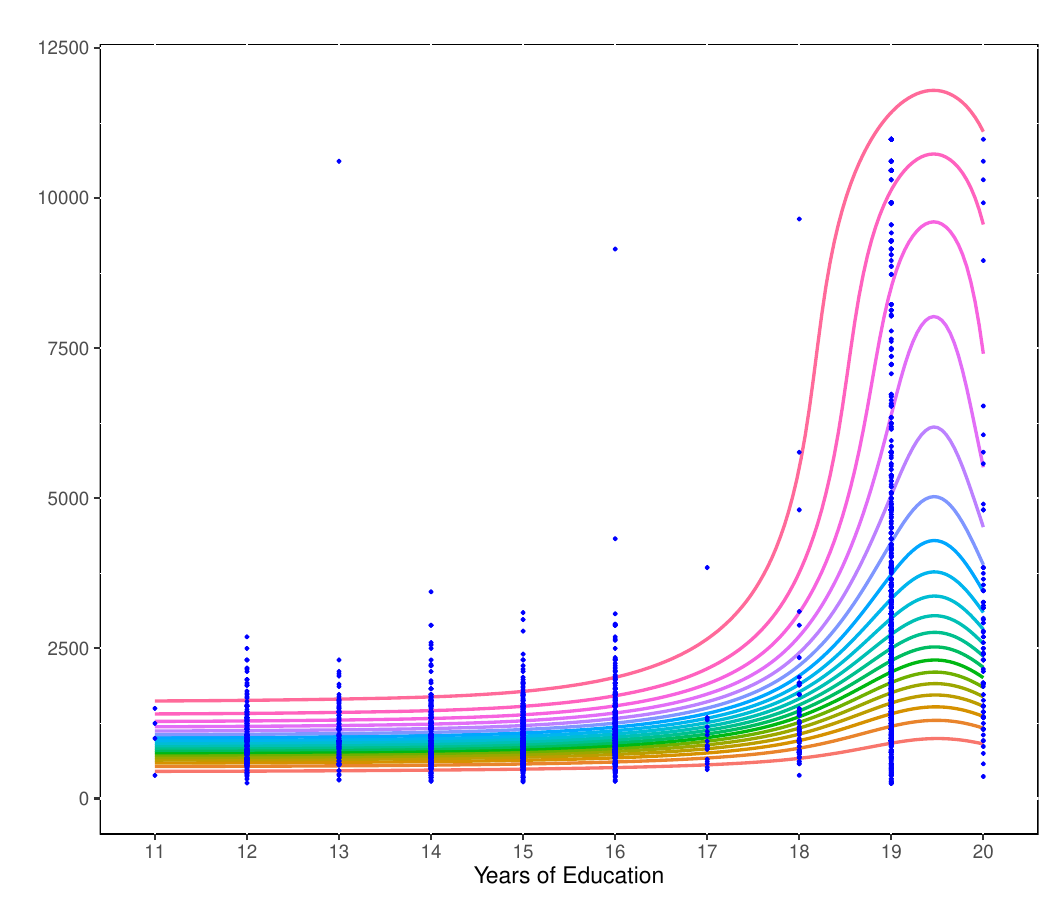}\hfill{}\includegraphics[width=7.25cm,height=5cm]{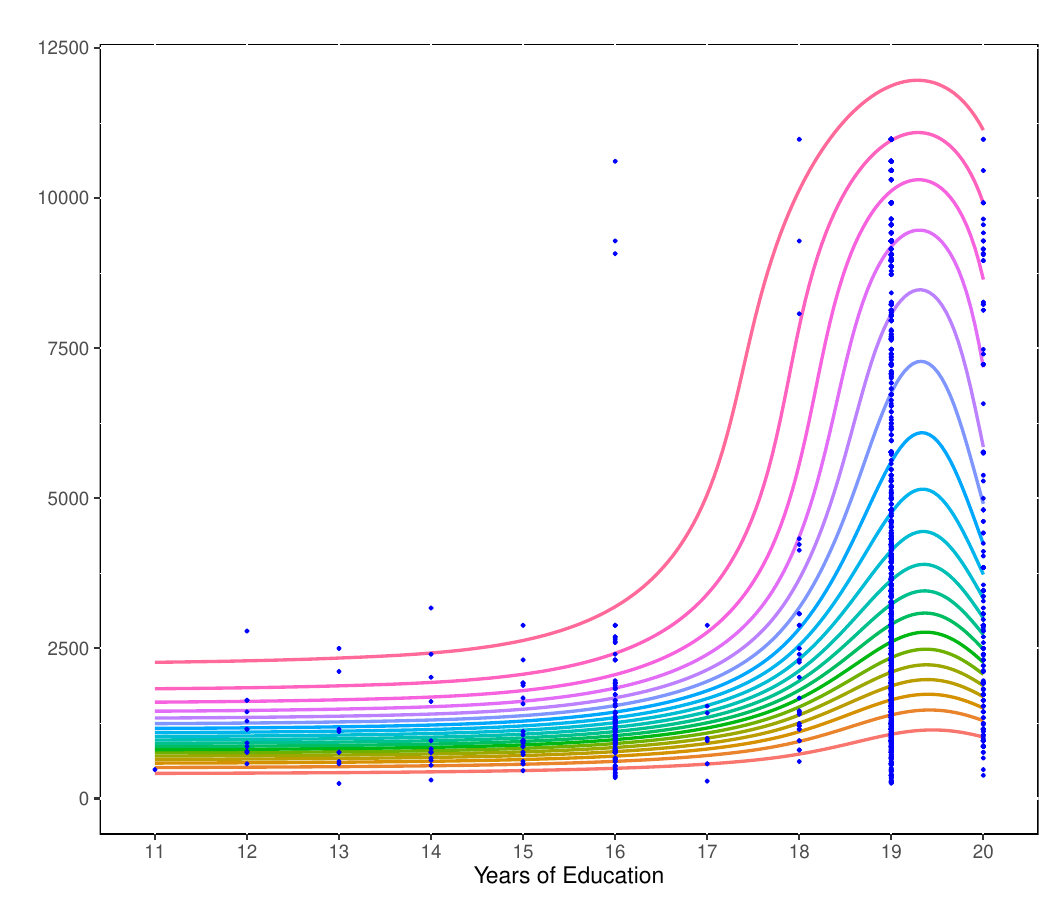}

}\caption{CQF for $\mathtt{Female}$ (left) and $\mathtt{Male}$ (right).\label{fig:CQFs}}
\end{figure}
Figures \ref{fig:PDFsCDFs}-\ref{fig:CQFs} show DRF estimates for the legal occupation in the services industry
and for $X_{2}=\widehat{Q}_{X_{2}}(0.5)$, the sample experience median. 
This example demonstrates that 
parsimonious Spline-Spline models 
can capture a wide variety of complex data features. The high
single peaked PDFs at $X_{1}\in\{12,16\}$ differ markedly from the
long-tailed and bimodal PDFs at $X_{1}\in\{19,20\}$.\footnote{$X_{1}\in\{12,16,19,20\}$ indicates completion of high school, bachelor,
professional and doctoral degrees, respectively.} This reveals two different regimes in the conditional wage distribution:
high concentration/relatively low wages at $X_{1}\in\{12,16\}$, and
high dispersion/relatively high wages at $X_{1}\in\{19,20\}$. The
two modes in the second regime are especially apparent for males in
Figure \ref{fig:PDFsCDFs}(A), and are reflected by the two inflection
points in the corresponding CDFs in Figure \ref{fig:PDFsCDFs}(B),
and by the large gap between the low and high CQFs 
in Figure
\ref{fig:CQFs}(A). CQF estimates capture both linearity and nonmonotonicity
over the $X_{1}$ support, reflecting substantial heteroskedasticity
and changes in mode locations for the wage distribution. 
Lower dispersion of CQFs up to the median reflects distributional asymmetry. 

Visual inspection strongly suggests that the DRFs differ across genders. Conditional PDFs for females are
higher for low wages, and those for males are higher 
for high wages. 
Conditional 
CDFs for males stochastically dominate those for females.
Upper quantile CQFs are higher for males. A Wald test of DRFs equality 
across genders reinforces this
diagnostic.\footnote{We perform a significance test for the $9$ nonzero coefficients of
$D[W(X)\otimes S(Y)]$ in (\ref{eq:e* wages}), with a test statistic
of $316>16.9$, the critical value at the $5\%$ level.}

Figure \ref{fig:QTE} gives a comprehensive picture of the quantile
gender wage gap. We find substantial heterogeneity, with a statistically
significant gap over the whole support of $X_{1}$ for higher quantiles
(median and above). We also find nonlinearity, with close to a constant
gap up to year 16, and then marked divergence across quantiles.

Overall, we find that parsimonious representations are able to capture
complex distributional data features. In the Supplemental Material we further
illustrate the overall finding that $\mathtt{GWG}(X,u)$ varies across
quantiles and exhibit substantial nonlinearity in $X_{1}$, with heterogeneous
patterns across industry-occupation pairs.
 The main features of the
selected Spline-Spline model are well-preserved by models with similar
penalization and BIC. 
In our simulations we also find that BIC
performs well in a variety of designs. Thus, although 
model selection 
in our context is an important topic for future research, 
BIC appears to be reliable for practical purposes.

\begin{figure}[t]
\includegraphics[width=9cm,height=5cm]{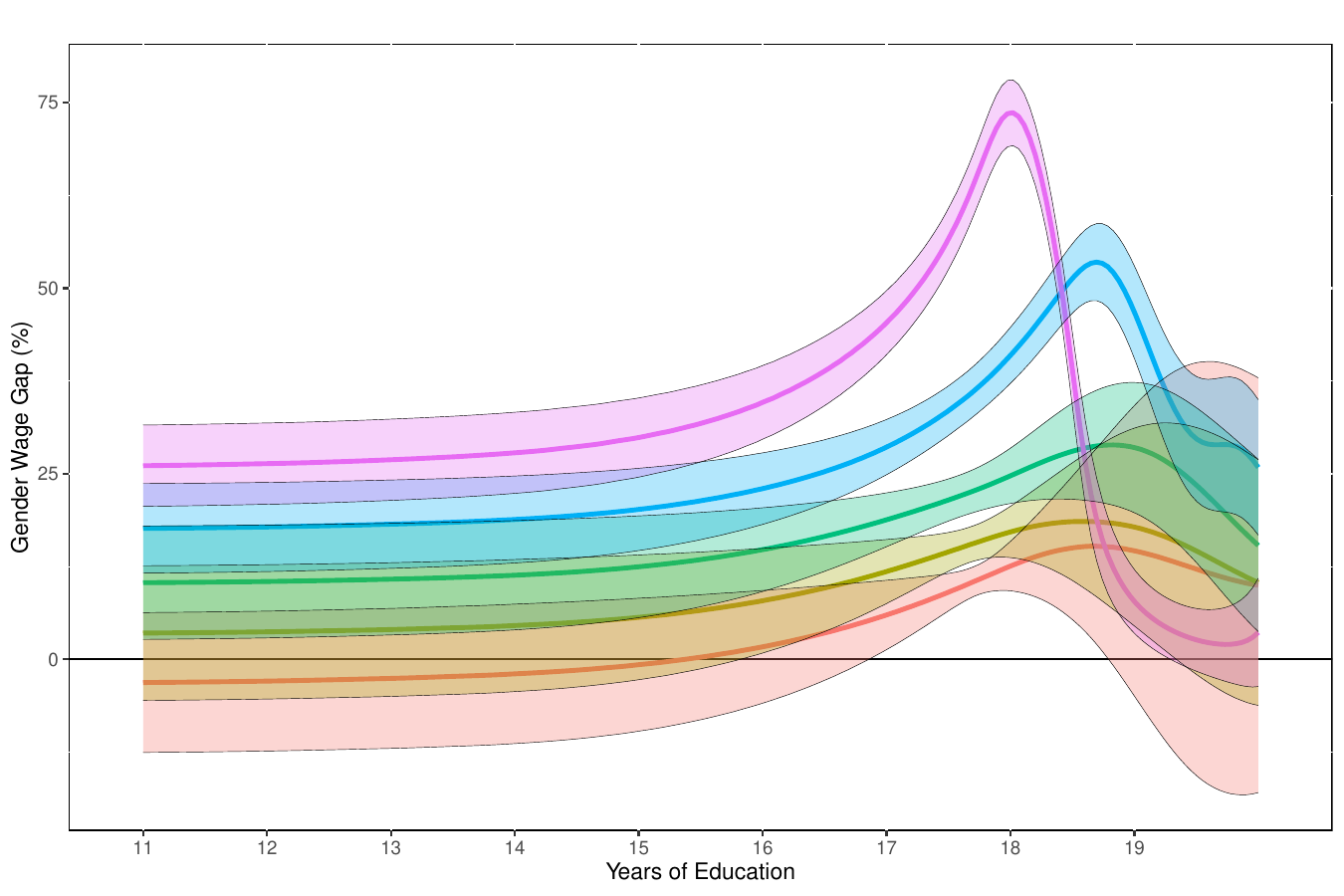}

\caption{Quantile gender wage gap.\label{fig:QTE}}
\end{figure}

\section{Extensions\label{sec:Section7}}

\subsection{Logistic transform regression}

This paper has exposited a theory of Gaussian transformation regression
based on the standard normal density. Other probability densities
could be employed instead, with the standard Logistic and Laplace
densities being appealing choices. We provide a brief analysis of
the Logistic case here.

We write as before (cf. (\ref{eq:e})) $e=b_{0}'T(X,Y)$, $\partial_{y}\{b_{0}'T(X,Y)\}=b_{0}'t(X,Y)>0$,
but replace $e\mid X\sim N(0,1)$ with $e\mid X\sim\Lambda$ where
$\Lambda$ is the standard logistic CDF with log-density $e-2\log(1+\exp(e))$.
Objective function (\ref{eq:popML}) is replaced by 
\[
Q(b)=E\left[b'T(X,Y)-2\log\left(1+\exp(b'T(X,Y))\right)+\log\left(b't(X,Y)\right)\right],\quad b\in\Theta.
\]
The corresponding first-order conditions that replace (\ref{eq:FOCs})
are
\[
E\left[-T(X,Y)\frac{\exp(b'T(X,Y)-1)}{\exp(b'T(X,Y)+1)}+\frac{t(X,Y)}{b't(X,Y)}\right]=0.
\]
The population objective function remains concave as in the Gaussian
case, since the expected Hessian is now given by:
\begin{equation}
-E\left[2\lambda(b'T(X,Y))\left\{ T(X,Y)T(X,Y)'+\frac{t(X,Y)t(X,Y)'}{\{b't(X,Y)\}^{2}}\right\} \right],\label{eq:hessian-1}
\end{equation}
where $\lambda(\cdot)$ is the standard Logistic density. In the Gaussian
version of (\ref{eq:hessian-1}) the term corresponding to $2\lambda(b'T(X,Y))$
is ``1'' because the log Gaussian density is a simple quadratic 
with coefficient $(-1/2)$ so that its second derivative is non-stochastic. 
Our limited experience with estimation based on the Logistic density 
suggests that it has no particular advantage or disadvantage with 
respect to the Gaussian.

\subsection{Mixed discrete-continuous outcomes}

Our formulation extends to the case where $Y$ has both a continuous component with support $\mathcal{Y}_{C}$ and a discrete component with support $\mathcal{Y}_{D}\equiv \{y_{1},\ldots,y_{M}\}$. 
For each $y\in\mathcal{Y}_{D}$, the PDF of $Y$ conditional on $X$ 
can be expressed as $\Pr[Y=y|X]=F_{Y|X}(y|X)-\lim_{z\rightarrow y^{-}}F_{Y|X}(z|X)$ (\citet{Mittelhammer:2013}, p. 67), where 
$\lim_{z\rightarrow y^{-}}$ is the limit as $z$ approaches $y$ from below.

For some specified continuous and strictly increasing CDF $F$ with derivative $f$, define $g(Y,X)\equiv F^{-1}(F_{Y|X}(Y|X))$. Then $f_{Y|X}(Y|X)$ can now be expressed as
\[
f_{Y|X}(Y|X) =\{F(g(Y,X))-\lim_{z\rightarrow Y^{-}}F(g(z,X))\}^{1(Y\in\mathcal{Y}_{D})}\,\{f(g(Y,X))\partial_{y}g(Y,X)\}^{1(Y\in\mathcal{Y}_{C})}.
\]
Let $D(Y)$ be a vector of dummy variables $D_{m}$, $m\in \mathcal{M}\equiv \{2,\ldots, M\}$, taking value one if outcome $y_{m}$ occurs and zero otherwise. For $Y\in \mathcal{Y}_{D}$ and $M \geq 2$, we take 
$g(Y,X)=d_{0}'T_{D}(X,Y)$, where $T_{D}(X,Y)=W(X)\otimes P(Y)$ with $P(Y)=(1,D(Y)')'$. When $M=1$, we set $P(Y)=1$. For $Y\in \mathcal{Y}_{C}$, we take 
$g(Y,X)=c_{0}'T_{C}(X,Y)$, where $T_{C}(X,Y)=W(X)\otimes Q(Y)$ with $Q(Y)=(1,C(Y)')'$, for $C(Y)$ a vector of known functions of $Y$. For $F(g(Y,X))$ to be a valid CDF model, we assume $(d_{0}',c_{0}')'$, and $C(Y)$ are such that $y\mapsto g(y,X)$ is right-continuous, 
so that $y \mapsto F(g(y,X))$ also is.

To illustrate, let $\mathcal{Y}_{C} = (y_{M},\infty)$,  
$C(Y)$ be such that $C(y_{M})=0$, $d_0=(d_{01}',\ldots, d_{0M}')'$, and $\delta_{m}(X)\equiv d_{0m}'W(X)$, $m\in \{1,\mathcal{M}\}$. For $Y \in \mathcal{Y}_{D}$,  $P(Y)=(1,D(Y)')'$, and hence:
\[
g(Y,X)=d_{0}'T_D(X,Y)=d_{01}'W(X) + \sum_{m\in \mathcal{M}}\{d_{0m}'W(X)\}D_{m}=\delta_{1}(X) +
\sum_{m\in  \mathcal{M}}\delta_{m}(X)D_{m}.
\]
Similarly, for $J\equiv \dim(Q)$, write $c_{0}=(c_{01}',\ldots, c_{0J}')'$ and let $\xi_{j}(X)\equiv c_{0j}'W(X)$, $j\in \{1,\mathcal{J}\}$, $\mathcal{J}\equiv \{2,\ldots,J\}$. 
For $Y \in \mathcal{Y}_{C}$, $Q(Y)=(1,C(Y)')'$, and hence: $g(Y,X)=\xi_{1}(X) + \sum_{j\in  \mathcal{J}}\xi_{j}(X)C_{j-1}(Y)$. 
By $\sum_{j\in  \mathcal{J}}\xi_{j}(X)C_{j-1}(y_{M})=0$ and $\sum_{m\in  \mathcal{M}}\delta_{m}(X)D_{m}=\delta_{M}(X)$ when $Y=y_{M}$, setting $\delta_1(X) + \delta_{M}(X) = \xi_{1}(X)$ implies right-continuity of $y\mapsto g(y,X)$ at $y_{M}$. 
If $E[W(X)W(X)']$ is nonsingular, this holds for $d_{01}+d_{0M}=c_{01}$.

\begin{rem}
For $M=1$ and $Q(Y)=(1,Y)'$, an important particular case is the Tobit model (\citet{Tobin:1958}). With censoring at 0, the standard model has $Y=W(X)'\beta+\sigma e$ if $W(X)'\beta+\sigma e>0$, $e|X\sim\Phi$, and $Y=0$ otherwise. For $Y>0$ this is the same as $e=-W(X)'(\beta/\sigma) + (1/\sigma)Y$, which is of the form $c_{0}'T_{C}(X,Y)$ above for  $F=\Phi$ and $c=(c_{01}', c_{02})'\equiv (-\beta'/\sigma, 1/\sigma)'$. This is \citet{Olsen:1978}'s concave proposal commonly used in practice, and our formulation provides a nonlinear and nonseparable generalization.
\begin{rem}
With $Y$ discrete, we have  $\lim_{z\rightarrow y_{1}^{-}}F_{Y|X}(z|X)=0$, $\lim_{z\rightarrow y_{m}^{-}}F_{Y|X}(z|X)=F_{Y|X}(y_{m-1}|X)$,  $m\in \mathcal{M}$, and $F_{Y|X}(y_M|X)=1$.  
Moreover, specifications with $S(Y)=(1,D(Y)')'$ are not restrictive in the $Y$ dimension, and coincide with the semiparametric distribution regression model $F_{Y|X}(Y|X)=F(W(X)'\beta(Y))$. 
In contrast with distribution regression,  monotonicity is obtained directly from 
the implied objective function, with the terms $\log(F(b'T(X,y_{m}))-F(b'T(X,y_{m-1})))$ 
ruling out nonincreasing CDFs. When also $X$ is discrete, with support $\{x_{1},\ldots,x_{L}\}$, a nonparametric formulation is achieved with $W(X)=(1,Z(X)')'$, for $Z(X)$ a vector of dummy variables $Z_{l}$, $l=2,\ldots,L$, taking value one when $X=x_{l}$ occurs and zero otherwise.

\end{rem}

\end{rem}

\subsection{Multiple outcomes}

With multiple outcomes $(Y_{1},\ldots,Y_{M})'\equiv Y$, $M\geq2$,
writing $\mathbb{Y}_{m}\equiv(Y_{1},\ldots,Y_{m})'$, a compact generalization
of (\ref{eq:e}) is the recursive formulation
\begin{align*}
e_{m} & =T_{m}(X,\mathbb{Y}_{m})'b_{0,m},\quad e_{m}\mid X,\mathbb{Y}_{m-1}\sim N(0,1),\quad m\in\{2,\ldots,M\},\\
e_{1} & =T_{1}(X,Y_{1})'b_{0,1},\quad e_{1}\mid X\sim N(0,1),
\end{align*}
where $T_{m}(X,\mathbb{Y}_{m})\equiv T_{m-1}(X,\mathbb{Y}_{m-1})\otimes S_{m}(Y_{m})$
and $T_{1}(X,Y_{1})\equiv W(X)\otimes S_{1}(Y_{1})$, with
\begin{align*}
\partial_{y_{m}}\{T_{m}(X,\mathbb{Y}_{m})'b_{0,m}\} & =t_{m}(X,\mathbb{Y}_{m})'b_{0,m}>0,\quad m\in\{1,\ldots,M\},
\end{align*}
where $t_{m}(X,\mathbb{Y}_{m})\equiv t_{m-1}(X,\mathbb{Y}_{m-1})\otimes s_{m}(Y_{m})$,
$m\in\{2,\ldots,M\}$, and $t_{1}(X,Y_{1})\equiv W(X)\otimes s_{1}(Y_{1})$.
By construction, the $e_{m}$'s are jointly Gaussian and mutually
independent, with variance-covariance the identity matrix. This is a Gaussian version of \citet{Rosenblatt:1952}'s multivariate
probability transformation. The conditional
CDF  is
\[
F_{Y\mid X}(y_{1},\ldots,y_{M}\mid X)=\int_{-\infty}^{y_{1}}\ldots\int_{-\infty}^{y_{M}}f_{Y\mid X}(t_{1},\ldots,t_{M}\mid X)dt_{1}\ldots dt_{M},
\]
where the conditional PDF takes the form
\[
f_{Y\mid X}(\overline{y}_{M}\mid X)=\prod_{m=1}^{M}\phi(T_{m}(X,\overline{y}_{m})'b_{0,m})\{t_{m}(X,\overline{y}_{m})'b_{0,m}\},\quad\overline{y}_{m}\equiv(y_{1},\ldots,y_{m})\in\mathbb{R}^{m}.
\]

\section{Conclusion\label{sec:Section8}}

The formulation of flexible models for the GT $e=g(Y,X)$ leads to
a unifying information-theoretic framework for the global estimation
of DRFs. The implied convex programming
formulation is easy to implement and the linear form of the proposed
GT regression models also constitutes a good starting point for nonparametric
estimation. In this paper we have considered a few extensions to our
original formulation such as misspecification, quantile treatment
effects, Logistic transform regression, 
mixed discrete-continuous distributions and multiple outcomes. 
Two important further extensions for future work are 
sample selection and endogenous regressors.

\appendix

\section{Proofs of Theorems \ref{thm:Thm1}-\ref{thm:Thm2}}

\subsection{Definitions and notation}

For $b\in\Theta$, define
\[
L(Y,X,b)\equiv-\frac{1}{2}[\log(2\pi)+(b'T(X,Y))^{2}]+\log(b't(X,Y)),
\]
and $f(Y,X,b)\equiv\phi(b'T(X,Y))\{b't(X,Y)\}$, and note that $Q(b)=E[L(Y,X,b)]=E[\log f(Y,X,b)]$.
Section \ref{sec:Auxiliary-Results} of the Supplemental Material
contains all lemmas used in the proofs of the theorems. 

\subsection{Proof of Theorem \ref{thm:Thm1}}

We show that $b_{0}$ is a point of maximum of $Q(b)$ in $\Theta$.
For $b\neq b_{0}$, $b\in\Theta$, by $E[\log f(Y,X,b_{0})]=E[\log f_{Y\mid X}(Y|X)]$
and Jensen's inequality:
\[
E\left[\log\left(\frac{f(Y,X,b_{0})}{f(Y,X,b)}\right)\right]\geq-\log E\left[\left(\frac{f(Y,X,b)}{f_{Y|X}(Y\mid X)}\right)\right]=-\log E\left[\int_{\mathbb{R}}f(y,X,b)dy\right]\geq0.
\]
The last inequality holds since, with probability one, 
\[
\int_{\mathbb{R}}f(y,X,b)dy=\lim_{y\rightarrow\infty}\Phi(b'T(X,y))-\lim_{y\rightarrow-\infty}\Phi(b'T(X,y))\in(0,1]
\]
by the properties of the Gaussian CDF and $y\mapsto\Phi(b'T(X,y))$
being strictly increasing for $b\in\Theta$. Therefore, $b_{0}$ is
a point of maximum. 
By Lemma \ref{lem:Concavity}, $Q(b)$ is strictly concave over any compact subset of $\Theta$, and hence admits at most one maximizer in every compact set that contains $b_{0}$.
Therefore, $b_{0}$ uniquely maximizes $Q(b)$ in $\Theta$. \qed

\subsection{Proof of Theorem \ref{thm:Thm2}}

\subsubsection{Proof of existence}

We first show that the level sets $\mathcal{B}_{\alpha}=\{b\in\Theta:-Q(b)\leq\alpha\}$,
$\alpha\in\mathbb{R}$, of $-Q(b)$ are bounded (Step 1) and closed
(Step 2), hence compact, and then use the fact that $-Q(b)$ is continuous
over $\Theta$, which implies existence 
(Step 3).

\textit{Step 1.} Given $b_{1},b_{2}\in\mathcal{B}_{\alpha}$, let
$t=||b_{1}-b_{2}||$ and $u=\frac{b_{1}-b_{2}}{||b_{1}-b_{2}||}$,
so that $||u||=1$ and $b_{1}=b_{2}+tu$. By Lemma \ref{lem:2Differentiability},
$Q(b)$ is twice continuously differentiable for $b\in\mathcal{B}_{\alpha}$.
Thus, by definition of $b_{1}$, a second-order Taylor expansion of
$t\mapsto-Q(b_{2}+tu)$ around $t=0$ yields, for some $\bar{b}$
on the line connecting $b_{1}$ and $b_{2}$ and some constant $B>0$,
\begin{align*}
\alpha\geq-Q(b_{1})=-Q(b_{2}+tu) & =-Q(b_{2})-t\nabla_{b}Q(b_{2})'u-\frac{t^{2}}{2}u'\nabla_{bb}Q(\bar{b})u\\
 & \geq-Q(b_{2})-t\nabla_{b}Q(b_{2})'u+B\frac{t^{2}}{2}\\
 & \geq-Q(b_{2})-t||\nabla_{b}Q(b_{2})||+B\frac{t^{2}}{2},
\end{align*}
where the penultimate inequality follows by Lemma \ref{lem:Concavity}.
Fixing $b_{2}\in\mathcal{B}_{\alpha}$, the above inequality implies
that $t$ is bounded and therefore $\mathcal{B}_{\alpha}$ is bounded.

\textit{Step 2. $\mathcal{B}_{\alpha}$ is closed.} Define the boundary
$\partial\Theta$ of $\Theta$ as
\[
\partial\Theta=\{b\in\mathbb{R}^{JK}:\Pr[b't(X,Y)=0]>0\}.
\]
For $b\in\partial\Theta$ with $b't(X,Y)<0$ on a set with positive
probability, we adopt the convention that the logarithmic barrier
function $\log(b't(X,Y))$ takes on the value $-\infty$ on that set
(e.g., Section 11.2.1 in \citet{BV:2004}). Consider a sequence $(b_{n})$
in $\mathcal{B}_{\alpha}$ such that $b_{n}\rightarrow\check{b}\in\partial\Theta$
as $n\rightarrow\infty$. Steps 2.1 and 2.2 below show that $-Q(b_{n})=E[-L(Y,X,b_{n})]\rightarrow\infty$
as $n\rightarrow\infty$, and hence that $\mathcal{B}_{\alpha}$ is
closed.

\textit{Step 2.1. We show that $E[\lim_{n\rightarrow\infty}-L(Y,X,b_{n})]\leq\lim_{n\rightarrow\infty}E[-L(Y,X,b_{n})]$.}
By $\mathcal{B}_{\alpha}$ being bounded, there exists a constant
$C>0$ such that $\log(b't(X,Y))\leq C||t(Y,X)||$ with probability
one for all $b\in\mathcal{B}_{\alpha}$, and hence such that
\[
-L(Y,X,b)=\frac{1}{2}[\log(2\pi)+(b'T(X,Y))^{2}]-\log(b't(X,Y))\geq-C||t(X,Y)||,\quad b\in\mathcal{B}_{\alpha},
\]
with probability one. Therefore,
\[
\varphi(Y,X,b)\equiv-L(Y,X,b)+\delta(Y,X)\geq0,\quad\delta(Y,X)\equiv C||t(X,Y)||,\quad b\in\mathcal{B}_{\alpha},
\]
with probability one, and where $E[|\delta(Y,X)|]<\infty$ under Assumption
\ref{ass: Ass2}. 

Moreover, by definition of $\partial\Theta$, we have that $\lim_{n\rightarrow\infty}\log(b_{n}'t(X,Y))=-\infty$
on a subset $\widetilde{\mathcal{YX}}$ of the joint support of $(Y,X)$
with positive probability, and hence
\begin{equation}
\lim_{n\rightarrow\infty}-L(Y,X,b_{n})=\frac{1}{2}[\log(2\pi)+\lim_{n\rightarrow\infty}\{b_{n}'T(X,Y)\}^{2}]-\lim_{n\rightarrow\infty}\log(b_{n}'t(X,Y))=\infty,\label{eq:limL}
\end{equation}
on $\widetilde{\mathcal{YX}}$, by $\{b'T(X,Y)\}^{2}/2\geq0$ for
all $b\in\mathbb{R}^{JK}$.

Letting $\chi_{\widetilde{\mathcal{YX}}}(Y,X)\equiv1\{(Y,X)\in\widetilde{\mathcal{YX}}\}$
and $\chi_{\widetilde{\mathcal{YX}}^{c}}(Y,X)\equiv1\{(Y,X)\in\widetilde{\mathcal{YX}}^{c}\}$,
with $\widetilde{\mathcal{YX}}^{c}$ denoting the complement of $\widetilde{\mathcal{YX}}$,
we have
\begin{align}
E[\lim_{n\rightarrow\infty}\varphi(Y,X,b_{n})] & =E[\chi_{\widetilde{\mathcal{YX}}}(Y,X)\lim_{n\rightarrow\infty}\varphi(Y,X,b_{n})]+E[\chi_{\widetilde{\mathcal{YX}}^{c}}(Y,X)\lim_{n\rightarrow\infty}\varphi(Y,X,b_{n})]\nonumber \\
 & =E[\chi_{\widetilde{\mathcal{YX}}}(Y,X)\lim_{n\rightarrow\infty}-L(Y,X,b_{n})]+E[\chi_{\widetilde{\mathcal{YX}}}(Y,X)\delta(Y,X)\}]\nonumber \\
 & +E[\chi_{\widetilde{\mathcal{YX}}^{c}}(Y,X)\lim_{n\rightarrow\infty}-L(Y,X,b_{n})]+E[\chi_{\widetilde{\mathcal{YX}}^{c}}(Y,X)\delta(Y,X)]\nonumber \\
 & =E[\lim_{n\rightarrow\infty}-L(Y,X,b_{n})]+E\left[\delta(Y,X)\right],\label{eq:LHS_Fatou}
\end{align}
where the second equality follows from $\lim_{n\rightarrow\infty}-L(Y,X,b_{n})$
and $\delta(Y,X)$ being nonnegative functions on $\widetilde{\mathcal{YX}}$
(e.g., Proposition 5.2.6(ii) in \citet{Rana:2002}), and 
having finite expectation on $\widetilde{\mathcal{YX}}^{c}$, 
by $\lim_{n\rightarrow\infty}b_{n}'t(X,Y)>0$ on that set 
and steps similar to those in the proof of Lemma \ref{lem:FiniteContinuityQ}.

By $\varphi(Y,X,b_{n})$ being nonnegative, Fatou's lemma implies
that
\begin{equation}
E[\lim_{n\rightarrow\infty}\varphi(Y,X,b_{n})]\leq\lim_{n\rightarrow\infty}E[\varphi(Y,X,b_{n})],\label{eq:Fatou}
\end{equation}
with
\begin{equation}
\lim_{n\rightarrow\infty}E[\varphi(Y,X,b_{n})]=\lim_{n\rightarrow\infty}E[-L(Y,X,b_{n})]+E[\delta(Y,X)],\label{eq:RHS_Fatou}
\end{equation}
by $E[|\delta(Y,X)|]<\infty$ and $E[|-L(Y,X,b_{n})|]<\infty$ for
$b_{n}\in\Theta$ by Lemma \ref{lem:FiniteContinuityQ}. Therefore,
\begin{equation}
E[\lim_{n\rightarrow\infty}-L(Y,X,b_{n})]\leq\lim_{n\rightarrow\infty}E[-L(Y,X,b_{n})]\label{eq:Fatou_L}
\end{equation}
follows by (\ref{eq:LHS_Fatou}), (\ref{eq:Fatou}) and (\ref{eq:RHS_Fatou}).

\begin{sloppy}\textit{Step 2.2. We show that $E[\lim_{n\rightarrow\infty}-L(Y,X,b_{n})]=\infty$.}
By (\ref{eq:limL}) and $f_{YX}(Y,X)>0$ 
on $\widetilde{\mathcal{YX}}$ a set with positive probability, we have 
$E[\chi_{\widetilde{\mathcal{YX}}}(Y,X)\lim_{n\rightarrow\infty}-L(Y,X,b_{n})]=\infty$,
and hence also $E[\chi_{\widetilde{\mathcal{YX}}}(Y,X)\lim_{n\rightarrow\infty}\varphi(Y,X,b_{n})]=\infty$ 
since $E[|\delta(Y,X)|]<\infty$. Moreover, $E[\chi_{\widetilde{\mathcal{YX}}^{c}}(Y,X)\lim_{n\rightarrow\infty}-L(Y,X,b_{n})]<\infty$,
and hence $E[\chi_{\widetilde{\mathcal{YX}}^{c}}(Y,X)\lim_{n\rightarrow\infty}\varphi(Y,X,b_{n})]<\infty$
by $E[|\delta(Y,X)|]<\infty$. Therefore $E[\lim_{n\rightarrow\infty}\varphi(Y,X,b_{n})]=\infty$,
and hence $E[\lim_{n\rightarrow\infty}-L(Y,X,b_{n})]=\infty$, by
(\ref{eq:LHS_Fatou}). This fact and (\ref{eq:Fatou_L}) together
imply $E[-L(Y,X,b_{n})]=-Q(b_{n})\rightarrow\infty$ as $n\rightarrow\infty$.
Hence the limit $\check{b}$ of a convergent sequence $(b_{n})$ in
$\mathcal{B}_{\alpha}$ is in $\Theta$. By continuity of $-Q(b)$
over $\Theta$, we have that $-Q(\check{b})=\lim_{n\rightarrow\infty}Q(b_{n})\leq\alpha$,
and hence $\check{b}\in\mathcal{B}_{\alpha}$ and $\mathcal{B}_{\alpha}$
is closed.\par\end{sloppy}

\textit{Step 3. }For each $\alpha\in\mathbb{R}$, Steps 1-2 imply
that $\mathcal{B}_{\alpha}$ is compact, by the Heine-Borel theorem,
and hence $Q(b)$ continuous implies that $-Q(b)$ has a minimizer
in $\mathcal{B}_{\alpha}$, by the Weierstrass theorem. Pick $\overline{\alpha}$
such that $\mathcal{B}_{\overline{\alpha}}$ is nonempty, and let
$\overline{b}$ denote a minimizer of $-Q(b)$ over $\mathcal{B}_{\overline{\alpha}}$.
By $\overline{b}\in\mathcal{B}_{\overline{\alpha}}$, we have $-Q(\overline{b})\leq\overline{\alpha}$.
Hence for any $b\in\Theta\backslash\mathcal{B}_{\overline{\alpha}}$,
we have $-Q(b)>\overline{\alpha}\geq-Q(\overline{b})$. Therefore
$\overline{b}$ is a minimizer of $-Q(b)$ in $\Theta$.

\subsubsection{Proof of uniqueness}

Uniqueness of $b^{*}$ follows by concavity of $Q(b)$ in Lemma \ref{lem:Concavity}.
By nonsingularity of $E[T(X,Y)T(X,Y)']$, for $\widetilde{b}\neq b^{*}$,
\[
E[\{(\widetilde{b}-b^{*})'T(X,Y)\}^{2}]=(\widetilde{b}-b^{*})'E[T(X,Y)T(X,Y)'](\widetilde{b}-b^{*})>0,
\]
which implies $(\widetilde{b}-b^{*})'T(X,Y)\neq0$. Therefore, $g^{*}(Y,X)\neq\widetilde{m}(Y,X)$
for $\widetilde{m}\in\mathcal{E}$ with $\widetilde{b}\neq b^{*}$,
by definition of $\mathcal{E}$. By strict monotonicity of $e\mapsto\Phi(e)$,
this also implies that $\Phi(g^{*}(Y,X))\neq\Phi(\widetilde{m}(Y,X))$,
and hence $F^{*}(Y,X)\neq\widetilde{F}(Y,X)$ for $\widetilde{F}\in\mathcal{F}$
with $\widetilde{m}\neq g^{*}$, by definition of $\mathcal{F}$.
For $m\in\mathcal{E}$, let $\widetilde{\mathcal{Y}}_{x}(m)\equiv\{y\in\mathcal{Y}_{x}:F^{*}(y,x)\neq\Phi(m(y,x))\}$,
with $\mathcal{Y}_{x}$ denoting the support of $Y$ conditional on 
$X=x$, and $\widetilde{\mathcal{U}}_{x}(m)\equiv\{u\in(0,1):F^{*}(y,x)=u\textrm{ for some }y\in\widetilde{\mathcal{Y}}_{x}(m)\}$.
With probability one, by strict monotonicity of $y\mapsto m(y,X)$
for all $m\in\mathcal{E}$, the composition $y\mapsto\Phi(\widetilde{m}(y,X))$
is also strictly monotone, which implies that $Q^{*}(u,X)\neq\widetilde{m}^{-1}(\Phi^{-1}(u),X)$, 
$u\in\widetilde{\mathcal{U}}_{X}(\widetilde{m})$, 
and hence $Q^{*}(u,X)\neq\widetilde{Q}(u,X)$ for $\widetilde{Q}\in\mathcal{Q}$ 
with $\widetilde{m}\neq g^{*}$, by definition of $\mathcal{Q}$. 
Finally, by $b^{*}$ being the unique maximizer of $Q(b)$ in $\Theta$, 
we have 
$E[\log f^{*}(Y,X)]>E[\log(\phi(\widetilde{m}(Y,X))\{\partial_{y}\widetilde{m}(Y,X)\}]$, 
and hence $f^{*}(Y,X)\neq\widetilde{f}(Y,X)$ for $\widetilde{f}\in\mathcal{D}$
with $\widetilde{m}\neq g^{*}$, by definition of $\mathcal{D}$.\qed

\section{Proof of Theorem \ref{thm:Thm6}\label{sec:Duality-Theory}}

\subsubsection*{Part (i)}

Write $T_{i}\equiv T(y_{i},x_{i})$ and $t_{i}\equiv t(y_{i},x_{i})$,
for $i\in\{1,\ldots,n\}$, and let $\mathbb{R}_{-}\equiv(-\infty,0)$
and $\mathbb{R}_{+}\equiv(0, \infty)$. Introducing the variables
$e_{i}=b'T_{i}$, $\eta_{i}=b't_{i}$, an equivalent formulation for
the GT regression problem is
\begin{align*}
\max_{(b,e,\eta)\in\Theta\times\mathbb{R}^{n}\times\mathbb{R}_{+}^{n}} & n\kappa-\sum_{i=1}^{n}\left\{ \frac{e_{i}^{2}}{2}-\log(\eta_{i})\right\} ,\quad\kappa\equiv-\frac{1}{2}\log(2\pi),\\
\textrm{subject to}\quad & e_{i}=b'T_{i},\quad\eta_{i}=b't_{i},\quad i\in\{1,\ldots,n\}.
\end{align*}
For all $(u,v)\in\mathbb{R}^{n}\times\mathbb{R}_{-}^{n}$, define
the Lagrange function for this problem as
\[
\mathscr{L}(b,e,\eta,u,v)=n\kappa-\sum_{i=1}^{n}\left\{ \frac{e_{i}^{2}}{2}-\log(\eta_{i})\right\} +\sum_{i=1}^{n}u_{i}\left\{ e_{i}-b'T_{i}\right\} +\sum_{i=1}^{n}v_{i}\left\{ \eta_{i}-b't_{i}\right\} ,
\]
and the Lagrange dual function (\citet{BV:2004}, Chapter 5) as
\begin{align*}
g(u,v) & \equiv\sup_{(b,e,\eta)\in\Theta\times\mathbb{R}^{n}\times\mathbb{R}_{+}^{n}}\mathscr{L}(b,e,\eta,u,v)\\
 & =\sup_{(e,\eta)\in\mathbb{R}^{n}\times\mathbb{R}_{+}^{n}}\sum_{i=1}^{n}\left\{ u_{i}e_{i}+v_{i}\eta_{i}-\left[-\kappa+\frac{e_{i}^{2}}{2}-\log(\eta_{i})\right]\right\} \\
 & +\sup_{b\in\Theta}\left\{ -\sum_{i=1}^{n}u_{i}(b'T_{i})-\sum_{i=1}^{n}v_{i}(b't_{i})\right\} \equiv\textrm{\ensuremath{I_{1}}}+\textrm{\ensuremath{I_{2}}}.
\end{align*}
In order to derive $g(u,v)$ we first show that for all $(u,v)\in\mathbb{R}^{n}\times\mathbb{R}_{-}^{n}$
the maximum of the mapping $(b,e,\eta)\mapsto\mathscr{L}(b,e,\eta,u,v)$
is attained and is unique, and we then evaluate $(b,e,\eta)\mapsto\mathscr{L}(b,e,\eta,u,v)$
at this value.

The first term $\textrm{I}_{1}$ in the dual function $g(u,v)$ is
the convex conjugate of the negative log-likelihood function, defined
as a function of the $n$-vectors $e$ and $\eta$. Define
\[
\mathcal{D}(e,\eta,u,v)\equiv\sum_{i=1}^{n}\left\{ u_{i}e_{i}+v_{i}\eta_{i}\right\} -\sum_{i=1}^{n}\left\{ -\kappa+\frac{e_{i}^{2}}{2}-\log(\eta_{i})\right\} .
\]
We first show that, for all $(u,v)\in\mathbb{R}^{n}\times\mathbb{R}_{-}^{n}$,
the map $(e,\eta)\mapsto\mathcal{D}(e,\eta,u,v)$ admits at least
one maximum in $\mathbb{R}^{n}\times\mathbb{R}_{+}^{n}$. Solving
the first-order conditions for $e_{i}$ and $\eta_{i}$ gives
\begin{equation}
e_{i}=u_{i},\quad\eta_{i}=-\frac{1}{v_{i}},\quad i\in\{1,\ldots,n\}.\label{eq:FOCs-1}
\end{equation}
Clearly, for all $(u,v)\in\mathbb{R}^{n}\times\mathbb{R}_{-}^{n}$
there exists $(e,\eta)\in\mathbb{R}^{n}\times\mathbb{R}_{+}^{n}$
such that (\ref{eq:FOCs-1}) holds.

We now show that, for all $(u,v)\in\mathbb{R}^{n}\times\mathbb{R}_{-}^{n}$,
the map $(e,\eta)\mapsto\mathcal{D}(e,\eta,u,v)$ admits at most one
maximum in $\mathbb{R}^{n}\times\mathbb{R}_{+}^{n}$. For $i\in\{1,\ldots,n\}$,
the second-order conditions are 
\begin{align*}
\partial_{e_{i},e_{i}}^{2}\mathcal{D}(e,\eta,u,v)=-1, & \qquad\partial_{e_{i},\eta_{i}}^{2}\mathcal{D}(e,\eta,u,v)=0\\
\partial_{\eta_{i},e_{i}}^{2}\mathcal{D}(e,\eta,u,v)=0, & \qquad\partial_{\eta_{i},\eta_{i}}^{2}\mathcal{D}(e,\eta,u,v)=-\frac{1}{\eta_{i}^{2}}.
\end{align*}
Therefore the Hessian matrix of $(e,\eta)\mapsto\mathcal{D}(e,\eta,u,v)$
is negative definite for all $(u,v)\in\mathbb{R}^{n}\times\mathbb{R}_{-}^{n}$.
Hence, $(e,\eta)\mapsto\mathcal{D}(e,\eta,u,v)$ is strictly concave
with unique maximum $(e_{i},\eta_{i})=(u_{i},-1/v_{i})$, $i\in\{1,\ldots,n\}$,
for all $(u,v)\in\mathbb{R}^{n}\times\mathbb{R}_{-}^{n}$. Evaluating
$(e,\eta)\mapsto\mathcal{D}(e,\eta,u,v)$ at the maximum yields, for
all $(u,v)\in\mathbb{R}^{n}\times\mathbb{R}_{-}^{n}$,
\begin{equation}
\sup_{(e,\eta)\in\mathbb{R}^{n}\times\mathbb{R}_{+}^{n}}\mathcal{D}(e,\eta,u,v)=-n(1-\kappa)+\sum_{i=1}^{n}\left\{ \frac{u_{i}^{2}}{2}-\log\left(-v_{i}\right)\right\} ,\label{eq:ll_conjugate}
\end{equation}
the conjugate function of the negative log-likelihood.

We now consider the second term $\textrm{I}_{2}$ in the definition
of the dual function $g(u,v)$. For all $(b,u,v)\in\Theta\times\mathbb{R}^{n}\times\mathbb{R}_{-}^{n}$,
define the penalty function
\[
\mathcal{P}(b,u,v)=\sum_{i=1}^{n}\left\{ -u_{i}(b'T_{i})-v_{i}(b't_{i})\right\} .
\]
The map $b\mapsto\mathcal{P}(b,u,v)$ is linear with partial derivative
$-\sum_{i=1}^{n}\left\{ u_{i}T_{i}+v_{i}t_{i}\right\} $. The value
of $\sup_{b\in\Theta}\mathcal{P}(b,u,v)$ is thus determined by the
set of all $(u,v)\in\mathbb{R}^{n}\times\mathbb{R}_{-}^{n}$ such
that the first-order conditions,
\begin{equation}
\nabla_{b}\mathcal{P}(b,u,v)=-\sum_{i=1}^{n}\left\{ u_{i}T_{i}+v_{i}t_{i}\right\} =0,\label{eq:P(b)}
\end{equation}
hold. For all such $(u,v)\in\mathbb{R}^{n}\times\mathbb{R}_{-}^{n}$
and any solution $\overline{b}\in\Theta$, we have that
\[
\sup_{b\in\Theta}\mathcal{P}(b,u,v)=\sum_{i=1}^{n}\left\{ -u_{i}(\overline{b}'T_{i})-v_{i}(\overline{b}'t_{i})\right\} =-\overline{b}'\sum_{i=1}^{n}\left\{ u_{i}T_{i}+v_{i}t_{i}\right\} =0.
\]
Therefore, for all $(u,v)\in\mathbb{R}^{n}\times\mathbb{R}_{-}^{n}$
such that $\nabla_{b}\mathcal{P}(\overline{b},u,v)=0$, the optimal
value of $\mathcal{P}(\overline{b},u,v)$ is $0$.

Combining (\ref{eq:ll_conjugate}) and (\ref{eq:P(b)}) gives the
Lagrange dual function $g(u,v)$ for all $(u,v)$ such that $\nabla_{b}\mathcal{P}(b,u,v)=0$.
The form of the dual problem (\ref{eq:Dual scores}) follows.

\subsubsection*{Part (ii)}

The Lagrangian for (\ref{eq:Dual scores}) is
\[
\mathcal{L}(u,v,b)=-n(1-\kappa)+\sum_{i=1}^{n}\left\{ \frac{u_{i}^{2}}{2}-\log(-v_{i})\right\} -b'\sum_{i=1}^{n}\left\{ T_{i}u_{i}+t_{i}v_{i}\right\} ,
\]
with first-order conditions
\[
\partial_{u_{i}}\mathcal{L}(u,v,b)=u_{i}-b'T_{i}=0,\quad\partial_{v_{i}}\mathcal{L}(u,v,b)=-\frac{1}{v_{i}}-b't_{i}=0,\quad i\in\{1,\ldots,n\},
\]
equivalently, upon solving for $u_{i}$ and $v_{i}$,
\begin{align}
u_{i} & =b'T_{i},\quad v_{i}=-\frac{1}{b't_{i}},\quad i\in\{1,\ldots,n\}.\label{eq:D-FOC}
\end{align}
Upon substituting in the constraints of (\ref{eq:Dual scores}) we
obtain representation (\ref{eq:MMrep}).

\subsubsection*{Part (iii)}

We show existence of  
$\widehat{b}\in\Theta$ 
in the proof
of Theorem \ref{thm:Thm4}(i). Uniqueness with probability approaching one follows by the sample Hessian matrix 
$-\Sigma_{i=1}^{n}\{T_{i}T_{i}'+t_{i}t_{i}'/(b't_{i})^{2}\}$ 
being negative definite with probability approaching one by Lemma
\ref{lem:Nonsingular} in the Supplemental Material.

Existence of a solution $(\widehat{u}',\widehat{v}')'$ to the dual
problem (\ref{eq:Dual scores}) follows from existence of a solution
$\widehat{b}\in\Theta$ to the first-order conditions of the primal
(\ref{eq:ML}) and the method-of-moments representation of (\ref{eq:Dual scores}),
upon setting $\widehat{u}_{i}=\widehat{b}'T_{i}$, $\widehat{v}_{i}=-1/(\widehat{b}'t_{i})$,
for $i\in\{1,\ldots,n\}$. We now show that, for all $b\in\Theta$,
the map $(u,v)\mapsto\mathcal{L}(u,v,b)$ admits at most one minimum 
in $\mathbb{R}^{n}\times\mathbb{R}_{-}^{n}$. For all $(u,v)\in\mathbb{R}^{n}\times\mathbb{R}_{-}^{n}$
and $i\in\{1,\ldots,n\}$, the second-order conditions for (\ref{eq:Dual scores})
are
\begin{align*}
\partial_{u_{i},u_{i}}^{2}\mathcal{L}(u,v,b)=1, & \qquad\partial_{u_{i},v_{i}}^{2}\mathcal{L}(u,v,b)=0\\
\partial_{v_{i},u_{i}}^{2}\mathcal{L}(u,v,b)=0, & \qquad\partial_{v_{i},v_{i}}^{2}\mathcal{L}(u,v,b)=\frac{1}{v_{i}^{2}}.
\end{align*}
Therefore, the Hessian matrix of $(u,v)\mapsto\mathcal{L}(u,v,b)$
is positive definite for all $b\in\Theta$. Hence, the map $(u,v)\mapsto\mathcal{L}(u,v,b)$
is strictly convex with unique solution $(\widehat{u}',\widehat{v}')'$.

\subsubsection*{Part (iv)}

Using (\ref{eq:D-FOC}), $\widehat{e}_{i}=\widehat{b}'T_{i}$ and
$\widehat{\eta}_{i}=\widehat{b}'t_{i}$, $i\in\{1,\ldots,n\}$, the
value of (\ref{eq:Dual scores}) is
\[
\mathcal{L}(\widehat{u},\widehat{v},\widehat{b})=-n(1-\kappa)+\sum_{i=1}^{n}\left\{ \frac{\widehat{e}_{i}^{2}}{2}+\log\left(\widehat{\eta}_{i}\right)\right\} -\sum_{i=1}^{n}\left\{ \widehat{e}_{i}^{2}-1\right\} =n\kappa-\sum_{i=1}^{n}\left\{ \frac{\widehat{e}_{i}^{2}}{2}-\log\left(\widehat{\eta}_{i}\right)\right\} ,
\]
the value of the ML problem (\ref{eq:ML}) at a solution. 
\qed

\newpage

\setcounter{section}{0}

\part*{Supplement to 
 ``Gaussian Transforms Modeling and the Estimation of Distributional
Regression Functions''}

\section{Summary}

In Section \ref{sec:Auxiliary-Results} of this Supplementary Material
we collect auxiliary results used in the proofs of our main results, Sections \ref{sec:Cor1Thm1} 
and \ref{sec:Asymptotic-Theory} contain proofs for Corollary \ref{cor:Cor1} and Theorems
\ref{thm:Thm3}-\ref{thm:Thm5}. In Section \ref{sec:Implementation}
we give implementation details and additional results for the empirical application. 
To assess the finite sample performance of
our estimator, Section \ref{sec:Numerical-Simulations1} gives results
of Monte Carlo simulations. We compare our Gaussian Transform Regression
(GTR) estimator to related methods for the estimation of distributional
regression functions. Overall, we find that GTR performs very well
in finite samples.

\section{Auxiliary Results\label{sec:Auxiliary-Results}}

\begin{lem}
If Assumption \ref{ass: Ass2} holds then $E[|L(Y,X,b)|]<\infty$
and $Q(b)$ is continuous over $\Theta$.\label{lem:FiniteContinuityQ}
\end{lem}
\begin{sloppy}
\begin{proof}
By the triangle inequality,
\[
E[|L(Y,X,b)|]\leq\frac{1}{2}E[|(b'T(X,Y))^{2})|]+E[|\log(b't(X,Y))|]+\frac{1}{2}\log(2\pi).
\]
The first term $E[|(b'T(X,Y))^{2}|/2]$ is finite by Cauchy-Schwartz
inequality and by $E[||T(X,Y)||^{2}]<\infty$. For the second term,
applying a mean-value expansion around $\bar{b}=(b_{0},0_{JK-1})$,
$b_{0}>0$, gives for some intermediate values $\tilde{b}$,
\[
|\log(b't(X,Y))|\leq|\log(b_{0})|+|(\tilde{b}'t(X,Y))^{-1}|\,||b-\bar{b}||\,||t(X,Y)||.
\]
Thus $E[|\log(b't(X,Y))|]<\infty$, since we have that $\tilde{b}'t(X,Y)>0$
with probability one and $E[||t(X,Y)||]<\infty$. Therefore $E[|L(Y,X,b)|]<\infty$.
Continuity of $Q(b)$ then follows from continuity of $b\mapsto L(Y,X,b)$
and dominated convergence. 
\end{proof}
\par\end{sloppy}

\begin{lem}
If Assumption \ref{ass: Ass2} holds then $Q(b)$ is twice continuously
differentiable over any compact subset $\overline{\Theta}\subset\Theta$,
and $\nabla_{bb}E[L(Y,X,b)]=E[\nabla_{bb}L(Y,X,b)]$.\label{lem:2Differentiability}
\end{lem}
\begin{sloppy}
\begin{proof}
By Lemma \ref{lem:FiniteContinuityQ}, $E[|L(Y,X,b)|]<\infty$. Moreover,
for $b\in\overline{\Theta}$,
\begin{align}
||\nabla_{b}L(Y,X,b)|| & =||-T(X,Y)(b'T(X,Y))+(b't(X,Y))^{-1}t(X,Y)||\nonumber \\
 & \leq||T(X,Y)(b'T(X,Y))||+|(b't(X,Y))^{-1}|\,||t(X,Y)||\nonumber \\
 & \leq C\left\{ ||T(X,Y)||^{2}+||t(X,Y)||\right\} ,\label{eq:Bounds1-1}
\end{align}
for some finite constant $C>0$. Therefore, $E[||T(X,Y)||^{2}]<\infty$
and $E[||t(X,Y)||]<\infty$ imply that $E[\sup_{b\in\overline{\Theta}}||\nabla_{b}L(Y,X,b)||]<\infty$. 
Lemma 3.6 in \citet{Newey:McFadden:1994}
then implies that $Q(b)$ is continuously differentiable in $b$,
and that the order of differentiation and integration can be interchanged.

Continuous differentiability of $\nabla_{b}Q(b)$ in $b\in\overline{\Theta}$
follows from applying steps similar to (\ref{eq:Bounds1-1}). By $||\nabla_{bb}L(Y,X,b)||\leq||T(X,Y)||^{2}+C||t(X,Y)||^{2}$,
for some finite constant $C>0$, we have that $E[||T(X,Y)||^{2}]<\infty$
and $E[||t(X,Y)||^{2}]<\infty$ imply that $E[\sup_{b\in\overline{\Theta}}||\nabla_{bb}L(Y,X,b)||]<\infty$. 
Lemma 3.6 in \citet{Newey:McFadden:1994}
then implies that $\nabla_{bb}Q(b)$ is continuously differentiable
in $b$, and that the order of differentiation and integration can
be interchanged.
\end{proof}
\par\end{sloppy}

\begin{lem}
If Assumption \ref{ass: Ass2} holds then, for any compact subset
$\overline{\Theta}\subset\Theta$, we have that $-\nabla_{bb}Q(b)$
exists for $b\in\overline{\Theta}$, with smallest eigenvalue bounded
away from zero uniformly in $b\in\overline{\Theta}$.\label{lem:Concavity}
\end{lem}
\begin{proof}
By Lemma \ref{lem:2Differentiability}, $Q(b)$ is twice continuously
differentiable over $\overline{\Theta}$ and the order of differentiation
and integration can be interchanged. Therefore,
\[
\nabla_{bb}\{-Q(b)\}=\Gamma_{1}+\Gamma_{2}(b),\;\Gamma_{1}\equiv E\left[T(X,Y)T(X,Y)'\right],\;\Gamma_{2}(b)\equiv E\left[\frac{t(X,Y)t(X,Y)'}{(b't(X,Y))^{2}}\right],
\]
exists for all $b\in\overline{\Theta}$. 
Denoting the smallest eigenvalue of a matrix $A$ by $\lambda_{\min}(A)$,
the result then follows from Weyl's Monotonicity Theorem (e.g., Corollary
4.3.12 in \citet{HJ:2012}) which implies
\[
\lambda_{\min}(\Gamma_{1}+\Gamma_{2}(b))\geq\lambda_{\min}(\Gamma_{1})\geq B,\quad b\in\overline{\Theta},
\]
for some constant $B>0$, by $\Gamma_{2}(b)$ being positive semidefinite
for all $b\in\overline{\Theta}$ and the smallest eigenvalue of $E\left[T(X,Y)T(X,Y)'\right]$
being bounded away from zero.
\end{proof}

\begin{lem}\label{lem:SetEquiv}
If the boundary conditions (\ref{eq:Boundary conditions})
hold for all $b\in\Theta$ with probability one, then the sets $\Theta$
and $\mathcal{D}$ are equivalent.
\end{lem}
\begin{proof}
Recall that two sets $\mathcal{A}$ and $\mathcal{B}$ are equivalent
if there is a one-to-one correspondence between them, i.e., if there
exists some function $\varphi:\mathcal{A}\rightarrow\mathcal{B}$
that is both one-to-one and onto. The two sets then have the same
cardinality (\citet{Dudley:2002}).

We note that by nonsingularity of $E[T(X,Y)T(X,Y)']$ the two sets
$\Theta$ and $\mathcal{E}$ are equivalent. Hence it suffices to
show that $\mathcal{E}$ and $\mathcal{D}$ are equivalent. For each
$f\in\mathcal{D}$, $m\in\mathcal{E}$, and $(y,x)\in\mathcal{YX}$,
we define
\[
(\varphi(f))(y,x)\equiv\Phi^{-1}\left(\int_{-\infty}^{y}f(t,x)dt\right),\quad(\psi(m))(y,x)\equiv\partial_{y}\Phi(m(y,x)).
\]
We first verify that $\varphi:\mathcal{D}\rightarrow\mathcal{F}$
and $\psi:\mathcal{F}\rightarrow\mathcal{D}$, and then establish
that $\varphi$ is one-to-one and onto by showing that $\varphi$
and $\psi$ are inverse functions of each other.

By definition of $f\in\mathcal{D}$, the Fundamental Theorem of Calculus,
and the boundary conditions (\ref{eq:Boundary conditions}), we have
\begin{align*}
(\varphi(f))(y,X) & =\Phi^{-1}\left(\int_{-\infty}^{y}\phi(T(X,v)'b)\{b't(X,v)\}dv\right)\\
 & =\Phi^{-1}\left(\Phi(T(X,y)'b)-\lim_{\alpha\rightarrow-\infty}\Phi(T(X,\alpha)'b)\right)=T(X,y)'b
\end{align*}
for some $b\in\Theta$ and all $y\in\mathcal{Y}$, and hence $T(X,Y)'b\in\mathcal{E}$.
Therefore $\varphi:\mathcal{D}\rightarrow\mathcal{E}$. By definition
of $m\in\mathcal{E}$ we have
\[
(\psi(m))(y,X)=\partial_{y}\Phi(T(X,y)'b)=\phi(T(X,y)'b)\{t(X,y)'b\},\quad y\in\mathcal{Y},
\]
for some $b\in\Theta$, 
and hence $\phi(T(X,Y)'b)\{t(X,Y)'b\}\in\mathcal{D}$.
Therefore $\varphi:\mathcal{E}\rightarrow\mathcal{D}$.

The conclusion then follows from $\psi$ being both the left-inverse
of $\varphi$, since
\[
(\psi(\varphi(f)))(y,X)=\partial_{y}\left\{ \Phi\left(\Phi^{-1}\left(\int_{-\infty}^{y}f(t,X)dt\right)\right)\right\} =\partial_{y}\left\{ \int_{-\infty}^{y}f(t,X)dt\right\} =f(y,X)
\]
for all $y\in\mathcal{Y}$, and the right-inverse of $\varphi$, since
\[
(\varphi(\psi(m)))(y,X)=\Phi^{-1}\left(\int_{-\infty}^{y}\partial_{y}\Phi(m(t,X))dt\right)=m(y,X),\quad y\in\mathcal{Y}.
\]
Therefore, $\psi$ is the inverse function of $\varphi$ and the result
follows.
\end{proof}

\begin{lem} \label{lem:Nonsingular}
If $\{(y_{i},x_{i})\}_{i=1}^{n}$ is i.i.d. and $E[T(X,Y)T(X,Y)']$
is nonsingular then $\sum_{i=1}^{n}T_{i}T_{i}'$ is nonsingular with
probability approaching one, for $T_{i}\equiv T(y_{i},x_{i})$.
\end{lem}

\begin{proof}
We want to show that $\lambda'\{\sum_{i=1}^{n}T_{i}T_{i}'\}\lambda > 0$ for all $\lambda \neq 0$, with probability approaching one. 
Since for any $c\in \mathbb{R}$ we have $(c \lambda)'  \{ \sum_i T_i T_i' \} (c \lambda) = c^2  \lambda \{\sum_i T_i T_i' \} \lambda$, it suffices to consider $\lambda$ on the unit sphere 
$\mathbb{S}^{JK-1}=\{\lambda \in \mathbb{R}^{JK} : \|\lambda\| = 1\}$. Let $\Lambda \subset \mathbb{S}^{JK-1}$ be a countable dense subset. 
By nonsingularity of $E[T(X,Y)T(X,Y)']$,
for each $\lambda \in \Lambda$ we have $\lambda'T(X,Y)\neq0$ on a set $\widetilde{\mathcal{YX}}$ 
that depends on $\lambda$ and 
with $\Pr[\widetilde{\mathcal{YX}}]>0$. Hence for $\{(y_{i},x_{i})\}_{i=1}^{n}$
i.i.d.,
\begin{align}
\Pr[\cap_{i\in\{1,\ldots,n\}}\{(y_{i},x_{i})\notin\widetilde{\mathcal{YX}}\}]&=\prod_{i=1}^{n}\Pr[(y_{i},x_{i})\notin\widetilde{\mathcal{YX}}] \nonumber \\
&=\prod_{i=1}^{n}(1-\Pr[\widetilde{\mathcal{YX}}])=(1-\Pr[\widetilde{\mathcal{YX}}])^{n}\rightarrow0, \label{eq:lamxneq0}
\end{align}
as $n\rightarrow\infty$. 
For each $\lambda \in \Lambda$, define $A_\lambda$ as the event $\{\lambda'T_i \neq 0 \textrm{ for some }i\in\{1,\ldots,n\}\}$, and let $A= \cap_{\lambda\in \Lambda} A_{\lambda}$.  As $n \rightarrow \infty$, we have
\[
\Pr[A^c] = \Pr[ \cup_{\lambda \in \Lambda} A^c_{\lambda} ] \leq \sum_{\lambda \in \Lambda} \Pr[A^c_{\lambda}] \rightarrow 0,
\]
since $\Pr[A_{\lambda}^c] \rightarrow 0$ for each $\lambda \in \Lambda$ by (\ref{eq:lamxneq0}), and hence $\Pr[A] \rightarrow 1$. 
On the event $A$, we have $\lambda'\{\sum_{i=1}^{n}T_{i}T_{i}'\}\lambda > 0$ for all $\lambda \in \Lambda$. By continuity of $\lambda \mapsto \lambda'\{\sum_{i=1}^{n}T_{i}T_{i}'\}\lambda$ and $\Lambda$ being dense in $\mathbb{S}^{JK-1}$, we also have $\lambda'\{\sum_{i=1}^{n}T_{i}T_{i}'\}\lambda > 0$ for all $\lambda \in \mathbb{S}^{JK-1}$. Therefore, $\sum_{i=1}^{n}T_{i}T_{i}'$ is nonsingular with probability approaching one. 
\end{proof}

\section{Proofs of Corollary \ref{cor:Cor1} and Theorem \ref{thm:Thm3}} \label{sec:Cor1Thm1}

\subsection{Proof of Corollary \ref{cor:Cor1}}

The proof follows by application of Theorem \ref{thm:Thm1} and by
the argument in the main text, using that $e=b_{0}'T(Y,X)|X\sim N(0,1)$.
\qed

\subsection{Proof of Theorem \ref{thm:Thm3}}

\begin{sloppy}We have $b^{*}=\arg\max_{b\in\Theta}E[\log(\phi(T(X,Y)'b)\{t(X,Y)'b\})]$, by Theorem \ref{thm:Thm2}.
Thus, $f(Y,X)=\phi(T(X,Y)'b)\{t(X,Y)'b\}\in\mathcal{D}$ for each
$b\in\Theta$, and the fact that $\Theta$ and $\mathcal{D}$ are
equivalent by Lemma \ref{lem:SetEquiv}, together imply that $f^{*}$
is the well-defined point of maximum of $E[\log f(Y,X)]$ in $\mathcal{D}$,
and hence
\begin{equation}
f^{*}=\arg\min_{f\in\mathcal{D}}-E\left[\log f(Y,X)\right]=\arg\min_{f\in\mathcal{D}}E\left[\log\left(\frac{f_{Y|X}(Y\mid X)}{f(Y,X)}\right)\right].\label{eq:KLIC closest}
\end{equation}
Moreover, by the boundary conditions (\ref{eq:Boundary conditions}),
each $f\in\mathcal{D}$ satisfies
\begin{equation}
\int_{\mathbb{R}}f(y,X)dy=\lim_{y\rightarrow\infty}\Phi(b'T(X,y))-\lim_{y\rightarrow-\infty}\Phi(b'T(X,y))=1, \label{eq:Integrate to one}
\end{equation}
with probability one for some $b\in\Theta$. Therefore, (\ref{eq:KLIC closest})
implies that $f^{*}(Y,X)$ is the KLIC closest probability distribution
to $f_{Y|X}(Y|X)$ in $\mathcal{D}$.

By $F^{*}(Y,X)=\Phi(g^{*}(Y,X))$ and $f^{*}(Y,X)=\phi(g^{*}(Y,X))\{\partial_{y}g^{*}(Y,X)\}$,
we have $\partial_{y}F^{*}(Y,X)=f^{*}(Y,X)$. Since $y\mapsto f^{*}(y,X)$
is continuous, we obtain $F^{*}(y,X)=\int_{-\infty}^{y}f^{*}(t,X)dt$
for all $y\in\mathbb{R}$ by the Fundamental Theorem of Calculus,
with $\lim_{y\rightarrow-\infty}F^{*}(y,X)=0$ and $\lim_{y\rightarrow\infty}F^{*}(y,X)=1$
by definition of $F^{*}(y,X)$ and (\ref{eq:Integrate to one}).\par\end{sloppy}

Finally, by $f^{*}(Y,X)>0$ and $y\mapsto F^{*}(y,X)$ strictly increasing,
with probability one, the inverse function of $y\mapsto F^{*}(y,X)$
is well-defined, denoted $u\mapsto Q^{*}(X,u)$, with $\partial_{u}Q^{*}(X,u)=1/f^{*}(Q^{*}(X,u),X)>0$,
$u\in(0,1)$, by $y\mapsto F^{*}(y,X)$ continuously differentiable
and the Inverse Function Theorem.\qed

\section{Asymptotic Theory\label{sec:Asymptotic-Theory}}

\subsection{Proof of Theorem \ref{thm:Thm4}}

\subsubsection*{Parts (i)-(ii)}

We verify the conditions of Theorem 2.7 in \citet{Newey:McFadden:1994}.
By Theorem \ref{thm:Thm2}, $b^{*}\in\Theta$ is the unique minimizer
of $Q(b)$, and their Condition (i) is verified. Condition (ii) is
satisfied by convexity of $\Theta$ and concavity of $Q_{n}(b)$.
Finally, since the sample is i.i.d. by Assumption \ref{ass:Ass4}(i),
pointwise convergence of $Q_{n}(b)$ to $Q(b)$ follows from $Q(b)$
bounded and application of Khinchine's law of large numbers. Hence,
all conditions of \citeauthor{Newey:McFadden:1994}'s Theorem 2.7
are satisfied. Therefore, there exists $\hat{b}\in\Theta$ with probability
approaching one, and $\hat{b}\rightarrow_{p}b^{*}$.\qed

\subsubsection*{Part (ii)}

The asymptotic normality result $n^{1/2}(\hat{b}-b^{*})\overset{}{\rightarrow}_{d}N(0,\Gamma^{-1}\Psi(\Gamma^{-1})')$
follows from verifying the assumptions of Theorem 3.1 in \citet{Newey:McFadden:1994},
for instance. Symmetry and nonsingularity of $\Gamma$ then implies
that $V=\Gamma^{-1}\Psi\Gamma^{-1}$.

By Theorem \ref{thm:Thm3}, $b^{*}$ is in the interior of $\Theta$
and their Condition (i) is satisfied. Condition (ii) holds by inspection.
Condition (iii) holds by $E[\psi(Y,X,b^{*})]=0$, existence of $\Gamma$
and the Lindberg-Levy Central Limit Theorem. For Condition (iv), we
apply Lemma 2.4 in \citet{Newey:McFadden:1994} with $a(Y,X,b)\equiv\nabla_{bb}L(Y,X,b)$.
Let $\overline{\Theta}$ denote a compact subset of $\Theta$ containing
$b^{*}$ in its interior. By the proof of Lemma \ref{lem:2Differentiability},
$E[\sup_{b\in\overline{\Theta}}||\nabla_{bb}L(Y,X,b)||]<\infty$.
In addition, data is i.i.d. by assumption and $\nabla_{bb}L(Y,X,b)$
is continuous for $b\in\overline{\Theta}$ by inspection. Conditions
of Lemma 2.4 in \citet{Newey:McFadden:1994} are verified, and hence
their Condition (iv) in Theorem 3.1 also is. Finally, $\Gamma$ is
nonsingular by Lemma \ref{lem:Concavity} which verifies their Condition
(v). 

To show $\hat{\Gamma}^{-1}\hat{\Psi}\hat{\Gamma}^{-1}\rightarrow_{p}\Gamma^{-1}\Psi\Gamma^{-1}$,
we verify the conditions in the discussion of Theorem 4.4 in \citet[bottom of page 2158]{Newey:McFadden:1994}.
First, $\widehat{b}\rightarrow_{p}b^{*}$ by Theorem \ref{thm:Thm4}(ii).
Second, $\log f(Y,X,b)$ is twice continuously differentiable and
$f(Y,X,b)>0$, $b\in\Theta$. Moreover, $\Gamma$ exists and is nonsingular
by Lemma \ref{lem:Concavity}. Thus Conditions (ii) and (iv) of Theorem
3.3 in \citet{Newey:McFadden:1994} are verified. Third,
\begin{align*}
||\psi(Y,X,b)||^{2} & =||-T(X,Y)(b'T(X,Y))+(b't(X,Y))^{-1}t(X,Y)||^{2}\\
 & \leq2||T(X,Y)(b'T(X,Y))||^{2}+2|(b't(X,Y))^{-1}|^{2}\,||t(X,Y)||^{2}\\
 & \leq C\left\{ ||T(X,Y)||^{4}+||t(X,Y)||^{2}\right\} ,
\end{align*}
so that $E[\sup_{\theta\in\overline{\Theta}}||\psi(Y,X,b)||^{2}]<\infty$,
by Assumption \ref{ass: Ass2} and \ref{ass:Ass4}(ii). Hence, for
a neighborhood $\mathcal{N}$ of $b^{*}$, we have that $E[\sup_{b\in\mathcal{N}}||\psi(Y,X,b)||^{2}]<\infty$.
Moreover, $b\mapsto\psi(Y,X,b)$ is continuous at $b^{*}$ with probability
one.  The result follows.\qed

\subsection{Proof of Theorem \ref{thm:Thm5}}

By Theorem \ref{thm:Thm4}, $n^{1/2}(\widehat{b}-b^{*})\rightarrow_{d}N(0,\Xi)$,
$\Xi=\Gamma^{-1}\Psi\Gamma^{-1}$ positive
definite. For $(y,x)\in\mathcal{YX}$, $b\mapsto\Phi(b'T(x,y))\equiv F(y,x,b)$
and $b\mapsto f(y,x,b)$ are continuously differentiable, with $\nabla_{b}F(y,x,b)=\phi(b'T(x,y))T(x,y)$
and 
\begin{align*}
\nabla_{b}f(y,x,b) & =-\{b'T(x,y)\}\phi(b'T(x,y))\{b't(x,y)\}T(x,y)+\phi(b'T(x,y))t(x,y)\\
 & =\phi(b'T(x,y))\left[-\{b'T(x,y)\}\{b't(x,y)\}T(x,y)+t(x,y)\right],
\end{align*}
respectively, by the properties of the normal PDF. For all $(y,x)\in\mathcal{YX}$
with $f(y,x,b)>0$, $b\in\Theta$, we have that $y\mapsto F(y,x,b)$
is invertible, and its inverse function $u\mapsto F^{-1}(u,x,b)$
is continuously differentiable with derivative $1/f(F^{-1}(u,x,b),x,b)$
for all $x\in\mathcal{X}$ and $u\in\mathcal{U}_{x}(m)$, $m(y,x,b)\equiv b'T(x,y)$,
by the Inverse Function Theorem. Hence, by $F^{-1}(\Phi(b'T(x,y)),x,b)=y$,
we have for $(y,x)\in\mathcal{YX}$,
\[
\nabla_{b}F^{-1}(\Phi(b'T(x,y)),x,b)=\frac{\phi(b'T(x,y))}{f(F^{-1}(u_{0},x,b),x,b)}T(x,y)+\nabla_{b}F^{-1}(u_{0},x,b)=0,
\]
with $u_{0}=\Phi(b'T(x,y))$, and hence, for $x\in\mathcal{X}$ and
$u\in\mathcal{U}_{x}(m)$,
\[
\nabla_{b}F^{-1}(u,x,b)=-\frac{\phi(b'T(x,y_{0}))}{f(y_{0},x,b)}T(x,y_{0})=-\frac{1}{b't(x,y_{0})}T(x,y_{0})
\]
where $y_{0}=F^{-1}(u,x,b)$, which is continuous in $b$ on $\Theta$,
so that $b\mapsto F^{-1}(u,x,b)$ is continuously differentiable on
$\Theta$. Parts (i) and (ii) in the statement of Theorem \ref{thm:Thm5}
then follow by the Delta method (e.g., Lemma 3.9 in \citet{Wooldridge:2010}).\qed

\section{Implementation and Additional Empirical Results\label{sec:Implementation}}

\subsection{Implementation }

For each industry-occupation pair selected in our application, we
estimate four model classes for $g^{*}(Y,X,D)$ and its derivative
function:

$\textsc{I}$. \textsc{Linear-Linear:} set $s(Y)=(0,1)'$, $S(Y)=(1,Y)'$ and
$W(X_{1},X_{2})=(1,X_{1},X_{2})'$.

$\textsc{II}$. \textsc{Linear-$Y$ and Spline-$X$:} set $s(Y)=(0,1)'$, $S(Y)=(1,Y)'$
and $W(X_{1},X_{2})=(1,\widetilde{W}_{1}(X_{1})',\widetilde{W}_{2}(X_{2})')'$,
for $\widetilde{W}_{m}(X_{m})$ a vector of B-spline functions, $m\in\{1,2\}$. 

$\textsc{III}$. \textsc{Spline-$Y$ and Linear-$X$}: set $s(Y)=(0,1,\widetilde{s}(Y)')'$,
for $\widetilde{s}(Y)$ a vector of $J-2$ B-spline functions, and
$S(Y)=(1,Y,\widetilde{S}(Y)')'$ where $\widetilde{S}_{j}(y)=\int_{-\infty}^{y}\widetilde{s}(r)dr$,
$j\in\{1,\ldots,J-2\}$, and $W(X_{1},X_{2})=(1,X_{1},X_{2})'$. 

$\textsc{IV}$. \textsc{Spline-Spline:} set $s(Y)=(0,1,\widetilde{s}(Y)')'$, $S(Y)=(1,Y,\widetilde{S}(Y)')'$
and $W(X_{1},X_{2})=(1,\widetilde{W}_{1}(X_{1})',\widetilde{W}_{2}(X_{2})')'$.

For $\widetilde{W}_{m}(X_{m})$, $m\in \{1,2\}$, in specification classes II and IV, we consider cubic B-splines with $\dim(\widetilde{W}_{m}(X))\in\{5,6,7,8\}$, $m\in\{1,2\}$.
For $\widetilde{s}(Y)$ in classes $\textsc{III}$ and \textsc{IV}, we consider 
both quadratic and cubic B-splines, each with $\dim(\widetilde{s}(Y)) \in\{5,6,7\}$.
Thus, in addition to the model in class $\textsc{I}$, we consider $4$ specifications of class $\textsc{II}$, $2\times 3$ of class $\textsc{III}$, and $2\times 3 \times 4$ of class $\textsc{IV}$, for a total of 35 different model specifications.\footnote{For components $\widetilde{W}_{m}(X_{m})$, $m\in\{1,2\}$, and $\widetilde{s}(Y)$, we generate B-splines using the 
$\texttt{SplineBasis}$ function in the $\texttt{orthogonalsplinebasis}$ R package (\citet{R:2012,R:2015}), where the number of B-splines is the number of specified knots minus the order of the splines. To obtain $\widetilde{S}(Y)$, we integrate each $\widetilde{s}(Y)$ component using the $\texttt{integrate}$ function in that same package.
}

For each specification, we implement an adaptive Lasso ML (\citet{LGF:2012}, \citet{HN:2020})
version of our ML estimator with three steps. First, we run a Lasso
version of our estimator for each of 10 penalty parameter values in
a logarithmically spaced grid from $0.1$ to a threshold penalty level,
with no penalty on the intercept and $Y$ coefficients.\footnote{The threshold is the approximately smallest penalty value such that
all coefficients are zero (following \citet{CS:2021}).} For the selected penalty level, we record estimated coefficients
$\widehat{b}_{FS}$. Second, we implement Lasso ML with penalty weights
$1/|\widehat{b}_{FS}|$ if $\widehat{b}_{FS}\neq0$, or 0 otherwise
(cf. Eq. (4.2) page 16 of the first version \citet{SS:2020}), over
the same grid as in the first step. For each step, we follow the literature
on adaptive Lasso ML (\citet{LGF:2012}, \citet{HN:2020}) and select the
penalty parameter value that minimizes BIC, among penalized estimates
that satisfy QGM on an an equispaced fine grid that covers $\mathcal{Y}\times\mathcal{X}\times\{0,1\}$
(cf. Remark \ref{Rk:QGM} in the main text). Third, we record the
BIC value of the selected penalized model. We select the model with
smallest BIC among all 35 specifications.

For the legal occupation in the services industry considered in the main manuscript, 
the adaptive Lasso procedure selects a model of dimension 154, with 27 parameters estimated 
to be nonzero after penalization. Specifically, the selected dimension for $\widetilde{W}_{1}(X_{1})$ and $\widetilde{W}_{2}(X_{2})$ is $5$ for each. For $\widetilde{s}(Y)$, cubic splines are selected also with dimension $5$. 
Recalling the definitions for $S(Y)$ and $W(X_{1},X_{2})$ in class $\textsc{IV}$, 
and that the model in the empirical application is of the form $[W(X)\otimes (1,D)']\otimes S(Y)$, 
the selected model has dimension $[\dim(W)\times 2]\times \dim(S) = [(1+5+5)\times 2]\times (2+5)=154$.

\begin{rem}
For robustness we repeated the steps above using two nonnested grids
of lengths 5 and 20 instead of 10. We obtain similar distributional
regression estimates.

\begin{rem}
To follow the theory in the main manuscript, our implementation first
monotonically transforms $Y$ to expand its support over $\mathbb{R}$,
by implementing our estimator with $W(X)=1$ and using the rescaled
outcome $e_{0}\equiv\widehat{b}'S(Y)$.\label{Rk:Boundedsupp}

\end{rem}
\end{rem}

\subsection{Other examples}

\begin{figure}[t]
\subfloat[Sales occupation in wholesale trade ($n=6,790$).\label{fig:CQF_CI}]{\begin{centering}
\includegraphics[width=5.05cm,height=5.45cm]{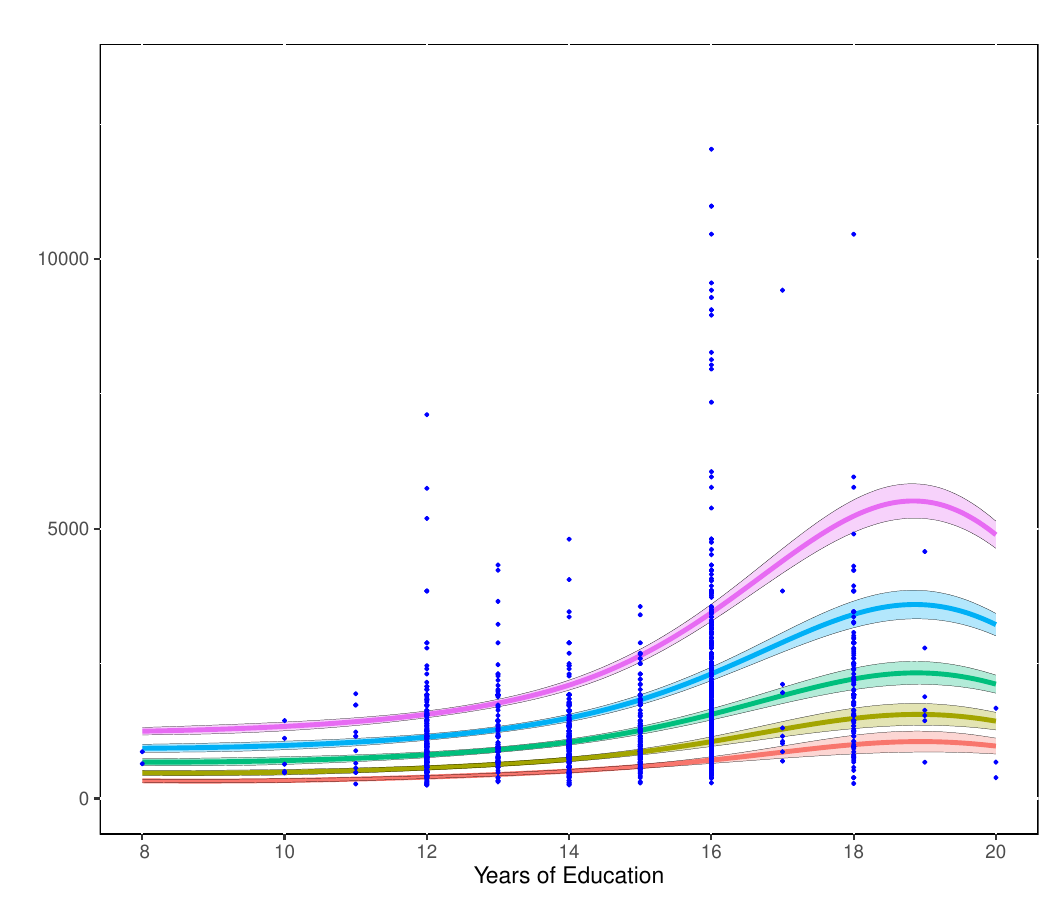}\hfill{}\includegraphics[width=5.05cm,height=5.45cm]{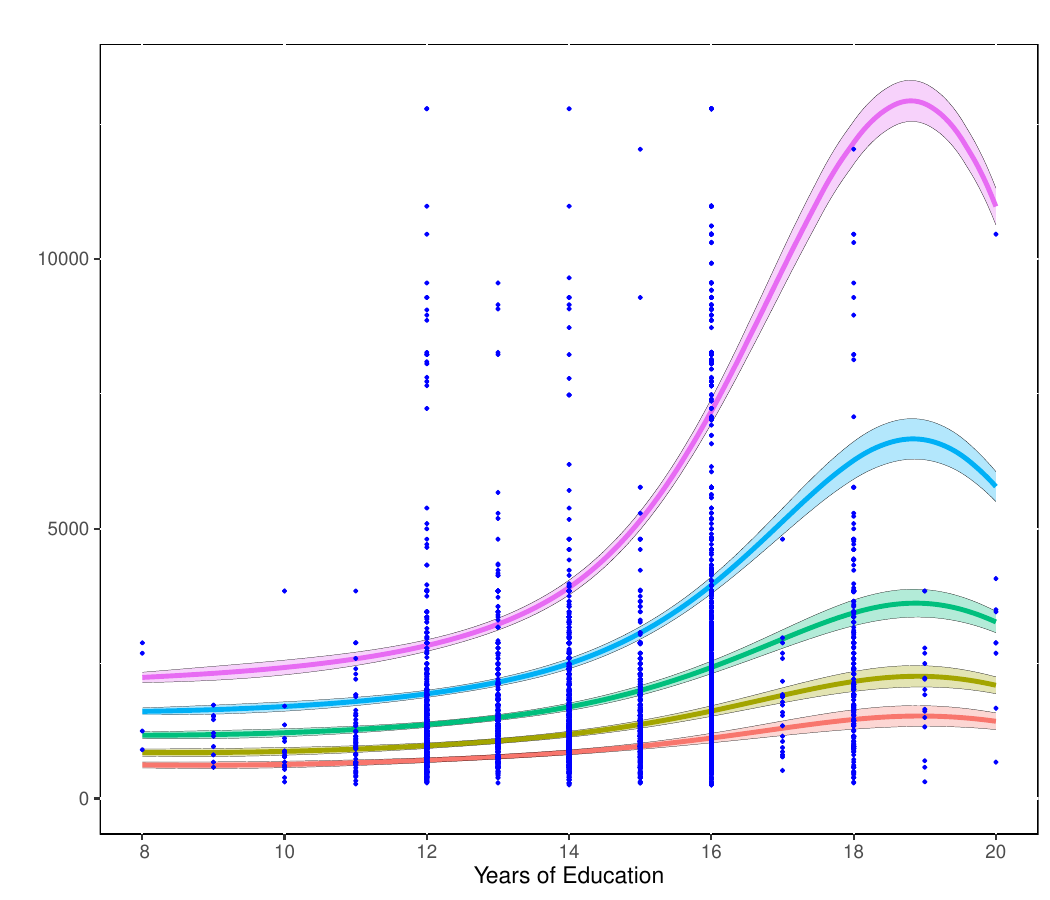}\hfill{}\includegraphics[width=5.05cm,height=5.45cm]{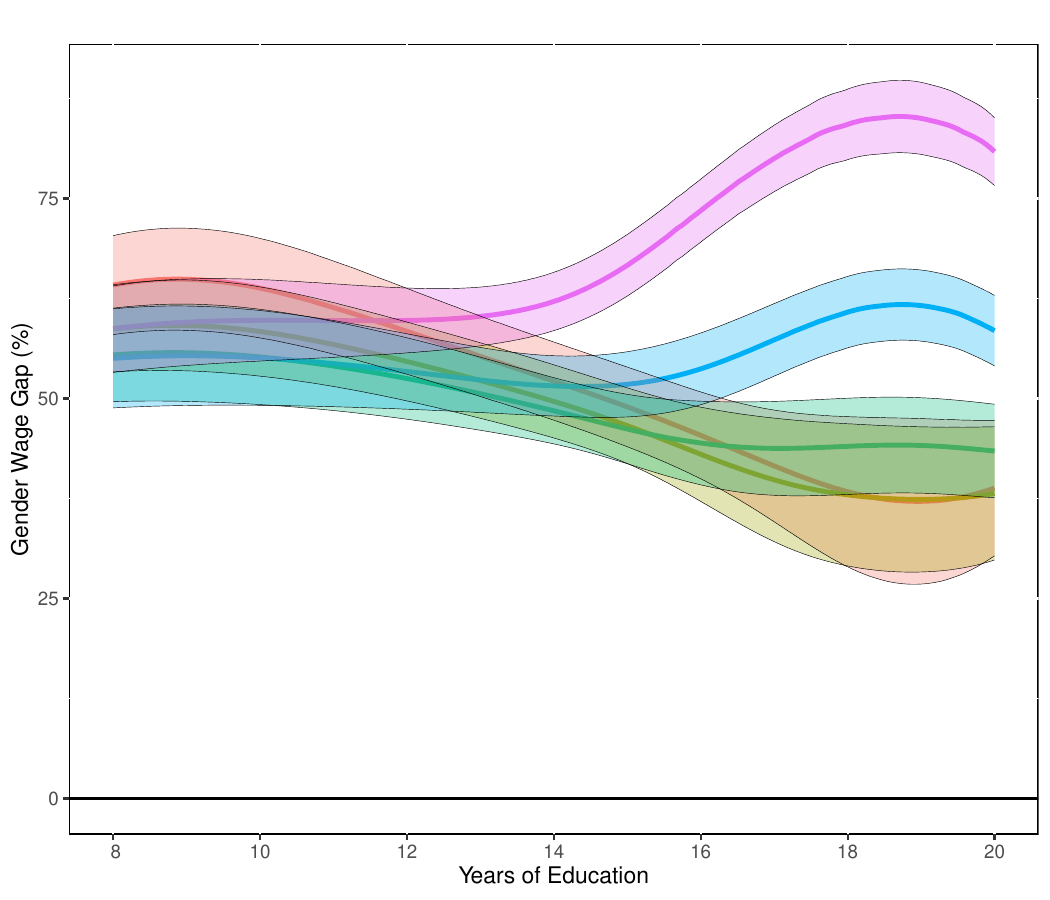}
\par\end{centering}
}

\subfloat[Architecture and engineering occupations in public administration
($n=1,068$).\label{fig:CQF_Full}]{\begin{centering}
\includegraphics[width=5.05cm,height=5.45cm]{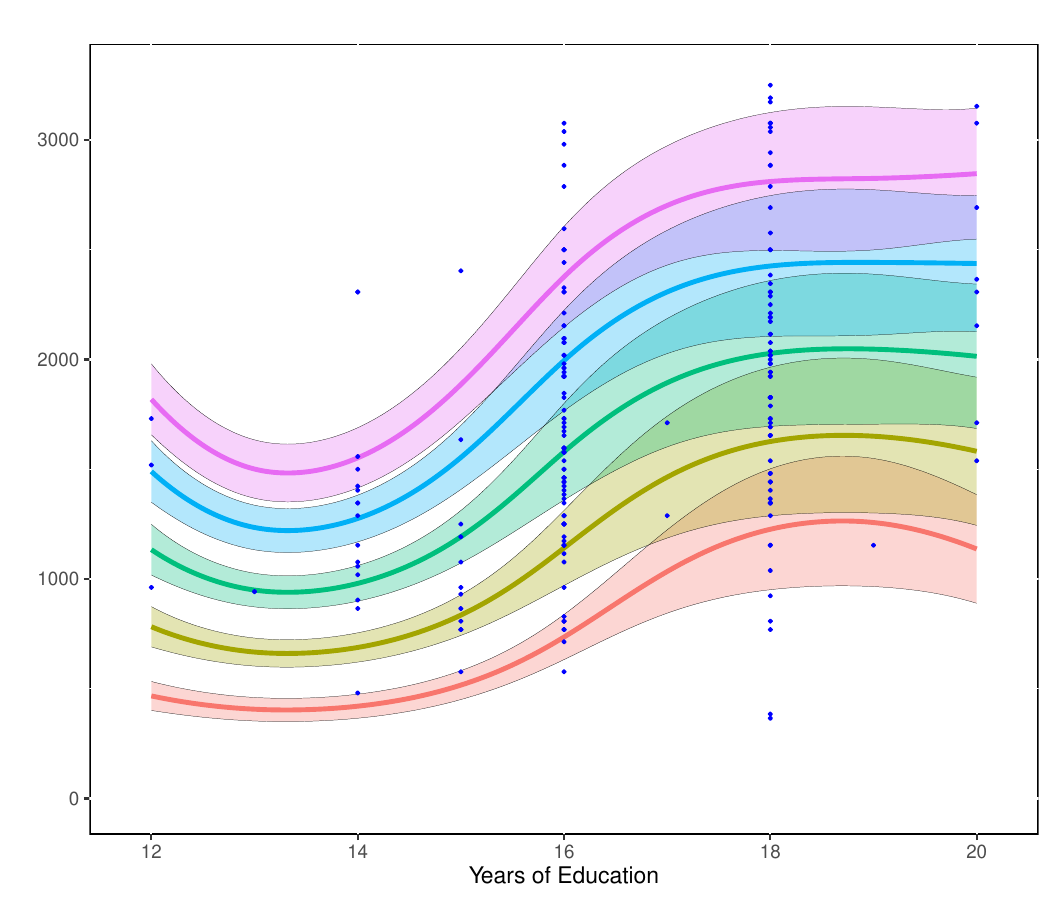}\hfill{}\includegraphics[width=5.05cm,height=5.45cm]{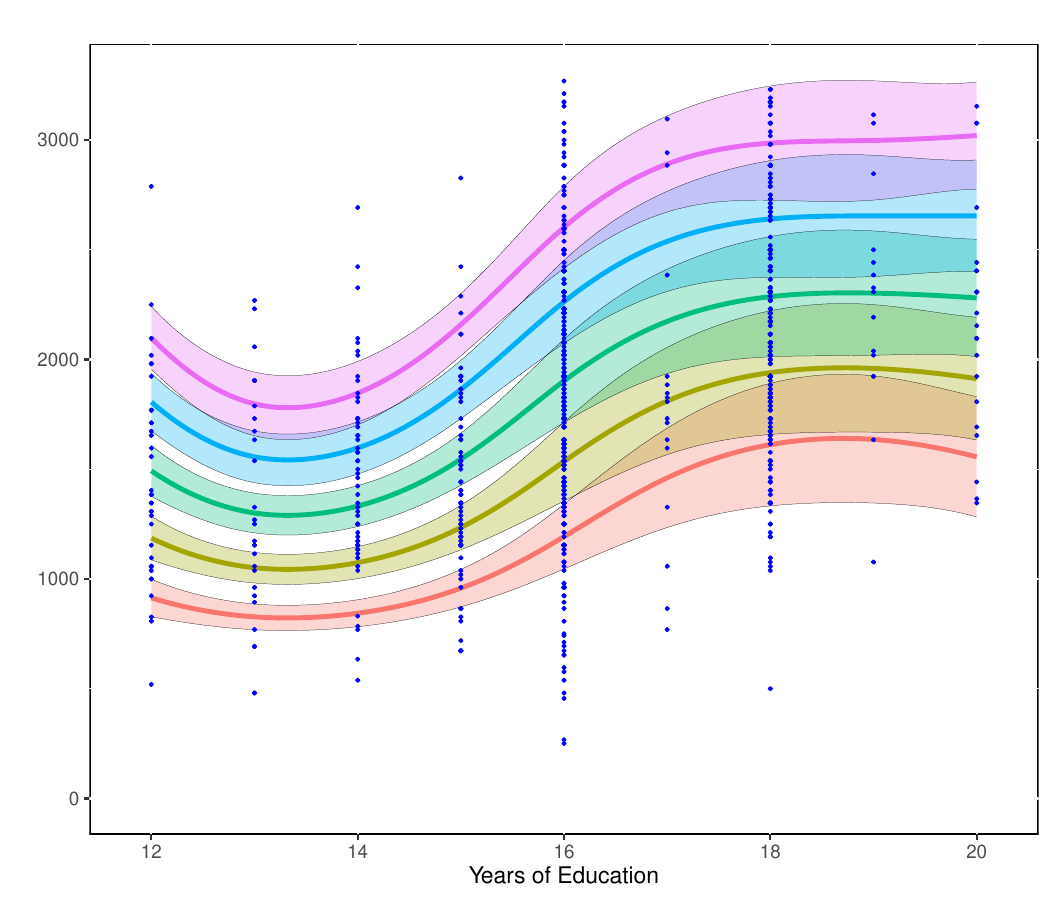}\hfill{}\includegraphics[width=5.05cm,height=5.45cm]{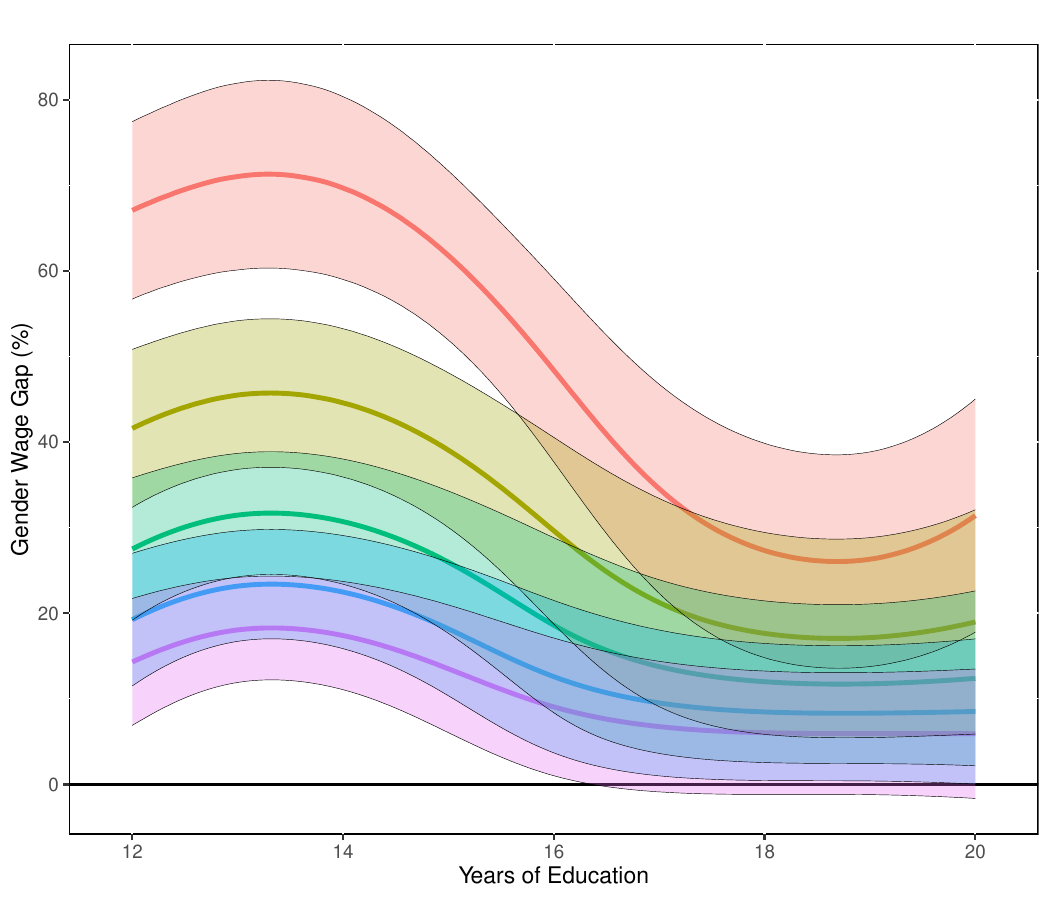}
\par\end{centering}
}\caption{CQF for $\mathtt{Female}$ (left), $\mathtt{Male}$ (center) and $\mathtt{GWG}(X,u)$
(right), for $u\in\{0.1,0.25,0.5,0.75,0.9\}$. Scatterplots by gender.\label{fig:CQFs}}
\end{figure}
We further illustrate the ability of our methods to capture a variety
of shapes and data features with two additional industry-occupation
pairs. For each, Figure \ref{fig:CQFs} shows CQFs and the quantile
gender wage gap function $\mathtt{GWG}(X,u)$. The gap is strongly
significant, with pronounced nonlinearities across education levels
and heterogeneity across quantiles. In Fig. \ref{fig:CQFs}(A), the
gap is homogeneous up to $u=0.5$, and diverges markedly across higher
quantiles, above year 12. In contrast, in Fig. \ref{fig:CQFs}(B)
heterogeneity is more pronounced for lower levels of education.

\section{Numerical Simulations\label{sec:Numerical-Simulations1}}

\subsection{Designs}

To assess the performance of GTR, we use the same implementation as
in the empirical application but without penalization. We borrow and
extend Monte Carlo experiments from \citet{Matzkin:2003} based on
two designs of the form $Y=m(X,\varepsilon)$.

\textsc{\uline{Design A:}} The first design in \citet{Matzkin:2003} is 
the linear location shift model
\begin{equation}
m(X,\varepsilon)=X+\varepsilon,\quad X\sim N(0,1),\quad\varepsilon\mid X\sim F_{\varepsilon},\label{eq:DGP}
\end{equation}
with $F_{\varepsilon}$ specified as (i) $N(0,1)$ as in \citet{Matzkin:2003},
(ii) $t$ with 5 degrees of freedom, (iii) mixture of 2 normals with
means $(-1,1)$, variances $(1,1)$, and weights $(2/3,1/3)$. GTR
is correctly specified for $\varepsilon\sim N(0,1)$, and misspecified
otherwise.
\begin{itemize}
\item For the implied $f_{Y\mid X}$ and $F_{Y\mid X}$, we compare GTR
to kernel estimators. 
\item For $m(X,\varepsilon)$, we compare GTR to (i) the Matzkin estimator,
and (ii) quantile regression (QR).
\end{itemize}
\textsc{\uline{Design B:}} The second design
in \citet{Matzkin:2003} is the nonlinear nonseparable model
\begin{equation}
m(X,\varepsilon)=(3^{4}/4^{4})X^{4}(-\varepsilon)^{-3},\quad X\sim N(6,1),\quad\varepsilon\sim N(-6,1).\label{eq:DGP-2}
\end{equation}
GTR is misspecified for this design.
\begin{itemize}
\item For the implied $f_{Y\mid X}$ and $F_{Y\mid X}$, we compare GTR
to kernel estimators. 
\item For $m(X,\varepsilon)$, we compare GTR to the Matzkin estimator only.
Linear QR is not consistent for this design.
\end{itemize}

\subsection{Implementation}

We consider datasets of sizes $n=100$, $250$ and $500$, and generate
$500$ of each size according to (\ref{eq:DGP}) and (\ref{eq:DGP-2}),
respectively. For each simulated dataset, we estimate the conditional
PDF and CDF at each sample values, $\{f_{Y|X}(y_{i}|x_{i})\}_{i=1}^{n}$
and $\{F_{Y|X}(y_{i}|x_{i})\}_{i=1}^{n}$, by GTR and kernel regression.
For $m(X,\varepsilon)$, we draw $20\times20$ fixed points from a
uniform distribution with support $[-2,2]\times[-2,2]$ for \textsc{Design
A} and with support $[4,8]\times[-8,-4]$ for \textsc{Design }B. This
mimics the steps in \citet{Matzkin:2003}. For each estimator, we
record the required estimate of $\{f_{Y|X}(y_{i}|x_{i})\}_{i=1}^{n}$, 
$\{F_{Y|X}(y_{i}|x_{i})\}_{i=1}^{n}$, and $m(x,e)$ at the $20\times20$
grid points.

For GTR, implementation and model specifications are as described
in Section \ref{sec:Implementation}, without using penalization.
For each simulated dataset we select the best GT model according to
the BIC. QR is implemented using the $\mathtt{quantreg}$ package
(\citet{K:2021}). The Matzkin estimator is implemented as in \citet[p. 1355]{Matzkin:2003},
with
\[
\widehat{m}(x,e)=\widehat{Q}_{Y|X}(\widehat{F}_{Y|X}\left((e/\overline{\varepsilon})\alpha\mid X=(e/\overline{\varepsilon})\overline{x}\right)\mid X=x),
\]
where $\overline{x}=n^{-1}\sum_{i=1}^{n}x_{i}$, and setting $\overline{\varepsilon}=1$
and $\alpha=m(\overline{x},1)$ for \textsc{Design A}, and $\overline{\varepsilon}=-6$
and $\alpha=m(\overline{x},-6)$ for \textsc{Design }B. For all kernel-based
estimators, as well as bandwidth selection by cross-validation for
each simulated dataset, we use the default $\mathtt{np}$ package
implementation (\citet{RH:2021}). 

\subsection{Results}

Tables \ref{tab:PDF_A}-\ref{tab:PDF_CDF_B} report simulation results
regarding the accuracy of conditional PDF and CDF GTR estimates, using
a kernel estimator as a benchmark. We report average estimation errors
across simulations of both estimators, and their ratio in percentage
terms. Estimation errors are measured in $L_{p}$ norms $\left\Vert \cdot\right\Vert _{p}$,
$p=1,2,$ and $\infty$, and are then averaged over the 500 simulations.

For both conditional PDFs and CDFs, GTR yields very large
improvements over kernel estimation for all sample sizes, distributions
and designs. Except for PDFs and $L_{1}$ norm in \textsc{Design B} with $n=100$,
the relative accuracy of GTR estimates is largely superior across
both designs. In particular, decreasing ratios in average errors in
$L_{\infty}$ norm across sample size reflect higher GTR accuracy
in estimation of extreme parts.

Tables \ref{tab:CQF_A}-\ref{tab:CQF_B} report simulation results
regarding the accuracy of $m(x,e)$ estimates by GTR relative to the
Matzkin and quantile regression estimators. We report absolute bias,
variance and mean-square error across simulations of each estimator,
and comparison ratios in percentage terms. We find that GTR largely
improves over all methods, for all sample sizes, distributions and
designs. For \textsc{Design A} and $n=500$\textsc{,} larger GTR absolute
bias relative to quantile regression reflects GTR misspecification.

Overall, these simulation results give a further illustration of very
large finite sample gains from using the methods introduced in this
paper. These simulation results complement those for a location-scale
model in the previous version of the paper (\citet{SS:2020}) that
show large improvements over distribution regression.
\begin{rem}
Tables \ref{tab:CQF_A-1}-\ref{tab:CQF_B-1} report additional results
for estimation of $\widehat{m}(x,e)$ where, for each point and estimated
function, we only used the simulations for which the estimated densities
and multiplications of densities that appear in the denominator of
the Matzkin estimator were above $0.025$. This mimics the implementation
in \citet[pp. 1360-1361]{Matzkin:2003}. The performance of the Matzkin
estimator improves for \textsc{Design B} but overall improvements
obtained by GTR remain large.
\end{rem}

\begin{table}
\begin{centering}
\centering %
\begin{tabular}{lccccccccccccc}
\hline 
 &  & \multicolumn{3}{c}{$n=100$} &  & \multicolumn{3}{c}{$n=250$} &  & \multicolumn{3}{c}{$n=500$} & \tabularnewline
\hline 
 &  &  &  &  &  &  &  &  &  &  &  &  & \tabularnewline
 &  & $L_{1}$ & $L_{2}$ & $L_{\infty}$ &  & $L_{1}$ & $L_{2}$ & $L_{\infty}$ &  & $L_{1}$ & $L_{2}$ & $L_{\infty}$ & \tabularnewline
\cline{3-13} \cline{4-13} \cline{5-13} \cline{6-13} \cline{7-13} \cline{8-13} \cline{9-13} \cline{10-13} \cline{11-13} \cline{12-13} \cline{13-13} 
 &  & \multicolumn{11}{c}{\textsc{(i) Gaussian distribution}} & \tabularnewline
\cline{3-13} \cline{4-13} \cline{5-13} \cline{6-13} \cline{7-13} \cline{8-13} \cline{9-13} \cline{10-13} \cline{11-13} \cline{12-13} \cline{13-13} 
GTR &  & ~31.6 & ~~2.0 & 121.3 &  & ~19.1 & ~~0.7 & ~79.0 &  & ~12.8 & ~~0.3 & ~57.2 & \tabularnewline
Kernel &  & ~58.6 & ~~6.4 & 311.0 &  & ~44.3 & ~~3.9 & 370.9 &  & ~35.8 & ~~2.6 & 421.4 & \tabularnewline
\textbf{Ratio$\times100$} &  & \textbf{~53.9} & \textbf{~31.1} & \textbf{~39.0} &  & \textbf{~43.1} & \textbf{~16.8} & \textbf{~21.3} &  & \textbf{~35.9} & \textbf{~11.2} & \textbf{~13.6} & \tabularnewline
 &  &  &  &  &  &  &  &  &  &  &  &  & \tabularnewline
\cline{3-13} \cline{4-13} \cline{5-13} \cline{6-13} \cline{7-13} \cline{8-13} \cline{9-13} \cline{10-13} \cline{11-13} \cline{12-13} \cline{13-13} 
 &  & \multicolumn{11}{c}{\textsc{(ii) $t$ distribution}} & \tabularnewline
\cline{3-13} \cline{4-13} \cline{5-13} \cline{6-13} \cline{7-13} \cline{8-13} \cline{9-13} \cline{10-13} \cline{11-13} \cline{12-13} \cline{13-13} 
GTR &  & ~37.2 & ~~2.4 & 122.3 &  & ~31.3 & ~~1.6 & 106.1 &  & ~28.4 & ~~1.3 & 106.9 & \tabularnewline
Kernel &  & ~58.5 & ~~5.3 & 215.8 &  & ~48.6 & ~~3.8 & 235.0 &  & ~42.0 & ~~2.8 & 244.0 & \tabularnewline
\textbf{Ratio$\times100$} &  & \textbf{~63.5} & \textbf{~44.7} & \textbf{~56.7} &  & \textbf{~64.5} & \textbf{~41.8} & \textbf{~45.2} &  & \textbf{~67.7} & \textbf{~45.7} & \textbf{~43.8} & \tabularnewline
 &  &  &  &  &  &  &  &  &  &  &  &  & \tabularnewline
\cline{3-13} \cline{4-13} \cline{5-13} \cline{6-13} \cline{7-13} \cline{8-13} \cline{9-13} \cline{10-13} \cline{11-13} \cline{12-13} \cline{13-13} 
 &  & \multicolumn{11}{c}{\textsc{(iii) Gaussian mixture distribution}} & \tabularnewline
\cline{3-13} \cline{4-13} \cline{5-13} \cline{6-13} \cline{7-13} \cline{8-13} \cline{9-13} \cline{10-13} \cline{11-13} \cline{12-13} \cline{13-13} 
GTR &  & ~29.8 & ~~1.5 & ~93.2 &  & ~24.1 & ~~0.9 & ~80.1 &  & ~21.9 & ~~0.7 & ~62.3 & \tabularnewline
Kernel &  & ~39.4 & ~~3.1 & 224.4 &  & ~30.3 & ~~2.0 & 280.6 &  & ~24.5 & ~~1.4 & 325.1 & \tabularnewline
\textbf{Ratio$\times100$} &  & \textbf{~75.7} & \textbf{~46.5} & \textbf{~41.5} &  & \textbf{~79.6} & \textbf{~46.6} & \textbf{~28.5} &  & \textbf{~89.5} & \textbf{~52.8} & \textbf{~19.2} & \tabularnewline
 &  &  &  &  &  &  &  &  &  &  &  &  & \tabularnewline
\hline 
\end{tabular}
\par\end{centering}
\caption{\textsc{Design A:} PDF. 
Average $L^{p}$ estimation errors $\times1000$ for GTR, kernel estimator and their ratio $\times100$, for $p=1,2$, $\infty$.
\label{tab:PDF_A}}
\end{table}

\begin{table}
\begin{centering}
\centering %
\begin{tabular}{lccccccccccccc}
\hline 
 &  & \multicolumn{3}{c}{$n=100$} &  & \multicolumn{3}{c}{$n=250$} &  & \multicolumn{3}{c}{$n=500$} & \tabularnewline
\hline 
 &  &  &  &  &  &  &  &  &  &  &  &  & \tabularnewline
 &  & $L_{1}$ & $L_{2}$ & $L_{\infty}$ &  & $L_{1}$ & $L_{2}$ & $L_{\infty}$ &  & $L_{1}$ & $L_{2}$ & $L_{\infty}$ & \tabularnewline
\cline{3-13} \cline{4-13} \cline{5-13} \cline{6-13} \cline{7-13} \cline{8-13} \cline{9-13} \cline{10-13} \cline{11-13} \cline{12-13} \cline{13-13} 
 &  & \multicolumn{11}{c}{\textsc{(i) Gaussian distribution}} & \tabularnewline
\cline{3-13} \cline{4-13} \cline{5-13} \cline{6-13} \cline{7-13} \cline{8-13} \cline{9-13} \cline{10-13} \cline{11-13} \cline{12-13} \cline{13-13} 
GTR &  & ~37.2 & ~~2.7 & 157.2 &  & ~22.9 & ~~1.0 & 110.1 &  & ~15.9 & ~~0.5 & ~84.8 & \tabularnewline
Kernel &  & ~66.1 & ~~7.8 & 291.0 &  & ~45.4 & ~~4.0 & 289.4 &  & ~34.1 & ~~2.4 & 289.3 & \tabularnewline
\textbf{Ratio$\times100$} &  & \textbf{~56.3} & \textbf{~34.1} & \textbf{~54.0} &  & \textbf{~50.5} & \textbf{~25.4} & \textbf{~38.1} &  & \textbf{~46.5} & \textbf{~20.9} & \textbf{~29.3} & \tabularnewline
 &  &  &  &  &  &  &  &  &  &  &  &  & \tabularnewline
\cline{3-13} \cline{4-13} \cline{5-13} \cline{6-13} \cline{7-13} \cline{8-13} \cline{9-13} \cline{10-13} \cline{11-13} \cline{12-13} \cline{13-13} 
 &  & \multicolumn{11}{c}{\textsc{(ii) $t$ distribution}} & \tabularnewline
\cline{3-13} \cline{4-13} \cline{5-13} \cline{6-13} \cline{7-13} \cline{8-13} \cline{9-13} \cline{10-13} \cline{11-13} \cline{12-13} \cline{13-13} 
GTR &  & ~41.1 & ~~3.2 & 168.1 &  & ~29.6 & ~~1.6 & 136.0 &  & ~23.7 & ~~1.0 & 127.5 & \tabularnewline
Kernel &  & ~63.6 & ~~7.5 & 290.4 &  & ~43.9 & ~~3.8 & 280.1 &  & ~32.9 & ~~2.2 & 280.9 & \tabularnewline
\textbf{Ratio$\times100$} &  & \textbf{~64.7} & \textbf{~42.3} & \textbf{~57.9} &  & \textbf{~67.4} & \textbf{~42.1} & \textbf{~48.6} &  & \textbf{~72.1} & \textbf{~45.6} & \textbf{~45.4} & \tabularnewline
 &  &  &  &  &  &  &  &  &  &  &  &  & \tabularnewline
\cline{3-13} \cline{4-13} \cline{5-13} \cline{6-13} \cline{7-13} \cline{8-13} \cline{9-13} \cline{10-13} \cline{11-13} \cline{12-13} \cline{13-13} 
 &  & \multicolumn{11}{c}{\textsc{(iii) Gaussian mixture distribution}} & \tabularnewline
\cline{3-13} \cline{4-13} \cline{5-13} \cline{6-13} \cline{7-13} \cline{8-13} \cline{9-13} \cline{10-13} \cline{11-13} \cline{12-13} \cline{13-13} 
GTR &  & ~37.4 & ~~2.6 & 142.6 &  & ~26.9 & ~~1.6 & 110.0 &  & ~20.5 & ~~0.7 & ~85.4 & \tabularnewline
Kernel &  & ~60.2 & ~~6.5 & 264.6 &  & ~41.6 & ~~3.3 & 262.4 &  & ~31.8 & ~~2.0 & 258.2 & \tabularnewline
\textbf{Ratio$\times100$} &  & \textbf{~62.0} & \textbf{~39.4} & \textbf{~53.9} &  & \textbf{~64.8} & \textbf{~49.8} & \textbf{~41.9} &  & \textbf{~64.6} & \textbf{~36.7} & \textbf{~33.1} & \tabularnewline
 &  &  &  &  &  &  &  &  &  &  &  &  & \tabularnewline
\hline 
\end{tabular}
\par\end{centering}
\caption{\textsc{Design A}: CDF. 
Average $L^{p}$ estimation errors $\times1000$ for GTR, kernel estimator and their ratio $\times100$, for $p=1,2$, $\infty$. 
\label{tab:CDF_A}}
\end{table}

\begin{table}
\begin{centering}
\centering %
\begin{tabular}{lccccccccccccc}
\hline 
 &  & \multicolumn{3}{c}{$n=100$} &  & \multicolumn{3}{c}{$n=250$} &  & \multicolumn{3}{c}{$n=500$} & \tabularnewline
\hline 
 &  &  &  &  &  &  &  &  &  &  &  &  & \tabularnewline
 &  & $L_{1}$ & $L_{2}$ & $L_{\infty}$ &  & $L_{1}$ & $L_{2}$ & $L_{\infty}$ &  & $L_{1}$ & $L_{2}$ & $L_{\infty}$ & \tabularnewline
\cline{3-13} \cline{4-13} \cline{5-13} \cline{6-13} \cline{7-13} \cline{8-13} \cline{9-13} \cline{10-13} \cline{11-13} \cline{12-13} \cline{13-13} 
 &  & \multicolumn{11}{c}{\textsc{(i) PDF}} & \tabularnewline
\cline{3-13} \cline{4-13} \cline{5-13} \cline{6-13} \cline{7-13} \cline{8-13} \cline{9-13} \cline{10-13} \cline{11-13} \cline{12-13} \cline{13-13} 
GTR &  & ~301.4 & ~423.4 & 3333.2 &  & ~180.1 & ~258.3 & 4340.0 &  & ~161.7 & ~237.2 & 5738.8 & \tabularnewline
Kernel &  & ~242.4 & ~447.0 & 3697.8 &  & ~224.3 & ~367.3 & 5161.5 &  & ~211.8 & ~371.3 & 6963.0 & \tabularnewline
\textbf{Ratio$\times100$} &  & \textbf{~124.3} & \textbf{~~94.7} & \textbf{~~90.1} &  & \textbf{~~80.3} & \textbf{~~70.3} & \textbf{~~84.1} &  & \textbf{~~76.4} & \textbf{~~63.9} & \textbf{~~82.4} & \tabularnewline
 &  &  &  &  &  &  &  &  &  &  &  &  & \tabularnewline
\cline{3-13} \cline{4-13} \cline{5-13} \cline{6-13} \cline{7-13} \cline{8-13} \cline{9-13} \cline{10-13} \cline{11-13} \cline{12-13} \cline{13-13} 
 &  & \multicolumn{11}{c}{\textsc{(ii) CDF}} & \tabularnewline
\cline{3-13} \cline{4-13} \cline{5-13} \cline{6-13} \cline{7-13} \cline{8-13} \cline{9-13} \cline{10-13} \cline{11-13} \cline{12-13} \cline{13-13} 
GTR &  & ~48.2 & ~~4.6 & 226.7 &  & ~44.9 & ~~4.3 & 264.0 &  & ~37.6 & ~~2.9 & 265.3 & \tabularnewline
Kernel &  & ~71.6 & ~~9.3 & 329.3 &  & ~50.4 & ~~4.9 & 322.8 &  & ~37.8 & ~~2.9 & 319.9 & \tabularnewline
\textbf{Ratio$\times100$} &  & \textbf{~67.3} & \textbf{~49.1} & \textbf{~68.9} &  & \textbf{~89.2} & \textbf{~87.3} & \textbf{~81.8} &  & \textbf{~99.4} & \textbf{~98.2} & \textbf{~82.9} & \tabularnewline
 &  &  &  &  &  &  &  &  &  &  &  &  & \tabularnewline
\hline 
\end{tabular}
\par\end{centering}
\caption{\textsc{Design B}: PDF and CDF. 
Average $L^{p}$ estimation errors $\times1000$ for GTR, kernel estimator and their ratio $\times100$, for $p=1,2$, $\infty$. 
\label{tab:PDF_CDF_B}}
\end{table}

\begin{table}
\begin{centering}
\centering %
\begin{tabular}{lcccccccccccccc}
\hline 
 &  &  & \multicolumn{3}{c}{$n=100$} &  & \multicolumn{3}{c}{$n=250$} &  & \multicolumn{3}{c}{$n=500$} & \tabularnewline
\hline 
 &  &  &  &  &  &  &  &  &  &  &  &  &  & \tabularnewline
 &  &  & |Bias| & Var. & MSE &  & |Bias| & Var. & MSE &  & |Bias| & Var. & MSE & \tabularnewline
\cline{4-14} \cline{5-14} \cline{6-14} \cline{7-14} \cline{8-14} \cline{9-14} \cline{10-14} \cline{11-14} \cline{12-14} \cline{13-14} \cline{14-14} 
 &  &  & \multicolumn{11}{c}{\textsc{(i) Gaussian distribution}} & \tabularnewline
\cline{4-14} \cline{5-14} \cline{6-14} \cline{7-14} \cline{8-14} \cline{9-14} \cline{10-14} \cline{11-14} \cline{12-14} \cline{13-14} \cline{14-14} 
GTR &  &  & 143.9 & ~32.5 & ~53.2 &  & ~88.1 & ~12.3 & ~20.1 &  & ~61.3 & ~~6.0 & ~~9.8 & \tabularnewline
Matzkin (M) &  &  & 249.0 & ~94.4 & 156.4 &  & 180.4 & ~47.0 & ~79.6 &  & 137.5 & ~27.5 & ~46.4 & \tabularnewline
\textbf{GTR/M$\times100$} &  &  & \textbf{~57.8} & \textbf{~34.4} & \textbf{~34.0} &  & \textbf{~48.9} & \textbf{~26.2} & \textbf{~25.2} &  & \textbf{~44.6} & \textbf{~21.8} & \textbf{~21.1} & \tabularnewline
QR &  &  & 176.9 & ~49.0 & ~80.3 &  & 109.4 & ~18.7 & ~30.6 &  & ~77.7 & ~~9.4 & ~15.5 & \tabularnewline
\textbf{GTR/QR$\times100$} &  &  & \textbf{~81.3} & \textbf{~66.4} & \textbf{~66.3} &  & \textbf{~80.6} & \textbf{~66.0} & \textbf{~65.6} &  & \textbf{~79.0} & \textbf{~63.5} & \textbf{~63.1} & \tabularnewline
 &  &  &  &  &  &  &  &  &  &  &  &  &  & \tabularnewline
\cline{4-14} \cline{5-14} \cline{6-14} \cline{7-14} \cline{8-14} \cline{9-14} \cline{10-14} \cline{11-14} \cline{12-14} \cline{13-14} \cline{14-14} 
 &  &  & \multicolumn{11}{c}{\textsc{(ii) $t$ distribution}} & \tabularnewline
\cline{4-14} \cline{5-14} \cline{6-14} \cline{7-14} \cline{8-14} \cline{9-14} \cline{10-14} \cline{11-14} \cline{12-14} \cline{13-14} \cline{14-14} 
GTR &  &  & 202.5 & ~60.0 & 101.0 &  & 144.4 & ~24.7 & ~45.6 &  & 116.3 & ~12.3 & ~25.9 & \tabularnewline
Matzkin (M) &  &  & 307.5 & 165.9 & 260.4 &  & 220.8 & ~80.6 & 129.4 &  & 170.6 & ~47.1 & ~76.2 & \tabularnewline
\textbf{GTR/M$\times100$} &  &  & \textbf{~65.8} & \textbf{~36.2} & \textbf{~38.8} &  & \textbf{~65.4} & \textbf{~30.7} & \textbf{~35.3} &  & \textbf{~68.2} & \textbf{~26.2} & \textbf{~33.9} & \tabularnewline
QR &  &  & 225.5 & ~82.9 & 133.8 &  & 144.2 & ~32.9 & ~53.7 &  & 101.3 & ~16.3 & ~26.6 & \tabularnewline
\textbf{GTR/QR$\times100$} &  &  & \textbf{~89.8} & \textbf{~72.4} & \textbf{~75.5} &  & \textbf{100.2} & \textbf{~75.2} & \textbf{~85.0} &  & \textbf{114.9} & \textbf{~75.6} & \textbf{~97.3} & \tabularnewline
 &  &  &  &  &  &  &  &  &  &  &  &  &  & \tabularnewline
\cline{4-14} \cline{5-14} \cline{6-14} \cline{7-14} \cline{8-14} \cline{9-14} \cline{10-14} \cline{11-14} \cline{12-14} \cline{13-14} \cline{14-14} 
 &  &  & \multicolumn{11}{c}{\textsc{(iii) Gaussian mixture distribution}} & \tabularnewline
\cline{4-14} \cline{5-14} \cline{6-14} \cline{7-14} \cline{8-14} \cline{9-14} \cline{10-14} \cline{11-14} \cline{12-14} \cline{13-14} \cline{14-14} 
GTR &  &  & 198.0 & ~52.8 & ~92.0 &  & 139.5 & ~23.5 & ~43.0 &  & 109.2 & ~10.9 & ~22.8 & \tabularnewline
Matzkin (M) &  &  & 271.4 & 111.5 & 185.2 &  & 192.8 & ~55.3 & ~92.4 &  & 149.3 & ~32.8 & ~55.0 & \tabularnewline
\textbf{GTR/M$\times100$} &  &  & \textbf{~72.9} & \textbf{~47.3} & \textbf{~49.7} &  & \textbf{~72.3} & \textbf{~42.6} & \textbf{~46.5} &  & \textbf{~73.1} & \textbf{~33.3} & \textbf{~41.5} & \tabularnewline
QR &  &  & 223.5 & ~78.7 & 128.7 &  & 143.1 & ~32.3 & ~52.8 &  & 102.8 & ~16.4 & ~26.9 & \tabularnewline
\textbf{GTR/QR$\times100$} &  &  & \textbf{~88.6} & \textbf{~67.1} & \textbf{~71.5} &  & \textbf{~97.5} & \textbf{~72.8} & \textbf{~81.4} &  & \textbf{106.2} & \textbf{~66.6} & \textbf{~84.7} & \tabularnewline
 &  &  &  &  &  &  &  &  &  &  &  &  &  & \tabularnewline
\hline 
\end{tabular}
\par\end{centering}
\caption{\textsc{Design A}: $\widehat{m}(x,e)$. 
Average $L^{p}$ estimation errors $\times1000$ for GTR, Matzkin estimator, QR and related ratios $\times100$. 
\label{tab:CQF_A}}
\end{table}

\begin{table}
\begin{centering}
\centering %
\begin{tabular}{lccccccccccccc}
\hline 
 &  & \multicolumn{3}{c}{$n=100$} &  & \multicolumn{3}{c}{$n=250$} &  & \multicolumn{3}{c}{$n=500$} & \tabularnewline
\hline 
 &  &  &  &  &  &  &  &  &  &  &  &  & \tabularnewline
 &  & |Bias| & Var. & MSE &  & |Bias| & Var. & MSE &  & |Bias| & Var. & MSE & \tabularnewline
\cline{3-13} \cline{4-13} \cline{5-13} \cline{6-13} \cline{7-13} \cline{8-13} \cline{9-13} \cline{10-13} \cline{11-13} \cline{12-13} \cline{13-13} 
GTR &  & 310.5 & 120.7 & 217.1 &  & 274.9 & ~75.0 & 150.6 &  & 231.1 & ~38.5 & ~91.9 & \tabularnewline
Matzkin &  & 571.1 & 575.5 & 901.6 &  & 433.8 & 316.8 & 504.9 &  & 340.6 & 210.3 & 326.3 & \tabularnewline
\textbf{Ratio$\times100$} &  & \textbf{~54.4} & \textbf{~21.0} & \textbf{~24.1} &  & \textbf{~63.4} & \textbf{~23.7} & \textbf{~29.8} &  & \textbf{~67.9} & \textbf{~18.3} & \textbf{~28.2} & \tabularnewline
 &  &  &  &  &  &  &  &  &  &  &  &  & \tabularnewline
\hline 
\end{tabular}
\par\end{centering}
\caption{\textsc{Design B}: $\widehat{m}(x,e)$. 
Absolute bias, variance and MSE ($\times1000$) for GTR, Matzkin estimator and their ratio $\times100$. 
\label{tab:CQF_B}}
\end{table}

\begin{table}
\begin{centering}
\centering %
\begin{tabular}{lcccccccccccccc}
\hline 
 &  &  & \multicolumn{3}{c}{$n=100$} &  & \multicolumn{3}{c}{$n=250$} &  & \multicolumn{3}{c}{$n=500$} & \tabularnewline
\hline 
 &  &  &  &  &  &  &  &  &  &  &  &  &  & \tabularnewline
 &  &  & |Bias| & Var. & MSE &  & |Bias| & Var. & MSE &  & |Bias| & Var. & MSE & \tabularnewline
\cline{4-14} \cline{5-14} \cline{6-14} \cline{7-14} \cline{8-14} \cline{9-14} \cline{10-14} \cline{11-14} \cline{12-14} \cline{13-14} \cline{14-14} 
 &  &  & \multicolumn{11}{c}{\textsc{(i) Gaussian distribution}} & \tabularnewline
\cline{4-14} \cline{5-14} \cline{6-14} \cline{7-14} \cline{8-14} \cline{9-14} \cline{10-14} \cline{11-14} \cline{12-14} \cline{13-14} \cline{14-14} 
GTR &  &  & 131.7 & ~31.3 & ~48.6 &  & ~80.4 & ~12.1 & ~18.6 &  & ~55.6 & ~~5.7 & ~~8.8 & \tabularnewline
Matzkin (M) &  &  & 221.7 & ~81.9 & 131.1 &  & 159.2 & ~39.8 & ~65.1 &  & 120.6 & ~22.6 & ~37.1 & \tabularnewline
\textbf{GTR/M$\times100$} &  &  & \textbf{~59.0} & \textbf{~38.0} & \textbf{~37.0} &  & \textbf{~50.0} & \textbf{~30.0} & \textbf{~28.0} &  & \textbf{~46.0} & \textbf{~25.0} & \textbf{~24.0} & \tabularnewline
 &  &  &  &  &  &  &  &  &  &  &  &  &  & \tabularnewline
\cline{4-14} \cline{5-14} \cline{6-14} \cline{7-14} \cline{8-14} \cline{9-14} \cline{10-14} \cline{11-14} \cline{12-14} \cline{13-14} \cline{14-14} 
 &  &  & \multicolumn{11}{c}{\textsc{(ii) $t$ distribution}} & \tabularnewline
\cline{4-14} \cline{5-14} \cline{6-14} \cline{7-14} \cline{8-14} \cline{9-14} \cline{10-14} \cline{11-14} \cline{12-14} \cline{13-14} \cline{14-14} 
GTR &  &  & 173.1 & ~58.6 & ~88.6 &  & 123.3 & ~22.8 & ~38.0 &  & ~97.9 & ~11.0 & ~20.6 & \tabularnewline
Matzkin (M) &  &  & 243.1 & ~97.0 & 156.1 &  & 169.1 & ~43.8 & ~72.4 &  & 127.8 & ~24.5 & ~40.9 & \tabularnewline
\textbf{GTR/M$\times100$} &  &  & \textbf{~71.0} & \textbf{~60.0} & \textbf{~57.0} &  & \textbf{~73.0} & \textbf{~52.0} & \textbf{~53.0} &  & \textbf{~77.0} & \textbf{~45.0} & \textbf{~50.0} & \tabularnewline
 &  &  &  &  &  &  &  &  &  &  &  &  &  & \tabularnewline
\cline{4-14} \cline{5-14} \cline{6-14} \cline{7-14} \cline{8-14} \cline{9-14} \cline{10-14} \cline{11-14} \cline{12-14} \cline{13-14} \cline{14-14} 
 &  &  & \multicolumn{11}{c}{\textsc{(iii) Gaussian mixture distribution}} & \tabularnewline
\cline{4-14} \cline{5-14} \cline{6-14} \cline{7-14} \cline{8-14} \cline{9-14} \cline{10-14} \cline{11-14} \cline{12-14} \cline{13-14} \cline{14-14} 
GTR &  &  & 178.8 & ~50.7 & ~82.7 &  & 126.5 & ~22.2 & ~38.2 &  & ~99.5 & ~10.0 & ~19.9 & \tabularnewline
Matzkin (M) &  &  & 233.3 & ~97.8 & 152.3 &  & 167.1 & ~46.0 & ~74.0 &  & 129.5 & ~28.4 & ~45.1 & \tabularnewline
\textbf{GTR/M$\times100$} &  &  & \textbf{~77.0} & \textbf{~52.0} & \textbf{~54.0} &  & \textbf{~76.0} & \textbf{~48.0} & \textbf{~52.0} &  & \textbf{~77.0} & \textbf{~35.0} & \textbf{~44.0} & \tabularnewline
 &  &  &  &  &  &  &  &  &  &  &  &  &  & \tabularnewline
\hline 
\end{tabular}
\par\end{centering}
\caption{\textsc{Design A} (trimming): $\widehat{m}(x,e)$. 
Absolute bias, variance and MSE ($\times1000$) for GTR, Matzkin estimator and QR, and related ratios $\times100$.
\label{tab:CQF_A-1}}
\end{table}

\begin{table}
\begin{centering}
\centering %
\begin{tabular}{lccccccccccccc}
\hline 
 &  & \multicolumn{3}{c}{$n=100$} &  & \multicolumn{3}{c}{$n=250$} &  & \multicolumn{3}{c}{$n=500$} & \tabularnewline
\hline 
 &  &  &  &  &  &  &  &  &  &  &  &  & \tabularnewline
 &  & |Bias| & Var. & MSE &  & |Bias| & Var. & MSE &  & |Bias| & Var. & MSE & \tabularnewline
\cline{3-13} \cline{4-13} \cline{5-13} \cline{6-13} \cline{7-13} \cline{8-13} \cline{9-13} \cline{10-13} \cline{11-13} \cline{12-13} \cline{13-13} 
GTR &  & 180.6 & ~82.0 & 114.6 &  & 157.0 & ~42.7 & ~67.3 &  & 131.2 & ~20.0 & ~37.2 & \tabularnewline
Matzkin &  & 318.6 & 157.1 & 258.6 &  & 216.7 & ~65.4 & 112.3 &  & 160.8 & ~35.5 & ~61.3 & \tabularnewline
\textbf{Ratio$\times100$} &  & \textbf{~57.0} & \textbf{~52.0} & \textbf{~44.0} &  & \textbf{~72.0} & \textbf{~65.0} & \textbf{~60.0} &  & \textbf{~82.0} & \textbf{~56.0} & \textbf{~61.0} & \tabularnewline
 &  &  &  &  &  &  &  &  &  &  &  &  & \tabularnewline
\hline 
\end{tabular}
\par\end{centering}
\caption{\textsc{Design B} (trimming): $\widehat{m}(x,e)$. 
Absolute bias, variance and MSE ($\times1000$) for GTR and Matzkin estimator, and their ratio $\times100$. 
\label{tab:CQF_B-1}}
\end{table}

\clearpage{}

\end{document}